\pdfoutput=1

\documentclass[twoside,openright,12pt]{report}
\usepackage[utf8]{inputenc}
\usepackage{geometry}
\usepackage{graphicx}
\usepackage{amsmath}
\usepackage{setspace}
\usepackage{ulem}
\usepackage{amssymb}
\usepackage{amsthm}
\usepackage{xcolor}
\usepackage{booktabs}
\usepackage[ruled, algo2e, vlined]{algorithm2e}
\usepackage[numbers]{natbib}
\usepackage{url}
\usepackage{tcolorbox}
\usepackage{pifont}
\usepackage[hidelinks]{hyperref}
\usepackage{bbm}
\usepackage{rotating}
\usepackage{tabularx}

\newcommand{\argmax}{\mathop{\mathrm{argmax\,}}}

\newcommand{\boldA}{{\boldsymbol{A}}}

\newcommand{\boldI}{{\boldsymbol{I}}}

\newcommand{\boldS}{{\boldsymbol{S}}}

\newcommand{\bolda}{{\boldsymbol{a}}}
\newcommand{\boldb}{{\boldsymbol{b}}}
\newcommand{\boldc}{{\boldsymbol{c}}}

\newcommand{\bolde}{{\boldsymbol{e}}}

\newcommand{\boldv}{{\boldsymbol{v}}}

\newcommand{\boldx}{{\boldsymbol{x}}}

\newtheorem{theorem}{Theorem}

\newtheorem{proposition}[theorem]{Proposition}

\theoremstyle{definition}
\newtheorem{definition}[theorem]{Definition}

\newcommand{\cmark}{\ding{51}} 
\newcommand{\xmark}{\ding{55}} 

\title{User-Side Realization}
\author{Ryoma Sato}
\date{}

\begin{document}

\doublespacing

\maketitle

\begin{abstract}
Users are dissatisfied with services. Since the service is not tailor-made for a user, it is natural for dissatisfaction to arise. The problem is, that even if users are dissatisfied, they often do not have the means to resolve their dissatisfaction. The user cannot alter the source code of the service, nor can they force the service provider to change. The user has no choice but to remain dissatisfied or quit the service.

User-side realization offers proactive solutions to this problem by providing general algorithms to deal with common problems on the user’s side. These algorithms run on the user's side and solve the problems without having the service provider change the service itself.

The most prominent challenge in user-side realization is a severe asymmetry between the provider and users in user-side realization. The users cannot intervene in the service directly. Users neither know the source code, architecture, and specifications the service uses. User-side realization accepts the fact that many real-world systems do not adopt user-friendly protocols, and try to realize the desired function in an inconvenient environment. Besides, a user-side algorithm should be easy and lightweight enough for end-users to use. Complicated methods involving computationally demanding processes do not work on end-user devices. Nevertheless, user-side algorithms do their best with these constraints and sometimes achieve the same or almost the same performance compared to provider-side algorithms. Such user-side algorithms offer end-users significant functionalities they would not enjoy without them. 

In this thesis, we first elaborate on the definition, requirements, and challenges of user-side realization. We introduce two general approaches for user-side realization, namely, the wrapper method and the reverse method. These approaches have their own advantages and disadvantages, which we will also discuss in the following. We then propose many user-side algorithms, including ones turning unfair recommender systems into fair ones, ones that create tailor-made search engines, and one that turns insecure translators into secure ones. All of them run on the user's side and enable end-users to fulfill their desires by themselves. We also introduce an attractive application of a user-side algorithm that collects data for a machine learning model at hand to learn by utilizing a user-side search engine. The projects described in this thesis have two prominent characteristics. First, the algorithms are general. Compared to the existing implementations specific to the target services, our algorithms abstract the conditions and interfaces so that they can be applied to many services with no or few modifications. Second, our algorithms are tested in many real-world services including Twitter (X), IMDb, and ChatGPT. These features made these projects attractive both in terms of theory and practice.

This thesis is the first systematic document for user-side realization. Although there have been sporadic and ad-hoc attempts to realize user-side algorithms, there have been no solid systems on user-side realization to the best of our knowledge. We hope this thesis will be a milestone in user-side realization and facilitate further research in this field.

\end{abstract}

\newpage
\setcounter{page}{3}

\chapter*{Acknowledgements}

I would like to express my sincere gratitude to the many people who have supported me throughout my research. I would like to thank my supervisor, Professor Hisashi Kashima, for his insightful advice. I enjoyed pleasant discussions with Professor Makoto Yamada. Professor Hidetoshi Shimodaira and Professor Akihiro Yamamoto served as committees for this thesis. They provided thought-provoking comments throughout the review process. I would like to appreciate their support, which made this thesis much better. Finally, I would like to express my deepest gratitude to my family and friends for their warm support.

\newpage
\setcounter{page}{4}

\chapter*{Publications}

This thesis is based on the following papers.

\begin{itemize}
    \item (Chapter 2) \textbf{Ryoma Sato}. Private Recommender Systems: How Can Users Build Their Own Fair Recommender Systems without Log Data? In \textit{Proceedings of the 2022 SIAM International Conference on Data Mining, SDM}, 2022.
    \item (Chapter 3) \textbf{Ryoma Sato}. Towards Principled User-side Recommender Systems. In \textit{Proceedings of the 31st ACM International Conference on Information \& Knowledge Management, CIKM}, 2022.
    \item (Chapter 4) \textbf{Ryoma Sato}. Retrieving Black-box Optimal Images from External Databases. In \textit{Proceedings of the 15th ACM International Conference on Web Search and Data Mining, WSDM}, 2022.
    \item (Chapter 5) \textbf{Ryoma Sato}. CLEAR: A Fully User-side Image Search System. In \textit{Proceedings of the 31st ACM International Conference on Information \& Knowledge Management, CIKM}, 2022.
    \item (Chapter 6) \textbf{Ryoma Sato}. Active Learning from the Web. In \textit{Proceedings of the ACM Web Conference 2023, WWW}, 2023.
    \item (Chapter 7) \textbf{Ryoma Sato}. Making Translators Privacy-aware on the User's Side. In arXiv preprint \url{https://arxiv.org/abs/2312.04068}, 2023.
\end{itemize}

\tableofcontents
\listoffigures
\listoftables

\newpage

\chapter{Introduction}

User-side realization refers to the implementation or execution of a feature on the end-user's side of the service. Users want many features. Typically, users request such features via the request form and wait for them to be implemented on the service's side. However, such requests are rarely adopted, and users get frustrated. User-side realization shows its true value in such cases. Algorithms for user-side realization enable users to enjoy such features even if the service provider ignores or refuses them.

Most widely accepted examples of user-side realization are found in browser add-ons, which enable features that the vanilla browser does not have. There are many browser add-ons that enhance the functionalities of web applications. Some of these are officially released by the service, and others are developed by fan users. The latter cases are typical examples of user-side realization.

However, most of these attempts are sporadic and ad-hoc, and no systematic research on user-side realization has been made so far. In this thesis, we elaborate on the requirements, challenges, and general solutions for user-side realization for the first time.

In this thesis, a service refers to web services, such as Facebook, X (Twitter), Instagram, and Flickr, in most cases, but the theory and algorithms are not restricted to these cases, and they are applicable to offline services as well. Users refer to end-users of these services. Such users do not have any privileged access to the resources of the service, including the source code, database, specification, and internal documents. User-side realization must be done with such limited access to resources. What to realize depends on the case, and we do not limit the scope. In a typical example, the service is X (Twitter), and the feature we want to realize is a fair recommender system. Although many algorithms to realize fair recommender systems have been proposed \cite{kamishima2012enhancement, yao2017beyond, biega2018equity, milano2020recommender}, most of them are designed for service providers, and it is not obvious how end-users, who do not have access to the database, realize a fair recommender system on their side. We will provide the solutions in Chapters 2 and 3.

\section{Requirements and Challenges}

There are several requirements for user-side realization, which introduce challenges.

\subsection{Limited Resources}

end-users have very limited access to the internal states of the service. The source code of the service is not publicly available in most cases because it is an industrial secret. Users do not have direct access to the database. This is a significant challenge for creating recommender systems and search engines. All of the existing algorithms for recommender systems and search engines assume the developer has full access to the database. However, this basic assumption does not hold on the user's side. Algorithms for user-side realization must work with such limited information and resources. In general, user-side algorithms have disadvantages over provider-side algorithms because there are many things end-users and user-side algorithms cannot access or alter. Conversely, service providers can do all the things users can do. Therefore, the scope and performance of provider-side algorithms are the upper bound of those of user-side algorithms. The problem is how close user-side algorithms can reach to the provider-side algorithm.

\subsection{Generality}

Many browser add-ons are designed for specific services.  Such services have significant industrial values. However, such systems lack generality and are of little importance as research. More valuable algorithms for user-side realization abstract the conditions and interfaces so that they can be applied to many services with no or few modifications. The algorithms we propose in this thesis meet this requirement. The algorithm for user-side realization of fair recommender systems is applicable to not only X (Twitter) but also to other services, including ones that will appear in the future, only if the assumptions of the algorithms are met. This requirement introduces challenges of elucidating the conditions so that the scope of the algorithms is as wide as possible while the performance is retained.

\subsection{Lightweight} \label{sec: lightweight}

Algorithms for user-side realization run on user's devices, while provider-side algorithms typically run on commercial servers. Thus, user-side algorithms should be lightweight. This further imposes restrictions on user-side algorithms. Besides, user-side algorithms should not involve complex configurations. Taking much time for such configurations may be amortized and pay on the provider's side, but this is not the case on the user's side. Computationally demanding training processes, which are common in algorithms for many recommender systems and image search engines, are also prohibited for the same reason.

\section{Approaches}

\begin{figure*}[p]
\begin{center}
\includegraphics[width=\hsize]{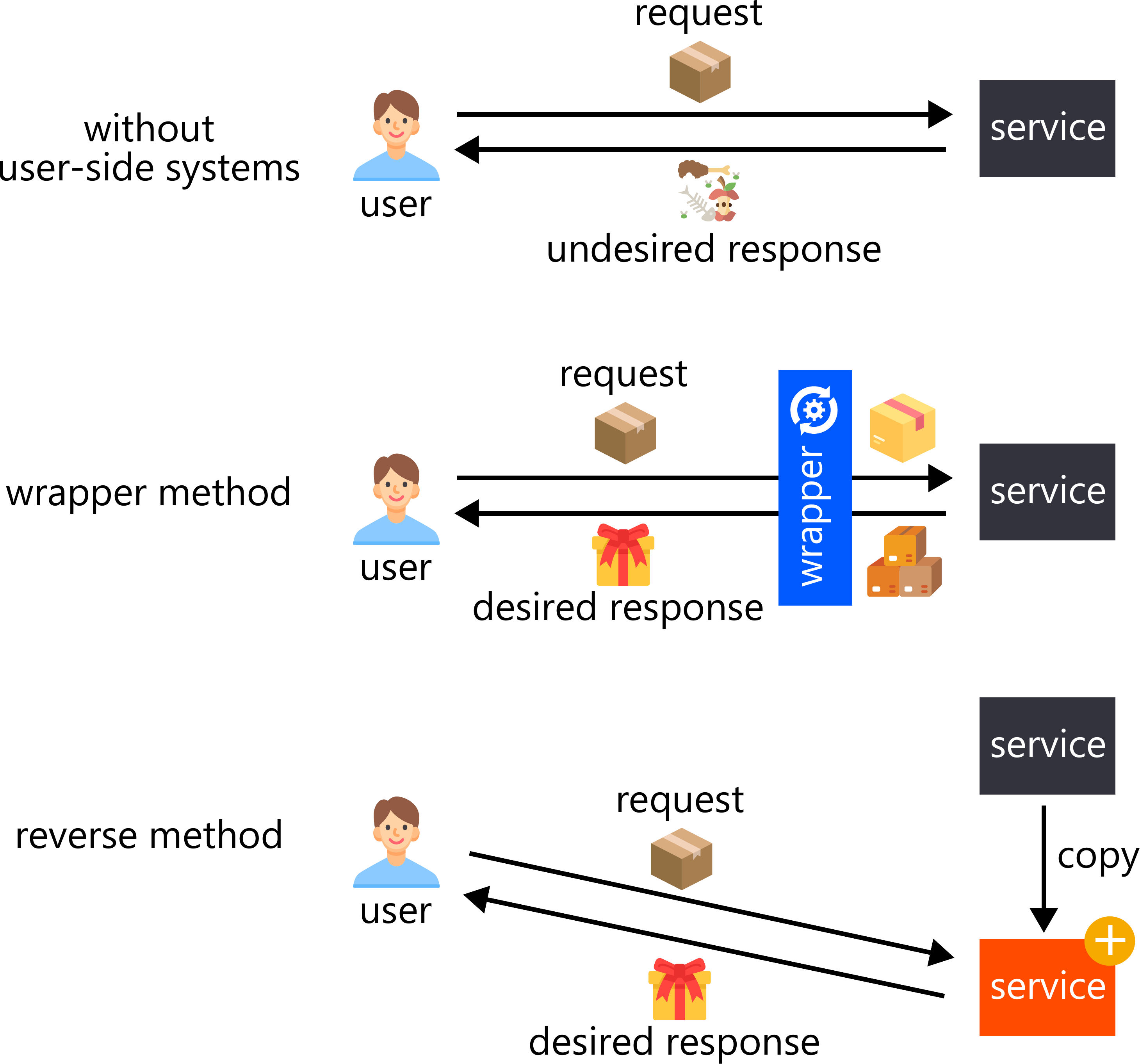}
\end{center}
\caption{Approaches of user-side realization.}
\label{fig: approaches}
\end{figure*}

We propose two general approaches for user-side realization (Figure \ref{fig: approaches}).

\subsection{Wrapper Method}

A wrapper method wraps the official functionality and adds features. It intervenes in the input, output, or both of them. A simple approach is post-processing, which modifies the output of the official service so that the user's requirement is automatically fulfilled. A more advanced method captures the input the user feeds, inputs proxy values to the service, and alters the output of the service accordingly. Typically, these functions are automatically and seamlessly realized, and the user can use the enhanced service as if such functionalities were implemented on the service side. An advantage of the wrapper approach is it tends to lead to simple algorithms. This is a huge plus for user-side realization because it should be lightweight as we discussed in Section \ref{sec: lightweight}, and the budget for user-side systems tends to be lower than provider-site systems. A disadvantage is that it tends to lead to ad-hoc algorithms. The resulting algorithm depends on the problem setting, and the details are largely deferred to algorithm designers.

\subsection{Reverse Method}

An reverse method ``reverse engineers'' the service and adds functionalities to the copied service. The most significant challenge of user-side realization is that end-users have limited access to the internal states of the service. An reverse method resolves this problem. Once the service is copied, it is under the control of the user, and the user has full access to the (copied) service and can apply ordinary provider-side algorithms to it. Note that it is rare to copy the entire service in practice, and a part of the service is inferred with some uncertainties. An advantage of the reverse method is it can be applied to any problem setting. Differences in problem setting are absorbed in the copying process and the method can be combined with any algorithm in the part of adding functionality. A disadvantage is that it tends to lead to computationally demanding and complex algorithms. Therefore, this approach is suitable for tasks that are difficult to accomplish or that are willing to pay a lot of cost to accomplish.

\section{Related Work and Concepts}

There have been many research topics related to user-side realization. We provide an overview of them. More detailed discussions are found in the related work sections in the following chapters.

\subsection{Dynamic and Client-Side Search Engines}

Decades ago, in the dawning of search engines, dynamic and client-side search engines \cite{bra1994information,chen1998smart,menczer2000adaptive} were studied to complement the static nature of index-based search engines. Some parts of the index of the traditional search engine may be outdated, and the crawlers of general-purpose search engines may not go as deep as some users want. Dynamic search engines have crawlers search the Web as soon as the user queries, and the crawlers retrieve relevant pages on the fly. The retrieved pages are fresh, and the crawlers go as deep as the user wants. The drawback of such systems is that it takes too much time, typically minutes, to complete a search. Some acceleration techniques such as cache \cite{bra1994information} and reinforcement learning \cite{menczer2000adaptive} are utilized in such systems to mitigate the problem. These systems can be seen as examples of user-side realization in a broad sense as they can be implemented and run on the user's side, but they are different from pure user-side realization because there is no clear distinction between providers and users.

\subsection{Decentralized and Modularized Systems}

\sloppy Decentralization systems, including peer-to-peer content distribution \cite{theotokis2004survey} and blockchains \cite{nakamoto2008bitcoin}, have been actively studied and utilized in many applications, such as Napster, Bitcoin, and Ethereum. Users of such systems can not only enjoy the system but also add functionalities, e.g., by implementing smart contracts \cite{buterin2013ethereum}. If not decentralized, some systems are composed of modules \cite{tanaka1996intelligent}, and users can add functionalities by implementing modules or add-ons. These systems can also be seen as examples of user-side realization in a broad sense, but they are different from pure user-side realization in two senses. First, there there is no clear distinction between providers and users. Second, these systems assume that all participants, including providers and users, follow the same (decentralized) protocol. This is not the case in many real-world systems. There is a severe asymmetry between the provider and users in user-side realization, and the users cannot intervene in the service directly. User-side realization rather accepts the fact that many real-world systems do not adopt user-friendly protocols, and try to realize the desired function in an inconvenient environment.

\subsection{Oracle Machines and Query Complexity}

In theoretical computer science, oracle machines have been studied to investigate the computational complexity and computability \cite{baker1975relativizations}. Oracles have also been utilized in optimization\cite{orlin2009faster,bertsimas2004solving}. This approach assumes a machine has access to some oracles and considers how many times the machine needs to query the oracles to solve a given problem. This approach can be seen as user-side realization by regarding the oracle as a service and the machine as a user. We also consider how many times a user needs to query the service to realize the desired function. Oracle machines provide the abstracted principles of user-side realization. The difference between these research realms and user-side realization lies in the fact that user-side realization is applied to real-world services such as Twitter (X) and ChatGPT.

\subsection{User-Side Realization}

The research of modern user-side realization has been initiated by Sato \cite{sato2022private} in 2021, and many pieces of research on user-side realization have been carried out, including recommender systems \cite{sato2022principled,shinden2023device} and search systems \cite{sato2021retrieving, sato2022clear}. WebGPT \cite{nakano2021webgpt} is a modern dynamic search engine utilizing a large language model, and can be seen as user-side realization to the same extent that a classical dynamic search engine is considered user-side realization. WebShop~\cite{yao2022webshop} enables automated shopping in ordinary e-commerce sites on the user's side by using an agent driven by a large language model. Shinden \cite{shinden2023device} proposed privacy-aware recommender systems running on users' devices by aggregating the recommendations to the representative accounts. WebArena~\cite{zhou2023webarena} is a general environment to test agents realizing rich functionalities on the user's side. EasyMark~\cite{sato2023embarassingly} realizes large language models with text watermarks on the user's side. Overall, the research field of user-side realization is growing. We hope this thesis becomes a milestone in user-side realization and facilitates further research in this field. 

\section{Organization of This Thesis}

In this chapter, we have introduced the requirements, challenges, and approaches of user-side realization.

In Chapter 2, we propose \textsc{PrivateRank} and \textsc{PrivateWalk} which realize recommender systems on the user's side. These algorithms can be used to realize many functionalities. For example, they can realize fair, diverse, and informative recommender systems even when the official recommender systems are not. The challenge is that the user does not have the log data of other users, which is the core assumption of standard collaborative filtering methods. \textsc{PrivateRank} and \textsc{PrivateWalk} solve this problem by utilizing the official recommender systems, i.e., they adopt the wrapper approach. In the experiments, we carry out case studies in the wild, in which we run our proposed method in real-world services, namely IMDb and Twitter (X), as well as offline evaluation on popular benchmarks. The case studies are conducted on the system in operation on the Internet. We emphasize that we do not have privileged access to IMDb and Twitter and have the same information as end-users. Nevertheless, \textsc{PrivateWalk} succeeds in building fair recommender systems while the official ones are not fair. These results demonstrate that the challenging problem of user-side realization of recommender systems is indeed feasible in the real-world environment.

In Chapter 3, we propose more principled methods for user-side recommender systems. This first approach is the first reverse method in the literature. We show that the inner states of the official recommender systems can be inferred only from the information users can see. End-users can realize many functionalities of recommender systems once the state is inferred. The drawback of this method is efficiency. To overcome this issue, we summarize the desiderata of user-side recommender systems and propose \textsc{Consul}, a user-side algorithm that meets all of these desiderata, while other methods including \textsc{PrivateRank} and \textsc{PrivateWalk} do not. \textsc{Consul} has preferable theoretical properties. Besides, in the experiments, we confirm that \textsc{Consul} is more effective than other user-side algorithms and almost matches the performance of the provider-side algorithm. This is remarkable as the information \textsc{Consul} can access is limited and has many disadvantages against the provider-side algorithm. We also carry out a user study involving many crowd workers and show \textsc{Consul} is preferable.

In Chapter 4, we propose \textsc{Tiara}, a user-side algorithm for realizing customized search engines. Many services are equipped with search engines, but most of them provide only vanilla functionalities and dissatisfy the users. \textsc{Tiara} realizes search engines with any score functions users want to use. It can also realize search engines that accept any form of query other than the official search engine supports. The challenge here is again users do not have access to the database, which is a crucial assumption of any search algorithm. \textsc{Tiara} realizes user-side search engines by wrapping the official engine. In the experiments, we carry out in-the-wild evaluations with the online Flickr environment as well as evaluations with popular benchmarks. Flickr supports searching images based on texts and tags but does not support image-to-image search. We demonstrate that \textsc{Tiara} realizes image-to-image search on the user's side.

In Chapter 5, we propose CLEAR, a lightweight algorithm for user-side search engines. CLEAR is so lightweight that it can realize real-time search while \textsc{Tiara} takes tens of seconds to minutes. Nevertheless, the quality of CLEAR matches that of \textsc{Tiara}. We built a system demonstration of CLEAR and any user can enjoy image-to-image search on Flickr. This system has additional privacy merit. As the query is processed completely on the client side, Flickr and we (i.e., the developer of this system) do not receive the original query.

In Chapter 6, we propose a non-trivial and attractive application of user-side search engines, Seafaring, which automatically collects effective data for machine learning from the Web. Seafaring uses techniques of pool-based active learning and user-side search engines. Pool-based active learning selects effective data from a pool of unlabelled data. A pool typically contains thousands to millions of data in the existing literature. Seafaring regards the Web as a huge pool of unlabelled data, which contains billions of data, and thereby realizes much more effective active learning than the existing methods. The main challenge is that the Web is too large to find effective data that meets the criterion of active learning. Seafaring solves this problem by a user-side search algorithm, namely Tiara. Tiara can search items based on any score function. Although normal search engines do not search items based on the criterion of active learning, Tiara can, and it collects many data that have high scores based on the criterion of active learning. We confirm that Seafaring provides a significant boost of performance for training machine learning models. This illustrates that the scope of applications of user-side realization is broad.

In Chapter 7, we propose PRISM, a user-side privacy protection for translation systems. Although several machine translation engines claim to prioritize privacy, the extent and specifics of such protection are largely ambiguous. First, there is often a lack of clarity on how and to what degree the data is protected. Even if the service providers believe they have sufficient safeguards in place, sophisticated adversaries might still extract sensitive information. Second, vulnerabilities may exist outside of these protective measures, such as within communication channels, potentially leading to data leakage. As a result, users are hesitant to utilize machine translation engines for data demanding high levels of privacy protection, thereby missing their benefits. PRISM resolves this problem. Instead of relying on the translation service to keep data safe, PRISM provides the means to protect data on the user's side. This approach ensures that even machine translation engines with inadequate privacy measures can be used securely. For platforms already equipped with privacy safeguards, PRISM acts as an additional protection layer, reinforcing their security further. PRISM adds these privacy features without significantly compromising translation accuracy. In the experiments, we demonstrate the effectiveness of PRISM using real-world translators, T5 and ChatGPT. PRISM effectively balances privacy protection with translation accuracy. This also illustrates that the scope of applications of user-side realization is not limited to information retrieval, but the framework is general and the application is broad.

Chapter 8 concludes the thesis and discusses future directions of user-side realization.

\chapter{Private Recommender Systems: How Can Users Build Their Own Fair Recommender Systems without Log Data?}


\section{Introduction}

Fair recommender systems have attracted much attention owing to their importance in society \citep{milano2020recommender}. A typical application is in the job market \citep{geyik2019fairness, feldman2015certifying, geyik2018building}.
The disparate impact theory prohibits a recruiting process that has an adverse impact on a protected group, even if the process appears neutral on its face.
Therefore, job recruiters must avoid using unfair talent recommender systems to remove (possibly unintended) biases in their recruiting process.

However, even if users of a service want to adopt fair recommender systems, they cannot utilize them if the service does not provide them. There are several difficulties facing the adoption of fair recommender systems. First, commercial services may be reluctant to implement fair systems because fairness and effectiveness are in a trade-off relation \citep{mehrotra2018towards, corbett2017algorithmic}, and fairness-aware systems are expensive for implementation and maintenance. Although some social networking services, such as LinkedIn \citep{geyik2019fairness, geyik2018building}, provide fair account recommendations, those of other services, such as Twitter and Facebook, are not necessarily fair with respect to gender or race.
Second, the fairness criterion required by a user differs from user to user. The fairness defined by the service may not match the criteria users call for. For example, even if the recommender system is fair with respect to gender, some users may require fairness with respect to race or the combination of gender and race. In general, different fairness criteria are required in different societies.
Third, companies may refuse to disclose the algorithms they use. This makes it difficult for users to assess the fairness of the system. Milano et al. \citep{milano2020recommender} pointed out that ``The details of RS currently in operation are treated as highly guarded industrial secrets. This makes it difficult for independent researchers to access information about their internal operations, and hence provide any evidence-based assessment.'' In summary, it is difficult for users eager to enjoy fair systems to ensure fairness if they use a recommender system provided by a service.

In this thesis, we propose a framework where each user builds their own fair recommender system by themselves. Such a system can provide recommendations in a fair manner each user calls for. In this framework, a user uses their own recommender system via a browser add-on instead of the recommender system provided by the service. We call this personal recommender system a \emph{private recommender system}.

In this work, we focus on item-to-item collaborative filtering, where a set of related items are recommended when a user visits an item page. Each item has a discrete sensitive attribute, such as gender and race. A recommender system must treat all sensitive groups equally.
Examples of this setting include:
\begin{itemize}
    \item \textbf{Job recruiting.} Here, a user is a recruiter, an item is a job seeker, and each job seeker has a sensitive attribute, such as gender and race.
    \item \textbf{Breaking the filter bubble \citep{pariser2011filter}.} Some recommender systems filter information too aggressively. For example, a news recommender system may recommend only conservative news to conservative users \citep{pariser2011filter}. Some users may want to receive unbiased recommendations with respect to ideology. In this case, the user is a reader, the item is a news article, and each news has an ideology label as the sensitive attribute.
    \item \textbf{Avoiding popular item bias.} Recommender systems tend to recommend popular items too much \citep{xiao2019beyond, mehrotra2018towards}. Some knowledgeable users need not receive ordinary items, and they may want to receive uncommon items. For example, IMDb recommends Forrest Gump, The Dark Knight, The Godfather, Inception, Pulp Fiction, and Fight Club in the Shawshank Redemption page\footnote{\url{https://www.imdb.com/title/tt0111161/}}. However, most cinema fans are already familiar with all of these titles, and these recommendations are not informative to them. In this case, we can set the sensitive attribute to be a popularity label (e.g., high, middle, and low, based on the number of reviews received), such that recommender systems must recommend uncommon (but related) items as well. Although readers may think of this problem as topic diversification \citep{ziegler2005improving}, we discuss these problems in a unified framework.
\end{itemize}
We assume that users have access to the sensitive attribute, but they do not necessarily have access to the content of items \citep{koren2009matrix}. Although there exist several methods for building fair recommender systems \citep{singh2018fairness, patro2020fairrec, liu2019personalized}, all of the existing fair recommendation methods are designed for service providers, who can access the entire log data, such as the rating matrix or browsing history of all users. In our setting, however, a user does not have access to the log data of other users nor latent representations of items. A clear distinction between this work and previous works is that our setting prohibits accessing such log data. This restriction makes the problem challenging. The key idea is that a user has access to unfair recommendations shown by the service provider. We propose methods to leverage the outputs of an unfair recommender system to construct a fair recommender system. We conduct several experiments and show that our proposed method can build much fairer recommender systems than provider recommender systems while keeping their performance. The contribution of this chapter is as follows:
\begin{itemize}
    \item We propose private recommender systems, where each user builds their own recommender system to ensure the fairness of recommendations. Private recommender systems enable fair recommendations in many situations where conventional recommendation algorithms cannot be deployed. Our proposed framework expands the application scope of fairness-aware recommender systems.
    \item We propose methods to develop private recommender systems without accessing log data or the contents of items. Although our methods are simple, they exhibit a positive trade-off between fairness and performance, even without accessing log data.
    \item We confirm that our proposed method works in real-world scenarios via qualitative case studies on IMDb and Twitter.
\end{itemize}

\noindent \uline{\textbf{Reproducibility:}} Our code is available at \url{https://github.com/joisino/private-recsys}.

\section{Notations}

For every positive integer $n \in \mathbb{Z}_+$, $[n]$ denotes the set $\{ 1, 2, \dots n \}$.
Let $K \in \mathbb{Z}_+$ be the length of a recommendation list. Let $\mathcal{U} = [m]$ denote the set of users and $\mathcal{I} = [n]$ denote the set of items, where $m$ and $n$ are the numbers of users and items, respectively. Without loss of generality, we assume that the users and items are numbered with $1, \dots, m$ and $1, \dots, n$, respectively.

\section{Problem Setting} \label{sec: setting-private}

We focus on item-to-item collaborative filtering. In this setting, when user $u$ accesses the page associated with item $i$, a recommender system aims to recommend a set of items that are relevant to item $i$. The recommendation may be personalized for user $u$ or solely determined by the currently displayed item $i$. We assume that the recommendation list does not contain any items that user $u$ has already interacted with. This is natural because already known items are not informative. Formally, a recommendation list is represented by a $K$-tuple of items, and a recommender system is represented by a function $\mathcal{F}\colon \mathcal{U} \times \mathcal{I} \to \mathcal{I}^K$ that takes a user and source item and returns a recommendation list when the user visits the source item. $\mathcal{F}(u, i)_k$ denotes the $k$-th item of $\mathcal{F}(u, i)$.

We assume that each item $i$ has a sensitive attribute $a_i \in \mathcal{A}$, and we can observe the sensitive attribute. $\mathcal{A}$ is a discrete set of sensitive groups. When we cannot directly observe $a_i$, we estimate it from auxiliary information. A recommender list is fair if the proportion of protected attributes in the list is approximately uniform. For example, ({\color[HTML]{FF4B00} man}, {\color[HTML]{FF4B00} man}, {\color[HTML]{FF4B00} man}, {\color[HTML]{FF4B00} man}, {\color[HTML]{005AFF} woman}, {\color[HTML]{FF4B00} man}) is not fair with respect to $\mathcal{A} = \{\text{\color[HTML]{FF4B00} man}, \text{\color[HTML]{005AFF} woman} \}$ because much more men accounts (items) are recommended than women accounts. We found that LinkedIn employs this kind of fairness in their recommendations. Specifically, at least two men and two women users were recommended out of five recommendation slots in a user's profile page as far as we observed. However, we observed that Twitter does not employ this kind of fairness. An example of our goal in this chapter is to develop a fair recommender system on Twitter, although they do not provide a fair recommender system.
In the experiments, we measure fairness quantitatively by least ratio (i.e., the minimum ratio of protected attributes in the list) and entropy. For example, least\_ratio({\color[HTML]{FF4B00} man}, {\color[HTML]{FF4B00} man}, {\color[HTML]{FF4B00} man}, {\color[HTML]{FF4B00} man}, {\color[HTML]{005AFF} woman}, {\color[HTML]{FF4B00} man}) $= 1/6$, and the entropy is $-1/6 \log 1/6 - 5/6 \log 5/6 \approx 0.65$. The higher these values are, the fairer the list is considered to be. In the LinkedIn example above, the least ratio was ensured to be at least $2/5 = 0.4$. If the numbers $n_a$ of men and women are the same, \begin{align*}
    &\text{least\_ratio} = 0.5 \\
    &\Longleftrightarrow \frac{K}{2} \text{ men and } \frac{K}{2} \text{ women are recommended} \\
    &\Longleftrightarrow \text{Pr}[Y \mid A = \text{man}] = \frac{K}{2n_a} = \text{Pr}[Y \mid A = \text{woman}] \\
    &\Longleftrightarrow \text{demographic parity}
\end{align*} holds, where $Y$ is a binary outcome variable that indicates whether an item is included in the recomendation list. When the numbers of men and women are different, we can slightly modify the definition and algorithm so that it includes demographic parity, yet we stick to the above definition in this chapter for ease of exposition. Note that these fairness scores do not take the order of recommendation lists into account, but only count the number of items. It is straightforward to incorporate order-aware fairness in our proposed methods by adjusting the post processing process (e.g., adopting \citep{zehlike2017fair}). We leave this direction for future work for ease of exposition.

We assume a recommender system $\mathcal{P}_{\text{prov}}$ in operation provided by a service is available. We call this recommender system a \emph{provider recommender system}. The provider recommender system is arbitrary and may utilize the contents of items by accessing exclusive catalog databases, as well as log data of all users. Although the provider recommender system suggests highly related items, it may be unfair to the protected group. Our goal is to enable each user to build their own fair recommender system. We call this new recommender system a private recommender system. Without loss of generality, we assume that the user who is constructing a private recommender system is user $1$. We assume that user $1$ does not have access to the details of the provider recommender system, including the algorithm, latent representations of items, and score function. The only information that user $1$ can obtain is the top-$K$ ranking of items when user $1$ visits each source item, which is the case in many real-world settings. Formally, let $\mathcal{P}_{\text{prov}, 1}\colon \mathcal{I} \to \mathcal{I}^K$ be the partial application of $\mathcal{P}_{\text{prov}}$ by fixing the first argument to be user $1$, i.e., $\mathcal{P}_{\text{prov}, 1}(i) = \mathcal{P}_{\text{prov}}(1, i) ~\forall i \in \mathcal{I}$. We assume that user $1$ can access function $\mathcal{P}_{\text{prov}, 1}$ by using the provider recommender system. The problem we tackle in this chapter is formalized as follows:

\vspace{0.05in}
\noindent \uline{\textbf{Private Recommender System Problem.}}\\
\textbf{Given:} Oracle access to the provider recommendations $\mathcal{P}_{\text{prov}, 1}$. Sensitive attribute $a_i$ of each item $i \in \mathcal{I}$.\\
\textbf{Output:} A private recommender system $\mathcal{Q}\colon \mathcal{I} \to \mathcal{I}^K$ that is fair with respect to $\mathcal{A}$.

\vspace{0.05in}
\noindent \uline{\textbf{Difference with posthoc processing.}} Note that this problem setting is different from posthoc fair recommendations because posthoc processing has access to the full ranking of items or scoring function, whereas we only have access to top-$K$ items presented by the service provider. Although it may be possible to apply posthoc processing to the top-$K$ list, it just alters the order of the list. Therefore, simply applying posthoc processing cannot improve the least ratio. For example, if the top-$K$ list contains only one group, posthoc processing does not recommend items in the other group. If $K$ is sufficiently large, reordering the list by a posthoc processing method and truncating the list may improve the least ratio. However, $K$ is typically small, and we focus on the challenging cases, where recommendation lists are short and may solely contain one group.

\section{Method}

The core idea of our proposed methods is regarding two items $(i, j)$ similar if item $j$ is recommended in item $i$'s page. They recommend items based on this similarity measure while maintaining fair recommendations. However, the similarities of only $K$ pairs are defined for each source item in this manner. The recommendation list may solely contain one protected group. In that case, we cannot retrieve similar items in different groups. To address this challenge, we propose to utilize the recommendation network to define a similarity measure.

\subsection{Recommendation Network}

Recommendation networks \citep{cano2006topology, celma2008new, seyerlehner2009limitation, seyerlehner2009browsing} have been utilized to investigate the property of recommender systems, such as the navigation of a recommender system and novelty of recommendations. The node set of a recommendation network is the set of items, i.e., $V = \mathcal{I}$, and a directed edge $(i, j) \in \mathcal{I} \times \mathcal{I}$ indicates that item $j$ is recommended in item $i$'s page, i.e., $E = \{ (i, j) \in \mathcal{I} \times \mathcal{I} \mid \exists k \text{ s.t. } \mathcal{P}_{\text{prov}, 1}(i)_k = j \}$. The key idea of our proposed methods is that two items can be considered to be similar if they are close in the recommender network. We use this to define the similarities between items. To the best of our knowledge, this is the first work to adopt the recommendation network for defining the similarities between items for constructing recommender systems.

\subsection{PrivateRank} \label{sec: PrivateRank}

In this section, we introduce \textsc{PrivateRank}. First, \textsc{PrivateRank} constructs a recommender network by querying the provider recommender system. We adopt a weighted graph to incorporate the information of item rankings. Inspired by the discounted cumulative gain, the weight function is inverse-logarithmically discounted. The adjacency matrix $\boldA$ is defined as:
\begin{align*}
    \boldA_{ij} = \begin{cases}
    \frac{1}{\log(k + 1)} & (\mathcal{P}_{\text{prov}, 1}(i)_k = j) \\
    0 & \text{(otherwise)}
    \end{cases}
\end{align*}
\textsc{PrivateRank} employs the personalized PageRank \citep{jeh2003scaling, page1999pagerank}, also called the random walk with restart, which is a classic yet powerful similarity measures between nodes on a graph. The personalized PageRank $\boldS_i \in \mathbb{R}^{n}$ of node $i$ measures similarities between node $i$ and other nodes. The personalized PageRank assumes a random surfer who follows a link incident to the current node with probability proportional to its weight or jumps to node $i$ with probability $1 - c$. $c > 0$ is a hyperparameter called a damping factor. The damping factor $c$ controls the spread of random walks. When $c$ is large, the surfer rarely jumps back to the source node, and it solely captures global structures. Therefore, a small $c$ is appropriate in capturing local structures around the source node. We empirically demonstrate it in the experiments. The personalized PageRank of node $j$ with respect to node $i$ is defined as the probability that the random surfer will arrive at node $j$. Formally,
\begin{align*}
    \boldS_{i} = c \tilde{\boldA}^\top \boldS_{i} + (1 - c) \bolde^{(i)}
\end{align*}
\noindent where $\tilde{\boldA}$ is the row-wise normalized adjacency matrix, i.e., $\tilde{\boldA}_i = \boldA_i / \sum_j \boldA_{ij}$, and $\bolde^{(i)}$ is the $i$-th standard basis.
We compute the personalized PageRank using the cumulative power iteration \citep{yoon2018tpa}:
\begin{align*}
    \hat{\boldS}_{i} = (1 - c) \sum_{k = 0}^L (c \tilde{\boldA}^\top)^k \bolde^{(i)},
\end{align*}
where $L$ is a hyperparameter. $L$ determines the trade-off between time consumption and accuracy. We find that a small $L$ is sufficient because recommendation networks are typically small-world. We use $L = 10$ in the experiments.

After we obtain the similarity matrix, \textsc{PrivateRank} ranks items in a fair manner. There are several existing fairness-aware ranking methods, including the optimization-based \citep{singh2018fairness}, learning-based \citep{beutel2019fairness}, and post processing-based approaches \citep{geyik2019fairness, liu2019personalized, zehlike2017fair}. We employ the post processing-based approach similar to \citep{geyik2019fairness, liu2019personalized}. We set the minimum number $\tau ~(0 \le \tau \le K/|\mathcal{A}|)$ of items of each group as a hyperparameter. \textsc{PrivateRank} greedily takes items as long as the constraint can be satisfied. Specifically, let $r$ be the number of items to be taken, and $c_a$ be the number of items in the list with protected attribute $a$. Then, if $\sum_{a \in \mathcal{A}} \max(0, \tau - c_a) \le r$ holds, we can satisfy the minimum requirement by completing the list. The pseudo code of \textsc{PrivateRank} is available in Section \ref{sec: PrivateRankcode} in the appendices. 

\textsc{PrivateRank} holds the following preferable properties. First, it ensures fairness if we increase the minimum requirement $\tau$.

\begin{theorem} \label{thm1}
If $\mathcal{I}$ contains at least $\tau$ items of each sensitive attribute, the least ratio of recommendation list generated by \textsc{PrivateRank} is at least $\tau/K$.
\end{theorem}

The proof is available in Section \ref{sec: thm1} in the appendices.

Second, \textsc{PrivateRank} does not lose performance when $\tau = 0$.

\begin{theorem}
If we set $c < \frac{1}{(K + 1)^2 \log^2(K + 1)}$, $L \ge 1$, and $\tau = 0$, the recommendation list generated by \textsc{PrivateRank} is the same as that of the provider recommender system. Therefore, the recall and nDCG of \textsc{PrivateRank} are the same as those of the provider recommender system.
\end{theorem}

The proof is available in Section \ref{sec: thm2} in the appendices.

This result is consistent with the intuition that small $c$ is good for \textsc{PrivateRank}. These theorems indicate that \textsc{PrivateRank} properly enhances the functionality of the provider recommender system. Specifically, it can recover the original system when $\tau = 0$, and in addition to that, it can control fairness by increasing $\tau$.

\vspace{0.1in}
\noindent \uline{\textbf{Time complexity:}} We analyze the time complexity of \textsc{PrivateRank}. Constructing the recommendation network issues $K$ queries to $\mathcal{P}_{\text{prov}, 1}$ for each item. Therefore, $Kn$ queries are issued in total. Computing an estimate $\hat{\boldS}_i$ of the personalized PageRank involves $L$ vector-matrix multiplications. A vector-matrix multiplication can be done in $O(K n)$ time because $Kn$ elements of $\tilde{\boldA}$ are non-zero. Therefore, computing $\hat{\boldS}_i$ runs in $O(KLn)$ time. The postprocessing runs in $O(K + |\mathcal{A}|)$ time. Hence, constructing the recommendation list takes $O(n (K + |\mathcal{A}|))$ time. In total, \textsc{PrivateRank} runs in $O(n (KL + |\mathcal{A}|))$ time if we assume that evaluating $\mathcal{P}_{\text{prov}, 1}$ runs in a constant time.

\subsection{PrivateWalk}

Although \textsc{PrivateRank} performs well in practice, the main limitation of this method is its scalability. Even if we use a faster approximation method for computing the personalized PageRank, constructing the recommendation network may be a bottleneck of the computation, which requires $Kn$ evaluations of $\mathcal{P}_{\text{prov}, 1}$. Many evaluations of $\mathcal{P}_{\text{prov}, 1}$ may hit the limitation of API, consume too much wall clock time for inserting appropriate intervals, or may be certified as a DOS attack. Therefore, batch methods that construct a full recommendation network are not suitable when many items are involved. We propose an algorithm to build a private recommender system that computes a recommendation list on demand when a user accesses an item page.

The central idea of \textsc{PrivateWalk} is common with that of \textsc{PrivateRank}: two items are similar if they are close in the recommendation networks of the provider recommender system. \textsc{PrivateWalk} utilizes random walks. Two items $(i, j)$ are considered similar if item $j$ can be reached from $i$ in short steps by a random walk. In contrast to \textsc{PrivateRank}, \textsc{PrivateWalk} runs random walks like a web crawler when a user visits a page, rather than building the recommendation networks beforehand. To achieve fairness, \textsc{PrivateWalk} also employs the post processing approach with the minimum requirement $\tau$. The pseudo code of \textsc{PrivateWalk} is available in Section \ref{sec: PrivateWalkcode} in the appendices.

\vspace{0.1in}
\noindent \uline{\textbf{Time complexity:}} The time complexity of \textsc{PrivateWalk} depends on the average length $L_{\text{ave}}$ of random walks, which is bounded by $L_{\text{max}}$. \textsc{PrivateWalk} runs the inner loop (Lines 7--14 in the pseudo code) $K L_{\text{ave}}$ time. In each loop, $\mathcal{P}_{\text{prov}, 1}$ is evaluated once, \texttt{CanAdd} is evaluated once, and $O(K)$ basic operations run. The number of loops of the fallback process (Lines 15--19) is a constant in expectation because the probability decays exponentially with respect to the number of iterations. In total, \textsc{PrivateWalk} runs in $O(K (K + |\mathcal{A}|) L_{\text{ave}})$ time if we assume evaluating $\mathcal{P}_{\text{prov}, 1}$ runs in constant time. The complexity is independent of the number $n$ of items in contrast to \textsc{PrivateRank}.

\section{Experiments}

We will answer the following questions via experiments.

\begin{itemize}
    \item (RQ1) How good trade-off between fairness and performance do our proposed methods strike?
    \item (RQ2) How sensitive are our proposed methods with respect to hyperparameters?
    \item (RQ3) Do our methods work in real-world scenarios?
\end{itemize}

\subsection{Experimental settings}

\begin{figure*}[tb]
\begin{center}
\begin{minipage}{0.21\hsize}
\includegraphics[width=\hsize]{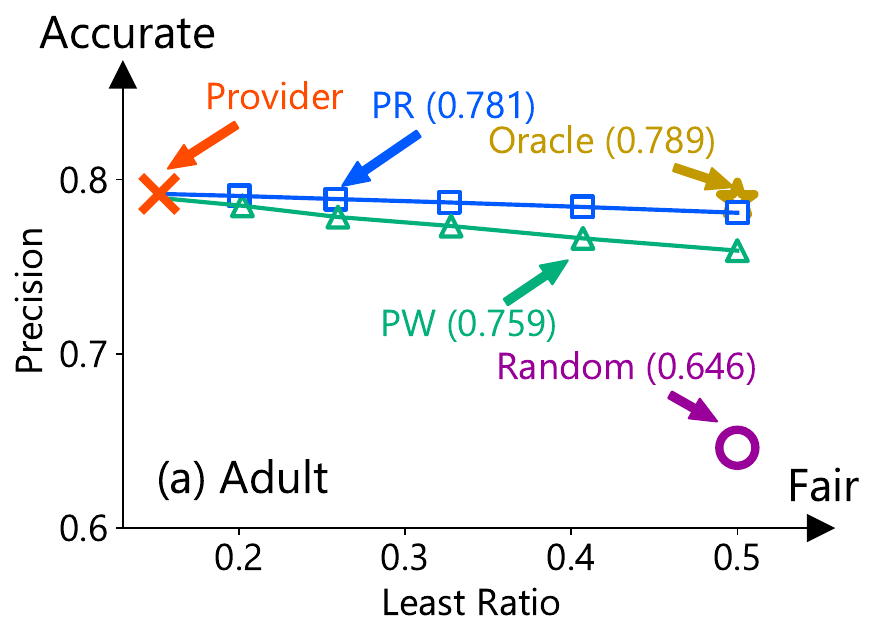}
\end{minipage}
\quad
\begin{minipage}{0.21\hsize}
\includegraphics[width=\hsize]{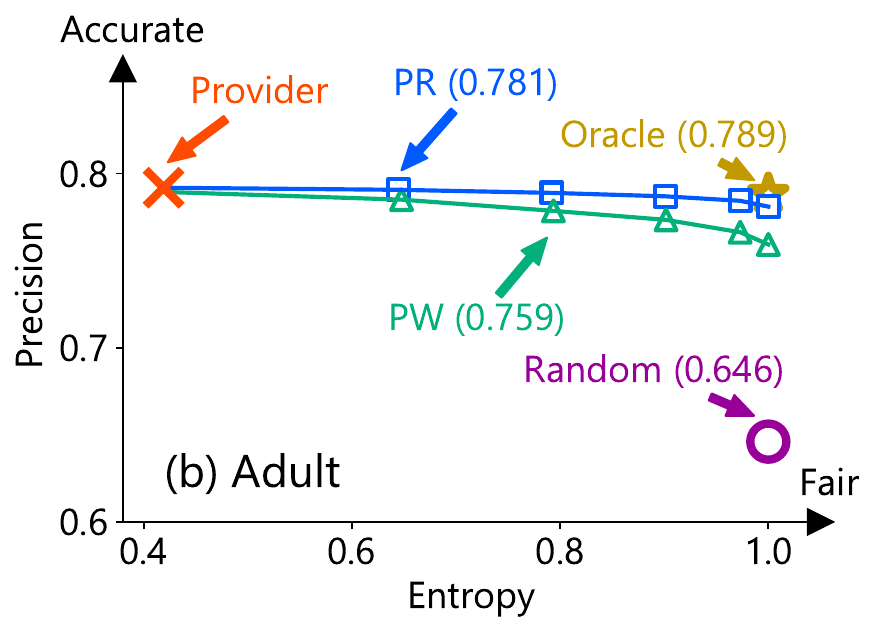}
\end{minipage}
\quad
\begin{minipage}{0.21\hsize}
\includegraphics[width=\hsize]{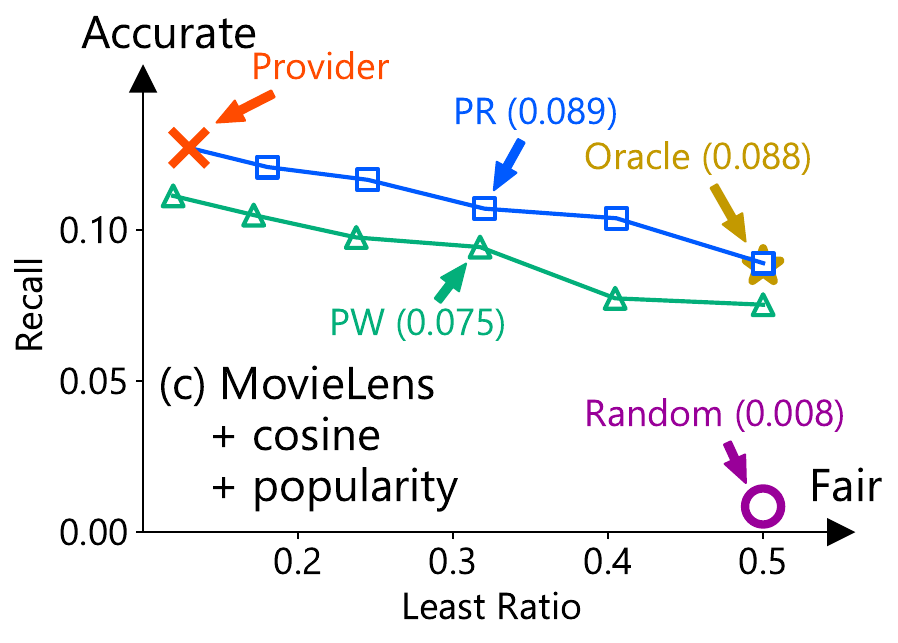}
\end{minipage}
\quad
\begin{minipage}{0.21\hsize}
\includegraphics[width=\hsize]{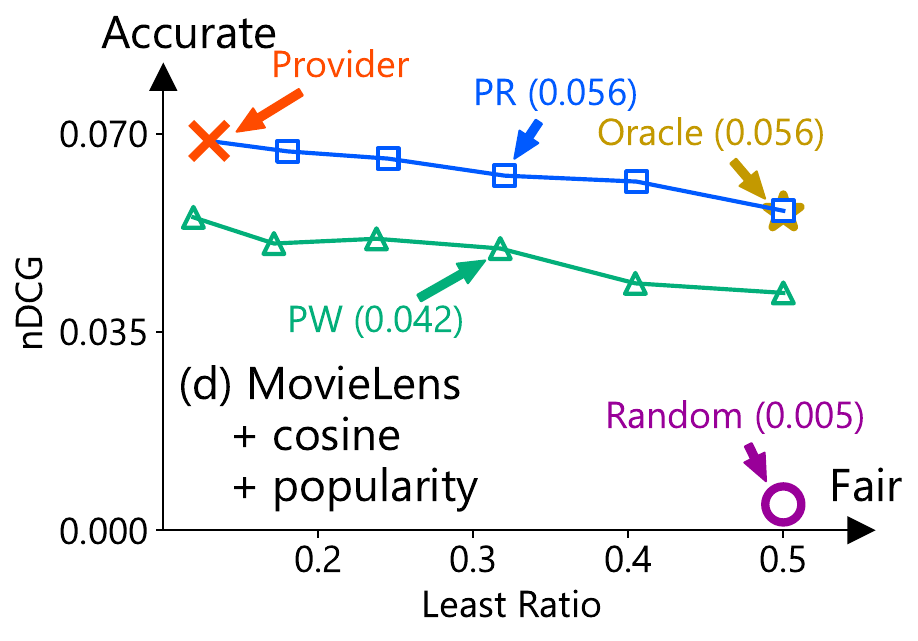}
\end{minipage}
\begin{minipage}{0.21\hsize}
\includegraphics[width=\hsize]{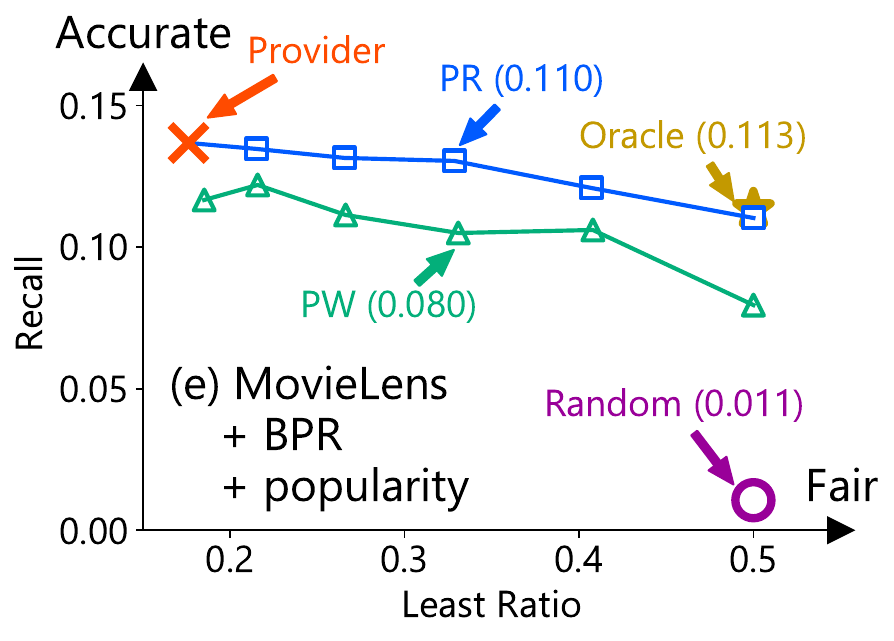}
\end{minipage}
\quad
\begin{minipage}{0.21\hsize}
\includegraphics[width=\hsize]{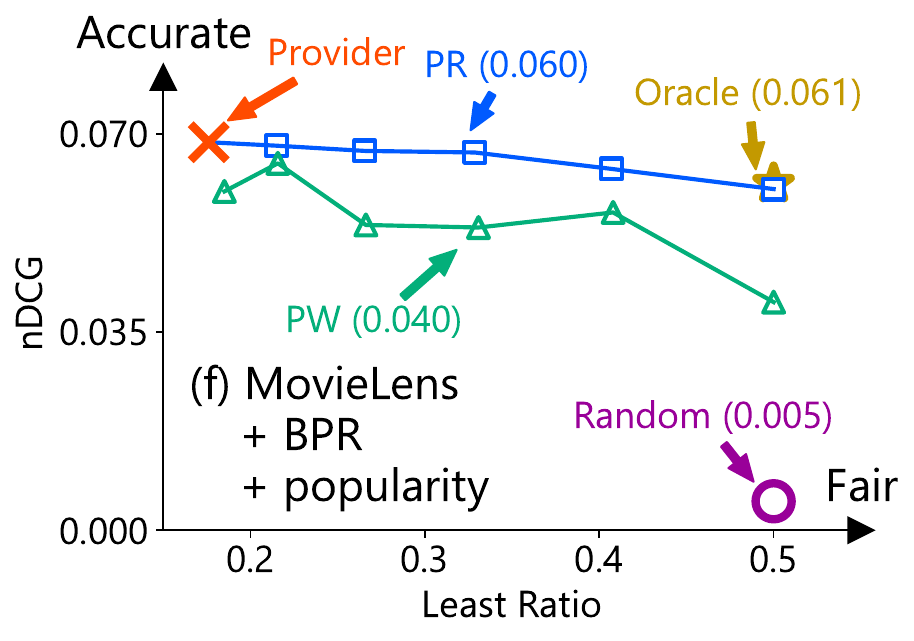}
\end{minipage}
\quad
\begin{minipage}{0.21\hsize}
\includegraphics[width=\hsize]{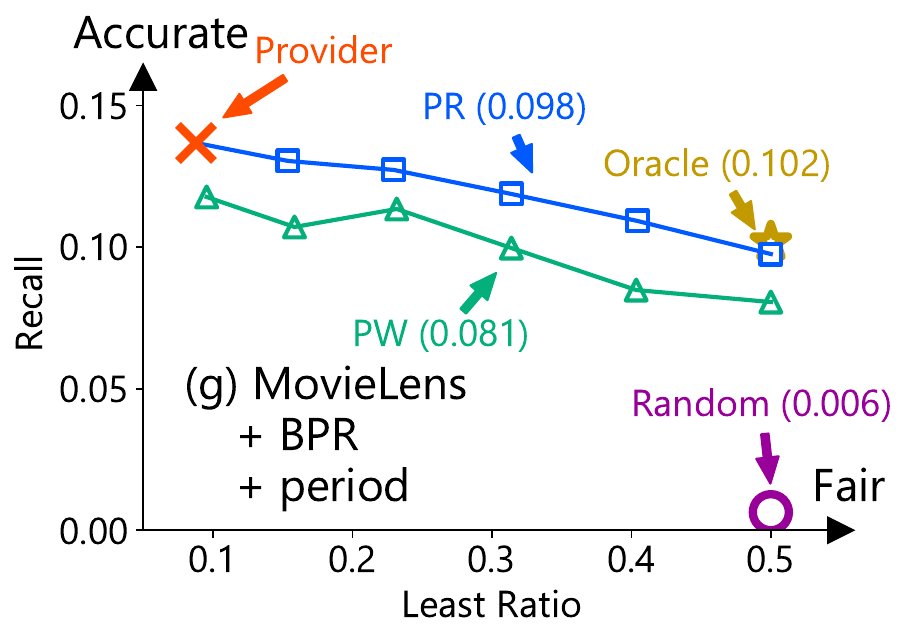}
\end{minipage}
\quad
\begin{minipage}{0.21\hsize}
\includegraphics[width=\hsize]{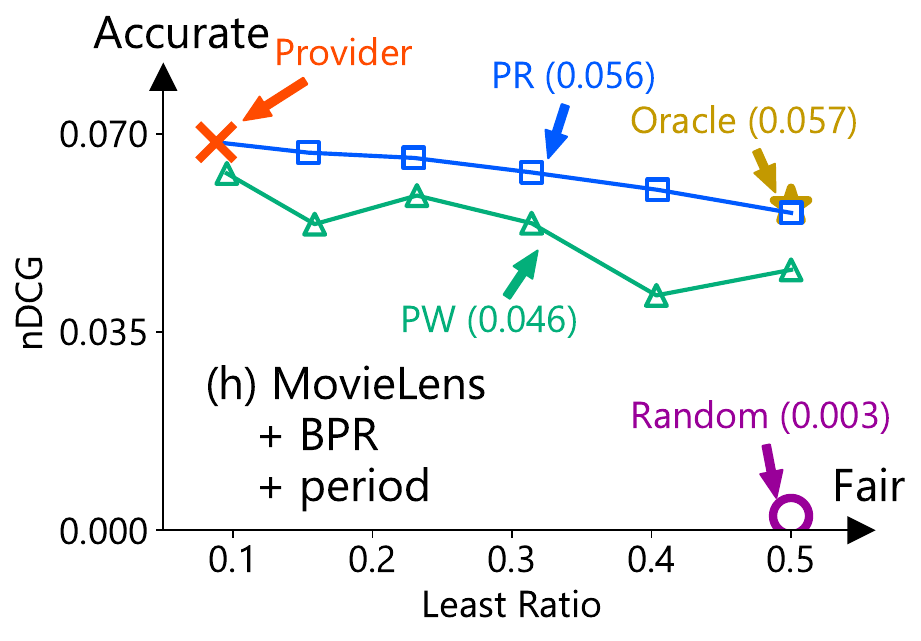}
\end{minipage}
\begin{minipage}{0.21\hsize}
\includegraphics[width=\hsize]{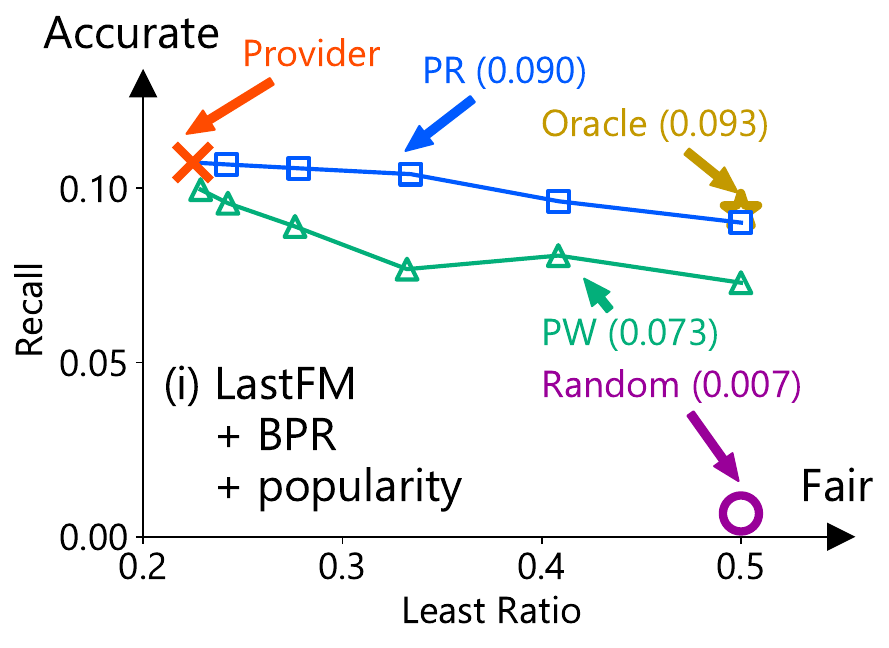}
\end{minipage}
\quad
\begin{minipage}{0.21\hsize}
\includegraphics[width=\hsize]{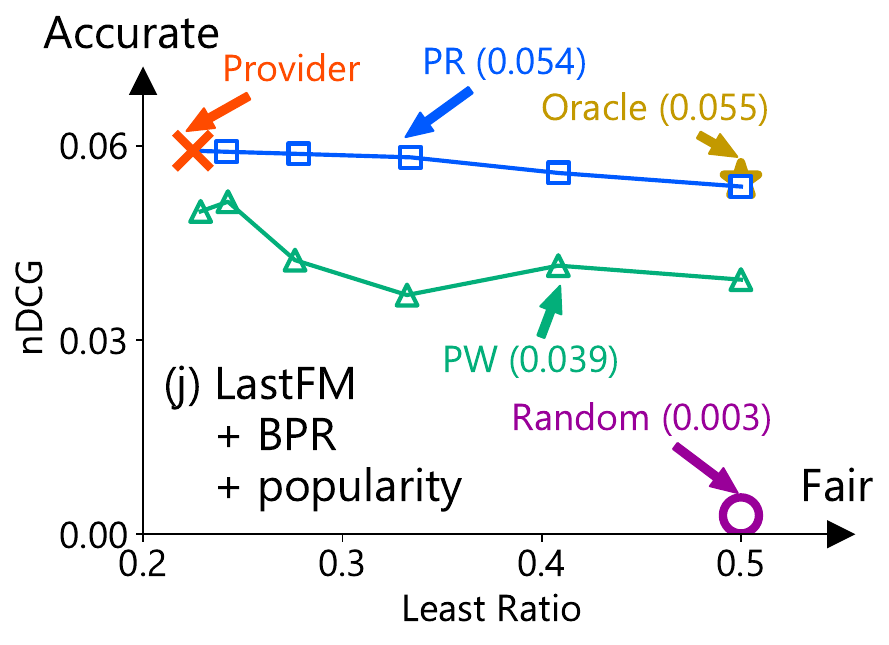}
\end{minipage}
\quad
\begin{minipage}{0.21\hsize}
\includegraphics[width=\hsize]{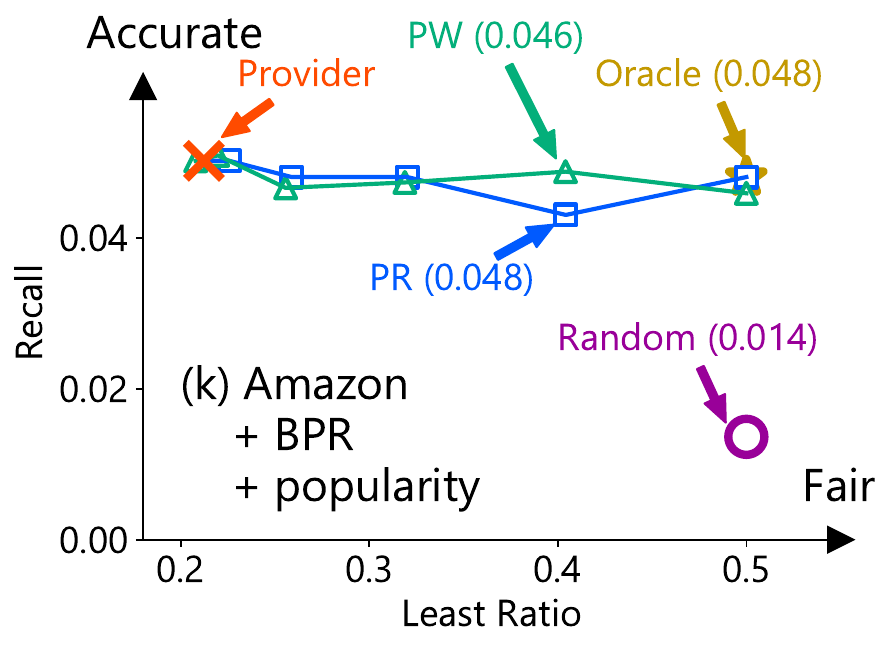}
\end{minipage}
\quad
\begin{minipage}{0.21\hsize}
\includegraphics[width=\hsize]{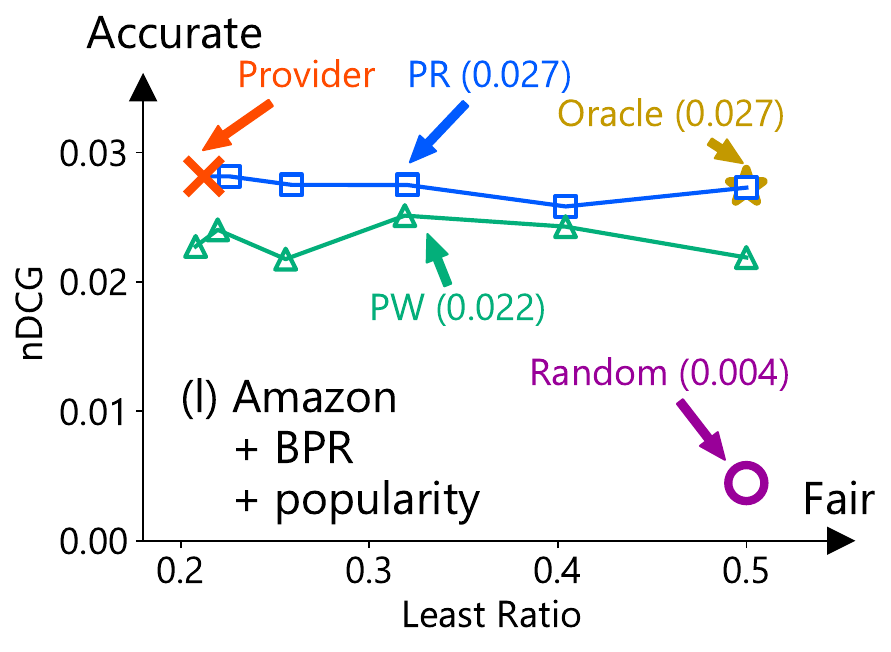}
\end{minipage}
\end{center}
\caption{Trade-off between fairness and performance: PR (blue curve) represents \textsc{PrivateRank}, and PW (green curve) represents \textsc{PrivateWalk}. The score reported in a parenthesis is the performance of each method \emph{when the recommendation is completely fair}. Even though \textsc{PrivateRank} does not use log data, its performance is comparable with the oracle method, which adopts prohibitive information. Although \textsc{PrivateWalk} performs worse than \textsc{PrivateRank} in exchange for fast evaluation, it remains significantly better than random guesses. }
\label{fig: result-private}
\end{figure*}

\subsubsection{Datasets}
We use four datasets, Adult \footnote{\url{https://archive.ics.uci.edu/ml/datasets/adult}}, MovieLens100k \citep{harper2016movielens}, LastFM \footnote{\url{https://grouplens.org/datasets/hetrec-2011/}}, and Amazon Home and Kitchen \citep{he2016ups, mcauley2015image}. In Adult, we recommend a set of people on each person's page, such that those recommended have the same label as the source person. A recommendation list is considered to be fair if it contains men and women in equal proportion. We use this dataset with talent searching in mind. In MovieLens, we consider two protected groups (i) movies released before 1990 and (ii) movies that received less than $50$ interactions. In LastFM and Amazon, items that received less than $50$ interactions are the protected groups. We adopt the implicit setting for MovieLens, LastFM, and Amazon datasets. We set the $ij$-th element of the interaction matrix to one if user $i$ has interacted with item $j$ and zero otherwise. More details about the datasets are available in the appendices.

\subsubsection{Provider Recommender System}
To construct private recommender systems, we first define the provider recommender system.

\noindent \textbf{Adult dataset.} We use a $K$-nearest neighbor recommendation for this dataset. It first standardizes features and recommends $K$-nearest people with respect to the Euclidean distance for each person.

\noindent \textbf{MovieLens, LastFM, and Amazon datasets.} We use nearest neighbor recommendation using rating vectors and Bayesian personalized ranking (BPR) \citep{rendle2009bpr}. In BPR, the similarity of items $i$ and $j$ is defined as the inner product of the latent vectors of items $i$ and $j$. The top-$K$ similar items with respect to these similarity measures are recommended in both methods. We use the Implicit package\footnote{\url{https://github.com/benfred/implicit}} to implement both methods.

\subsubsection{Evaluation Protocol} \hfill

\noindent \textbf{Adult dataset.} 
We evaluate methods using precision, i.e., the ratio of recommended items with the same class label as the source item.

\noindent \textbf{MovieLens and Amazon datasets.} We adopt leave-one-out evaluation following previous works \citep{he2017neural, rendle2009bpr}. We do not adopt negative sampling but full evaluation to avoid biased evaluation \citep{krichene2020sampled}. We evaluate methods using recall@$K$ and nDCG@$K$. These metrics are computed for the recommendation list for the second latest item that user $u$ interacted with. The positive sample is the latest interacted item, which is left out in the dataset.

\noindent \textbf{LastFM dataset.} This dataset does not contain timestamps. We randomly arrange the interactions and adopt leave-one-out evaluation as in MovieLens and Amazon datasets.

\subsubsection{Hyperparameters} Throughout the experiments, we set the length $K$ of recommendation lists to be $10$ for both provider and private recommender systems.
We set $c = 0.01$, $L = 10$ for \textsc{PrivateRank} and $L_\text{max} = 100$ for \textsc{PrivateWalk}. We inspect the trade-off between performance and fairness by varying $\tau$.
We use the default hyperparameters in the Implicit package for the BPR. Namely, the number of dimensions is $100$, the learning rate is $0.01$, the regularization parameter is $0.01$, and the number of iterations is $100$.

\subsection{Fairness and Performance Trade-off (RQ1)} \label{sec: result}

\subsubsection{Baselines} First, \textbf{Provider} is the provider recommender system. It serves as an upper bound of private recommender systems in terms of accuracy; however, it is unfair. \textbf{Random} shuffles items in random order and recommends items in a fair manner by the same post processing as \textsc{PrivateRank}. \textbf{Oracle} is a posthoc fair recommender system that adopts the same backbone algorithm as the provider recommender system and the same post processing as \textsc{PrivateRank}. This algorithm corresponds to the posthoc method in \citep{geyik2019fairness}. Note that this method is unrealistic and cannot be used in our setting because it does have access to log data. We use \textbf{Oracle} to investigate the performance deviations from idealistic settings. Note that existing methods for C-fairness, such as FATR \citep{zhu2018fairness} and Beyond Parity \citep{yao2017beyond}, cannot be used as baselines in this setting because we adopt fairness for \emph{items} in item-to-item recommendations in this work.

\subsubsection{Results}
Figure \ref{fig: result-private} shows the trade-off between the fairness measures and the performance of the recommender systems. The score reported in parenthesis is the performance of each method \emph{when the recommendation is completely fair} (i.e., least ratio is $0.5$ and entropy is $1.0$, not the score at the pointed position).

\begin{table}[tb]
    \caption{Case studies on IMDb and Twitter. \textsc{PrivateWalk} can retrieve relevant items in a fair manner, although it does not know the detail of the provider recommender system or have access to log data. Recall that we do not take orders into consideration in the fairness criterion. Once both groups are contained in a recommendation list, it is easy to adjust the order by any posthoc processing. (e.g., show the USA and non-USA movies alternatively in the first example if you want.)} 
    \centering
    \scalebox{0.65}{
    \begin{tabular}{lll} \toprule
    & Provider (IMDb) & \textsc{PrivateWalk} \\ \midrule
    Source & Toy Story & Toy Story \\ \midrule
    1st & Toy Story 3 ({\color[HTML]{FF4B00} USA})    & Toy Story 3 ({\color[HTML]{FF4B00} USA})              \\
    2nd & Toy Story 2 ({\color[HTML]{FF4B00} USA})    & Coco ({\color[HTML]{FF4B00} USA})                     \\
    3rd & Finding Nemo ({\color[HTML]{FF4B00} USA})   & The Incredibles ({\color[HTML]{FF4B00} USA})          \\
    4th & Monsters, Inc. ({\color[HTML]{FF4B00} USA}) & Spirited Away ({\color[HTML]{005AFF} non USA})        \\
    5th & Up ({\color[HTML]{FF4B00} USA})             & Castle in the Sky ({\color[HTML]{005AFF} non USA})    \\
    6th & WALL·E ({\color[HTML]{FF4B00} USA})         & Howl's Moving Castle ({\color[HTML]{005AFF} non USA}) \\ \bottomrule
    \end{tabular}
    }
    \newline
    \vspace{0.1in}
    \newline
    \scalebox{0.65}{
    \begin{tabular}{lllll} \toprule
    & Provider (Twitter) & \textsc{PrivateWalk} & Provider (Twitter) & \textsc{PrivateWalk} \\ \midrule
    Source & Tom Hanks & Tom Hanks & Tim Berners-Lee & Tim Berners-Lee \\ \midrule
    1st & Jim Carrey ({\color[HTML]{FF4B00} man})        & Jim Carrey ({\color[HTML]{FF4B00} man})        & Jimmy Wales ({\color[HTML]{FF4B00} man})      & Vinton Gray Cerf ({\color[HTML]{FF4B00} man}) \\
    2nd & Hugh Jackman ({\color[HTML]{FF4B00} man})      & Sarah Silverman ({\color[HTML]{005AFF} woman}) & Vinton Gray Cerf ({\color[HTML]{FF4B00} man}) & Lawrence Lessig ({\color[HTML]{FF4B00} man}) \\
    3rd & Samuel L. Jackson ({\color[HTML]{FF4B00} man}) & Hugh Jackman ({\color[HTML]{FF4B00} man})      & Nigel Shadbolt ({\color[HTML]{FF4B00} man})   & Nigel Shadbolt ({\color[HTML]{FF4B00} man}) \\
    4th & Dwayne Johnson ({\color[HTML]{FF4B00} man})    & Samuel L. Jackson ({\color[HTML]{FF4B00} man}) & Anil Dash ({\color[HTML]{FF4B00} man})        & Kara Swisher ({\color[HTML]{005AFF} woman}) \\
    5th & Seth MacFarlane ({\color[HTML]{FF4B00} man})   & Emma Watson ({\color[HTML]{005AFF} woman})     & Lawrence Lessig ({\color[HTML]{FF4B00} man})  & danah boyd ({\color[HTML]{005AFF} woman}) \\
    6th & Sarah Silverman ({\color[HTML]{005AFF} woman}) & Alyssa Milano ({\color[HTML]{005AFF} woman})   & Brendan Eich ({\color[HTML]{FF4B00} man})     & Adrienne Porter Felt ({\color[HTML]{005AFF} woman}) \\ \bottomrule
    \end{tabular}
    }
    \label{table: qualitative}
\end{table}

\noindent \textbf{Adult} (Figures \ref{fig: result-private} (a) and (b)): The least ratio of the provider recommendations is $0.152$, which indicates that the provider recommendations are not fair with respect to sex. In contrast, our proposed methods can increase fairness by increasing the threshold. In particular, \textsc{PrivateRank} strikes an excellent trade-off between fairness measures and precision. It achieves perfect fairness (i.e., least ratio is $0.5$ and entropy is $1.0$) while it drops precision by only one percent. \textsc{PrivateWalk} performs slightly worse than \textsc{PrivateRank} but much better than random access. Because the least ratio and entropy have a one-to-one correspondence, we will report only the least ratios in the following because of space limitation.

\noindent \textbf{MovieLens} (Figures \ref{fig: result-private} (c) to (h)): Cosine and BPR in the figure represent the type of the provider recommender system, and popularity and period represent the protected attributes. \textsc{PrivateRank} strikes an excellent trade-off between accuracy and fairness in all settings. In particular, \textsc{PrivateRank} is comparable with the oracle method, which has access to the unavailable information in our setting. It can also be observed that BPR increases the performance of the provider recommender system compared to the cosine similarity, and it also increases the performance of private recommender systems accordingly. This indicates that effective provider recommender systems induce effective private recommender systems. The overall tendencies are common in all settings. These results indicate that our proposed methods can generate effective recommendations in a fair manner regardless of the algorithms used in the provider recommender system and the protected attributes. Because we observe similar tendencies in Cosine and BPR for other datasets, we report only BPR in the following due to space limitation.

\noindent \textbf{LastFM and Amazon} (Figures \ref{fig: result-private} (i) to (l)): Our proposed methods exhibits good trade-off between accuracy and fairness here as well. These results indicate that our proposed methods are effective regardless of the domain of recommendations.

In summary, our proposed methods can achieve perfect fairness, even where the provider recommender systems are not fair, and perform well in various settings.

\subsection{Sensitivity of Hyperparameters (RQ2)} \label{sec: sensitivity}

We investigate the sensitivity of hyperparameters of \textsc{PrivateRank} and \textsc{PrivateWalk}. We fix the minimum requirement $\tau$ to be $5$ (i.e., perfect fairness) and evaluate performance with various hyperparameters. We evaluate performance by precision for the Adult dataset and by the recall for the other datasets. We normalize these performance measures such that the maximum value is one to illustrate relative drops of performance. Figure \ref{fig: sensitivity} (Left) reports the sensitivity of the number $L$ of iterations of the cumulative power iteration, and Figure \ref{fig: sensitivity} (Center) reports the sensitivity of the damping factor $c$. It can be observed that $L$ is not sensitive as long as $L \ge 6$, and $c$ is neither sensitive as long as $c \le 0.01$ in all settings. Figure \ref{fig: sensitivity} (Right) reports the sensitivity of the maximum length $L_\text{max}$ of random walks. It also shows that this hyperparameter is not sensitive as long as $L_\text{max} \ge 100$. 

\begin{figure}[tb]
\begin{minipage}{0.32\hsize}
\centering
\includegraphics[width=\hsize]{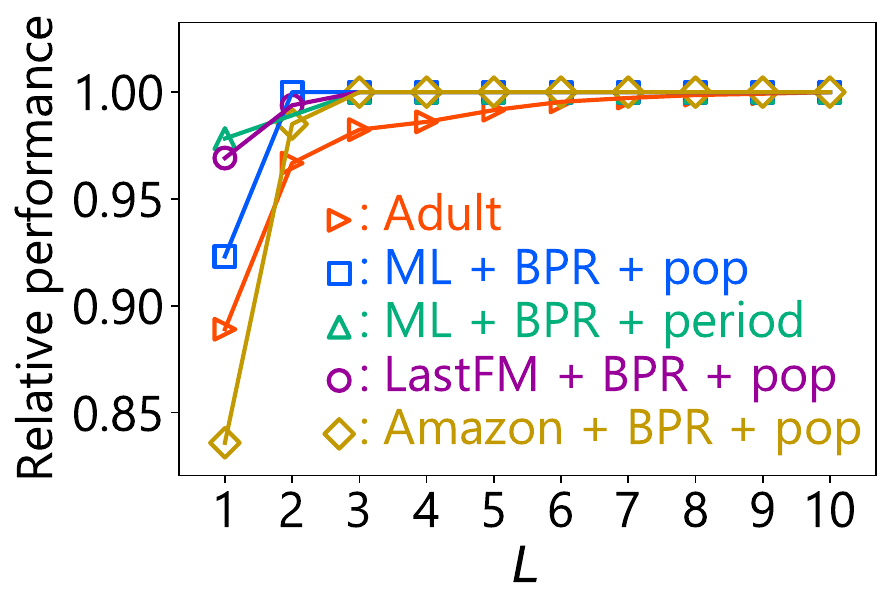}
\end{minipage}
\begin{minipage}{0.32\hsize}
\centering
\includegraphics[width=\hsize]{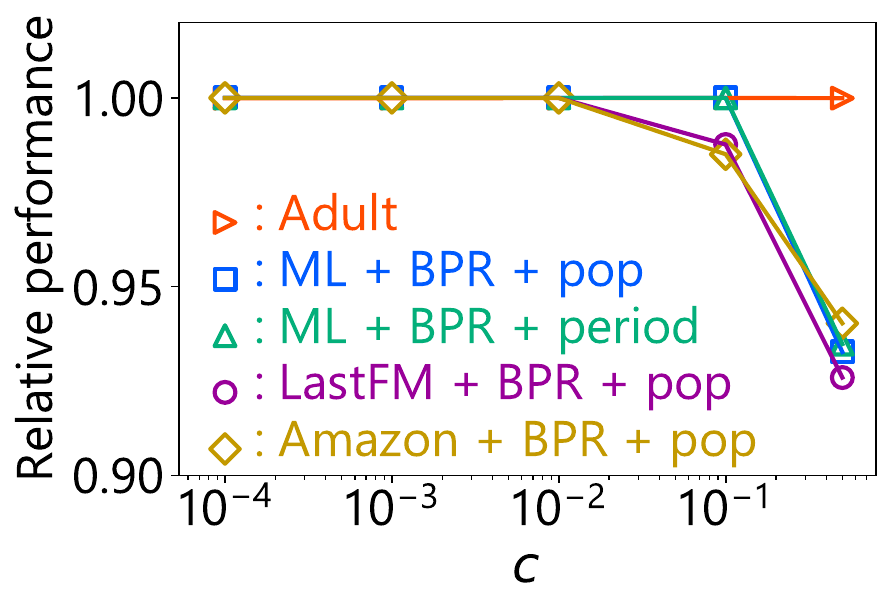}
\end{minipage}
\begin{minipage}{0.32\hsize}
\centering
\includegraphics[width=\hsize]{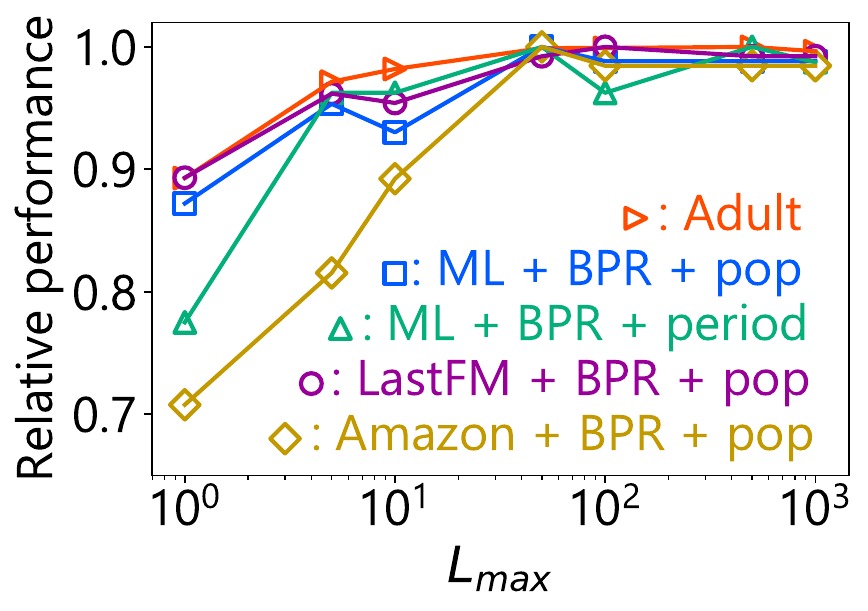}
\end{minipage}
\caption{Sensitivity of hyperparameters: Our methods are insensitive to hyperparameters.}
\label{fig: sensitivity}
\end{figure}

\subsection{Case Studies in the Wild (RQ3)} \label{sec: wild}

We run \textsc{PrivateWalk} with real recommender systems in operation in IMDb and Twitter for qualitative case studies. The main goal of this section is to show that our method is feasible even in real-world environments. To the best of our knowledge, there are no existing methods that enable Twitter users to use fair recommender systems. The experiments in this section provide a proof-of-concept that such a challenging task is feasible.

\noindent \textbf{IMDb:} We found that IMDb recommended only American movies in the Toy Story page\footnote{\url{https://www.imdb.com/title/tt0114709/}}. Although these recommendations are consistent, they are not informative to cinema fans. We consider whether the movie is from the USA as the protected attribute and run \textsc{PrivateWalk} on IMDb. Table \ref{table: qualitative} (Left) reports that \textsc{PrivateWalk} also recommends Spirited Away, Castle in the Sky, and Howl's Moving Castle, which are Japanese animation movies. This result indicates that \textsc{PrivateWalk} recommends in a fair manner with respect to the country attribute while it recommends relevant items.

\noindent \textbf{Twitter:} As we mentioned in Section \ref{sec: setting-private}, we found that Twitter's user recommendations were not fair with respect to gender, which is in contrast to LinkedIn. We set the protected attribute to be gender and run \textsc{PrivateWalk} on Twitter. We annotate labels and remove non-person accounts manually\footnote{As we mentioned in Section \ref{sec: setting-private}, this annotation process can be automated if there is a classifier that takes an account as input and estimates the label.}. Table \ref{table: qualitative} (Middle, Right) reports that \textsc{PrivateWalk} recommends three men users and three woman users in both examples, and all recommended users are celebrity actors for Tom Hanks, and tech-related people for Tim Berners-Lee.

We emphasize that we do not know the recommendation algorithms used in IMDb and Twitter nor have access to log data. Nonetheless, \textsc{PrivateWalk} generates fair recommendations by utilizing blackbox recommendation results. We also point out that \textsc{PrivateRank} is infeasible in this setting because there are too many items in IMDb and Twitter to hit the API limit. In contrast, \textsc{PrivateWalk} runs with few evaluations of the provider recommender systems in an on-demand manner. \textsc{PrivateRank} is beneficial when the number of items is small (e.g., when we retrieve items only from a specific category) or when we can crawl the website beforehand because it performs better when allowed.

\section{Discussion}
\noindent \uline{\textbf{Limitation.}} The main limitation of private recommender systems is that it is expensive for each user to develop a system. Although it may still be worthwhile for companies to build their own private recommender systems for a fair recruiting process, most individual users cannot afford to develop their own browser add-ons. We consider two specific scenarios in which individual users benefit from private recommender systems. First, some enthusiastic users of the service can develop browser add-ons and distribute them. For example, a cinema fan may develop a fair private recommender system for IMDb with respect to popularity as a hobby project, such that other IMDb users can enjoy it. Second, general purpose software may help build private recommender systems. In particular, the program we developed in Section \ref{sec: wild} computes random walks and rankings in common procedures, and we only have to specify in which elements of a web page recommendation lists and sensitive attributes are described. Although users have to write additional scripts for each service in the current version, more elaborated software may further reduce the burden of developing private recommender systems.

We have focused on item-to-item collaborative filtering, which is one of the popular choices in steerable recommender systems \citep{green2009generating, lamprecht2015improving}. Extending the concept of private recommender systems to other settings, including item-to-user setting, would be fruitful future direction.

\vspace{0.05in}
\noindent \uline{\textbf{Extending recommendation lists.}} So far, we have assumed that the length of the recommendation lists of the private recommender system is the same as that of the provider recommender system. However, some users may want to know more items than the service provider offers. Our proposed methods can drop this assumption and provide longer lists than the service provider.

\section{Related Work}

Burke \citep{bruke2017mltisided} classified fairness of recommender systems into three categories. C-fairness is the fairness of consumers or users of a service, P-fairness is the fairness of producers or items, and CP-fairness considers both sides. In this study, we focused on P-fairness. Note that existing methods for C-fairness such as FATR \citep{zhu2018fairness} cannot be applied in this setting. Examples of previous studies on P-fairness are as follows. Geyik et al. \citep{geyik2019fairness} proposed a fair ranking of job-seekers in LinkedIn Talent Search. Ekstrand et al. \citep{ekstrand2018exploring} studied book recommendations and found that some recommendation algorithms are biased toward books written by men authors. Beutel et al. \citep{beutel2019fairness} were concerned that unfair recommendations in social network services might under-rank posts by some demographic groups, which limits the groups' visibility. Mehrotra et al. \citep{mehrotra2018towards} and Patro et al. \citep{patro2020fairrec} proposed methods to equalize the visibility of items and realize a fair marketplace. The crucial difference between these studies and our research is that all of these previous methods are for service providers and require log data. In contrast, our methods can be used even without log data and even if a service provider does not offer a fair recommender system. To the best of our knowledge, this is the first work to address this challenging problem.

\section{Conclusion}

In this chapter, we investigated a situation where a service provider does not offer a fair recommender system for the first time. We proposed a novel framework, private recommender systems, where each user builds their own fair recommender system. We proposed \textsc{PrivateRank} and \textsc{PrivateWalk} to build a private recommender system without requiring access to log data. \textsc{PrivateRank} is effective and exhibits an excellent trade-off between performance and fairness; however, it requires many evaluations of the provider recommender system. Although \textsc{PrivateWalk} is less effective, it makes recommendations in an on-demand manner and requires few evaluations of the provider recommender system. The two proposed methods complement each other’s weaknesses. We empirically validated the effectiveness of the proposed methods via offline quantitative evaluations and qualitative experiments in the wild. Our approach realizes fair recommendations in many cases where conventional fair algorithms cannot be deployed, and it promotes the spread of fair systems.

\chapter*{Appendix}

\section*{Notations}

Notations are summarized in Table \ref{tab: notations-private}.

\begin{table}[tb]
    \centering
    \caption{Notations.}
    \begin{tabular}{ll} \toprule
        Notations & Descriptions \\ \midrule
        $[n]$ & The set $\{ 1, 2, \dots, n \}$. \\
        $a, \bolda, \boldA$ & A scalar, vector, and matrix. \\
        $\boldA^\top$ & The transpose of $\boldA$. \\
        $K$ & The length of a recommendation list. \\
        $\mathcal{U} = [m]$ & The set of users. \\
        $\mathcal{I} = [n]$ & The set of items. \\
        $\mathcal{A}$ & The set of protected groups. \\
        $a_i \in \mathcal{A}$ & The protected attribute of item $i \in \mathcal{I}$. \\
        $\mathcal{P}_{\text{prov}}$ & A provider recommender system. \\
        $\mathcal{Q}$ & A private recommender system. \\
        \bottomrule
    \end{tabular}
    \label{tab: notations-private}
\end{table}

\section*{Pseudo Code of PrivateRank} \label{sec: PrivateRankcode}

\setlength{\textfloatsep}{5pt}
\begin{algorithm2e}[t]
\caption{CanAdd$(\mathcal{R}, i)$}
\label{algo: can}
\DontPrintSemicolon 
\nl\KwData{List $\mathcal{R}$ of items, New item $i$.}
\nl\KwResult{Whether we can add item $i$.}
    \nl Initialize $\boldc_a \leftarrow 0 ~(\forall a \in \mathcal{A}$) \;
    \nl \For{$j \textup{ \textbf{in} } \mathcal{R} \cup \{i\}$}{
    \nl     $\boldc_{a_j} \leftarrow \boldc_{a_j} + 1$ \tcp*{Count attributes}
        }
    \nl \textbf{return} $\sum_{a \in \mathcal{A}} \max(0, \tau - c_a) \le K - \text{len}(\mathcal{R})$ \;
\end{algorithm2e}

Algorithm \ref{algo: PrivateRank} describes the pseudo-code of \textsc{PrivateRank}. We construct the recommendation network in Line 4 and compute the personalized PageRank in Line 5. In Lines 6--7, we iterate items in the descending order of the personalized PageRank. In Lines 8--9, we add item $i$ to the recommendation list if the user has not interacted with item $i$, and we can preserve the constraint when we insert item $i$.

\section*{Proof of Theorem 4.1} \label{sec: thm1}

\begin{proof}
Let $o_i \in [n]$ be the rank of item $i$ in $\hat{\boldS}_s$. Item $i$ appears in the $o_i$-th iteration, and item ord$_i$ appears in the $i$-th iteration of the loop in Lines 7--11 in Algorithm \ref{algo: PrivateRank}. We prove that for each iteration $j \in [n]$ and each sensitive attribute $a \in \mathcal{A}$, $|\{i \in \mathcal{R} \mid a_i = a\} \cup \{i \in \mathcal{I} \mid a_i = a \land o_i \ge j\}| \ge \tau$ holds at the start of the $j$-th iteration. We prove this assertion by mathematical induction. The $j = 1$ case holds because there are at least $\tau$ items of each attribute. Suppose the proposition holds for $j$. Let $\hat{a} = a_{\text{ord}_j}$. The proposition holds for sensitive attribute $a \neq \hat{a}$ because neither $\{i \in \mathcal{R} \mid a_i = a\}$ nor $\{i \in \mathcal{I} \mid a_i = a \land o_i \ge j\}$ changes in the $j$-th iteration. If $|\{i \in \mathcal{R} \mid a_i = \hat{a}\}| \ge \tau$ holds in the $j$-th iteration, the proposition holds for $\hat{a}$. Otherwise, item $\text{ord}_j$ is adopted in the $j$-th iteration because $\sum_{a \in \mathcal{A}} \max(0, \tau - c_a)$ decreases by one, and CanAdd returns \textbf{True}. Therefore, $|\{i \in \mathcal{R} \mid a_i = a\}|$ increases by one, and $|\{i \in \mathcal{I} \mid a_i = a \land o_i \ge j\}|$ decreases by one at the $j$-th iteration. From the inductive hypothesis, the proposition holds for $j + 1$. Therefore, $|\{i \in \mathcal{R} \mid a_i = a\}| \ge \tau$ holds true at the end of the procedure for each attribute $a \in \mathcal{A}$. 
\end{proof}

\section*{Proof of Theorem 4.2} \label{sec: thm2}

\begin{proof}
Let $D = \sum_{i=1}^K \frac{1}{\log (i + 1)}$. From the definition,
\begin{align*}
    \hat{\boldS}_i = (1-c) \bolde^{(i)} + (1-c)c \left( \tilde{\boldA}^\top \bolde^{(i)} + c\tilde{\boldA}^\top \sum_{k = 0}^{L-2} (c \tilde{\boldA}^\top)^k \bolde^{(i)} \right).
\end{align*}
Therefore, for the $k$-th item $j$ of $\mathcal{P}_{\text{prov}, 1}(i)$, 
\[\hat{\boldS}_{ij} \ge (1-c)c (\tilde{\boldA}^\top \bolde^{(i)})_j = (1-c)c \frac{1}{D \log (k + 1)}.\]
Suppose $c < \frac{1}{(K + 1)^2 \log^2(K + 1)}$ holds, then, 
\begin{align*}
    & \|c \tilde{\boldA}^\top \sum_{k = 0}^{L - 2} (c \tilde{\boldA}^\top)^k \bolde^{(i)}\|_\infty \\
    &\le c \|\tilde{\boldA}^\top\|_1 \|\sum_{k = 0}^{\infty} (c \tilde{\boldA}^\top)^k\|_1 \|\bolde^{(i)}\|_1 \\
    &\le c \|\tilde{\boldA}^\top\|_1 \frac{1}{1 - \|c \tilde{\boldA}^\top\|_1} \|\bolde^{(i)}\|_1 \\
    &\le \frac{c}{1 - c} = \frac{1}{(K + 1)^2 \log^2(K + 1) - 1} \\
    &< \frac{1}{D(K+1)\log^2(K + 1)} < \frac{\left( \frac{1}{\log K} - \frac{1}{\log (K + 1)} \right)}{D}, \\
\end{align*}
where $\|\bolda\|_1$ is the induce norm. Therefore, for $l \not \in \mathcal{P}_{\text{prov}, 1}(i) \cup \{i\}$, 
\[ \hat{\boldS}_{il} = (1-c)c (c\tilde{\boldA}^\top \sum_{k = 0}^{L-2} (c \tilde{\boldA}^\top)^k \bolde^{(i)}) < (1-c)c \frac{1}{D\log (K + 1)},\] and the $k$-th item $j$ of $\mathcal{P}_{\text{prov}, 1}(i)$, 
\[\hat{\boldS}_{ij} < (1-c)c \frac{1}{D\log (k)}.\]
Therefore, the top-$K$ ranking does not change.

\end{proof}

\section*{Pseudo Code of PrivateWalk} \label{sec: PrivateWalkcode}

\setlength{\textfloatsep}{5pt}
\begin{algorithm2e}[t]
\caption{\textsc{PrivateRank}}
\label{algo: PrivateRank}
\DontPrintSemicolon 
\nl\KwData{Oracle access to $\mathcal{P}_{\text{prov}, 1}$, Source item $i \in \mathcal{I}$, Protected attributes $a_i ~\forall i \in \mathcal{I}$, Damping factor $c$, Minimum requirement $\tau$, Set $\mathcal{H}$ of items that user $i$ has already interacted with.}
\nl\KwResult{Recommended items $\mathcal{R} = \{j_k\}_{1 \le k \le K}$.}
    \nl Initialize $\mathcal{R} \leftarrow []$ (empty) \;
    \nl $\boldA_{ij} = \begin{cases}
    \frac{1}{\log(k + 1)} & (\mathcal{P}_{\text{prov}, 1}(i)_k = j) \\
    0 & \text{(otherwise)}
    \end{cases}$\;
    \nl $\hat{\boldS}_{i} = (1 - c) \sum_{k = 0}^L (c \tilde{\boldA}^\top)^k \bolde^{(i)}$ \;
    \nl $l \leftarrow \text{argsort}(\hat{\boldS}_{i}, \text{descending})$ \;
    \nl \For{$i \textup{ \textbf{in} } l$}{
    \nl     \If{$i \textup{ \textbf{not in} } \mathcal{H} \textup{ \textbf{and} } \textup{CanAdd}(\mathcal{R}, i)$}{
    \nl         Push back $i$ to $\mathcal{R}$ \;
            }
        }
    \nl \textbf{return} $\mathcal{R}$ \;
\end{algorithm2e}

\setlength{\textfloatsep}{5pt}
\begin{algorithm2e}[t]
\caption{\textsc{PrivateWalk}}
\label{algo: PrivateWalk}
\DontPrintSemicolon 
\nl\KwData{Oracle access to $\mathcal{P}_{\text{prov}, 1}$, Source item $i \in \mathcal{I}$, Protected attributes $a_i ~\forall i \in \mathcal{I}$, Minimum requirement $\tau$, Set $\mathcal{H}$ of items that user $i$ has already interacted with, Maximum length $L_{\text{max}}$ of random walks.}
\nl\KwResult{Recommended items $\mathcal{R} = \{j_k\}_{1 \le k \le K}$.}
    \nl Initialize $\mathcal{R} \leftarrow []$ (empty) \;
    \nl \For{$k \gets 1$ \textbf{to} $K$}{
    \nl     \text{cur} $\leftarrow i$ \tcp*{start from the source node}
    \nl     \text{found} $\leftarrow$ \textbf{False} \;
    \nl     \For{\textup{iter} $\gets 1$ \textbf{to} $L_{\text{max}}$}{
    \nl         \text{next\_rank} $\sim$ $\text{Cat}(\text{normalize}([1 / \log (r + 1)]))$ \;
    \nl         \text{cur} $\leftarrow \mathcal{P}_{\text{prov}, 1}(\text{cur})_{\text{next\_rank}}$ \;
    \nl         \If{$\textup{cur} \textup{ \textbf{not} \textbf{in} } \mathcal{R} \cup \mathcal{H}$ \textup{\textbf{and}} $\textup{CanAdd}(\mathcal{R}, \textup{cur})$}{
    \nl             \tcc{The first encountered item that can be added. Avoid items in $\mathcal{R} \cup \mathcal{H}$.}
    \nl             Push back \text{cur} to $\mathcal{R}$. \;
    \nl             \text{found} $\leftarrow$ \textbf{True} \;
    \nl             \textbf{break} \;
                }
            }
    \nl     \While{\textbf{\textup{not}} \textup{found}}{
    \nl         $i \leftarrow$ \text{Uniform}($\mathcal{I}$) \tcp*{random item}
    \nl         \If{$i \textup{ \textbf{not} \textbf{in} } \mathcal{R} \cup \mathcal{H}$ \textup{\textbf{and}} $\textup{CanAdd}(\mathcal{R}, i)$}{
    \nl             Push back \text{cur} to $\mathcal{R}$. \;
    \nl             \textbf{break} \;
                }
            }
        }
    \nl \textbf{return} $\mathcal{R}$ \;
\end{algorithm2e}

Algorithm \ref{algo: PrivateWalk} presents the pseudo-code of \textsc{PrivateWalk}. $\text{Cat}(\theta)$ denotes a categorical distribution of parameter $\theta$, and $\text{normalize}(x)$ denotes the normalization of vector $x$, such that the sum is equal to one. We run a random walk in Lines 7--14. In Line 8, the next node is sampled with a probability proportional to edge weights. If we cannot find an appropriate item in $L_\text{max}$ steps, we add a random item in Lines 15--19. Note that we can determine an item in Lines 7--14 with high probability in practice. We add the fallback process in Lines 15--19 to make the algorithm well-defined.

\section*{Datasets} \label{sec: datasets}

\begin{itemize}
    \item \textbf{Adult} \citep{dua2017uci}. In this dataset, each record represents a person and contains demographic data such as age, sex, race, educational background, and income. This dataset has been mainly adopted in supervised learning \citep{kamishima2012fairness, xu2020algorithmic}, where the covariates are demographic attributes, and the label represents whether income exceeds \$50,000 per year. We use this dataset for individual recommendations. We regard an individual record as a talent and recommend a set of people on each person's page, such that those recommended have the same label as the source person. A recommendation list is considered to be fair if it contains men and women in equal proportion. We adopt the same preprocessing as \citep{xu2020algorithmic}. After preprocessing, this dataset contains $39190$ items and $112$ features. We set sex as the protected attribute. We use this dataset with talent searching in mind.
    \item \textbf{MovieLens100k} \citep{harper2016movielens}. In this dataset, an item represents a movie. It contains $943$ users, $1682$ items, and $100000$ interactions in total. Each user has at least $20$ interactions. We consider two protected attributes. The first is based on periods of movies. Some recommender systems may recommend only new movies, but some users may want to know old movies as well. We divide movies into two groups, such that the protected group contains movies that were released before 1990. The second protected attribute is based on popularity. Movies that received less than $50$ interactions are considered to be in the protected group. 
    \item \textbf{LastFM} \footnote{\url{https://grouplens.org/datasets/hetrec-2011/}}. In this dataset, an item represents a piece of music. To discard noisy items and users, we extract $10$-cores of the dataset, such that each user and item have at least $10$ interactions. Specifically, we iteratively discard items and users with less than $10$ interactions until all items and users have at least $10$ interactions. It contains $1797$ users, $1507$ items, and $62376$ interactions in total after preprocessing. We consider popularity as a protected attribute. The protected group contains music that received less than $50$ interactions.
    \item \textbf{Amazon Home and Kitchen} \citep{he2016ups, mcauley2015image}. In this dataset, an item represents a home and kitchen product on amazon.com. We extract $10$-cores of the dataset using the same preprocessing as in the LastFM dataset. It contains $1395$ users, $1171$ items, and $25445$ interactions in total after preprocessing. We consider popularity as a protected attribute. The protected group contains products that received less than $50$ interactions.
\end{itemize}

\chapter{Towards Principled User-side Recommender Systems}

\section{Introduction}

Recommender systems have been adopted in various web services \citep{linden2003amazon, geyik2019fairness}, and in particular trustworthy systems have been demanded, including fair \citep{kamishima2012enhancement, biega2018equity, milano2020recommender}, transparent \citep{sinha2002role, balog2019transparent}, and steerable \citep{green2009generating, balog2019transparent} recommender systems. However, the adoption of such trustworthy systems is still limited \citep{sato2022private}, and many users still receive black-box and possibly unfair recommendations. Even if these users are dissatisfied with the system, they have limited power to change recommendations. For example, suppose a Democrat supporter may acknowledge that she is biased towards Democrats, but she wants to get news about Republicans as well. However, a news recommender system may show her only news related to Democrats to maximize the click through rate. In other words, filter bubbles exist in recommender systems \citep{pariser2011filter}. The only recourse available to the user is to wait until the news feed service implements a fair/diverse recommendation engine. Green et al. \cite{green2009generating} also noted that ``If users are unsatisfied with the recommendations generated by a particular system, often their only way to change how recommendations are generated in the future is to provide thumbs-up or thumbs-down ratings to the system.'' Worse, there are many fairness/diversity criteria, and the service provider cannot address all the criteria. For example, even if the recommendation results are fair with respect to race, some other users may call for fairness with respect to gender. Or some users call for demographic parity when the service realizes the equalized odds.

User-side recommender systems \citep{sato2022private} provide a proactive solution to this problem. The users can build their own (i.e., private, personal, or user-side) recommender systems to ensure that recommendations are generated in a fair and transparent manner. As the system is build by the user, the system can be tailored to meet their own criteria with user-defined protected groups and user-defined criteria. Therefore, user-side recommender systems can be seen as ultimate personalization.

Although this concept is similar to that of steerable \citep{green2009generating} (or scrutable \citep{balog2019transparent}) recommender systems, the difference lies in the fact that steerable systems are implemented by the service provider, whereas user-side recommender systems are implemented by the users. If the recommender system in operation is not steerable, users need to wait for the service provider to implement a steerable system. By contrast, user-side recommender systems enable the users to make the recommender system steerable by their own efforts even if the service provider implemented an ordinary non-steerable system.

Although user-side recommender systems are desirable, the problem of building user-side recommender systems is challenging. In particular, end-users do not have access to the trade-secret data stored in the database of the system unlike the developers employed by the service provider. This introduces information-theoretic limitation when users build user-side recommender systems. Most modern recommender systems use users' log data and/or item features to make recommendations. At first glance, it seems impossible to build an effective recommender system without such data. Satp \cite{sato2022private} solved this problem by utilizing official recommender systems in the targeted web service, whose outputs are available yet are black-box and possibly unfair.
Specifically, Sato \cite{sato2022private} proposed two algorithms for user-side recommender systems, \textsc{PrivateRank} and \textsc{PrivateWalk} and achieved empirically promising results. However, they rely on ad-hoc heuristics, and it is still unclear whether building recommender systems without log data is a well-defined problem. In this study, we address this question by using the metric recovery theory for unweighted $k$-nearest neighbor graphs. This result provides a theoretically grounded way to construct user-side recommender systems. However, we found that this approach suffers from large communication costs, and thus is not practical.

Then, we formulate the desirable properties for user-side recommender systems, including minimal communication cost, based on which we propose a method, \textsc{Consul}, that satisfies these properties. This method is in contrast to the existing methods that lack at least one axioms.

The contributions of this study are as follows:
\begin{itemize}
    \item The metric recovery theory validates that curious users can reconstruct the original information of items without accessing log data or item databases and solely from the recommendation results shown in the web page. On the one hand, this result shows that the raw recommendation suggested to the user contains sufficient information to construct user-side recommender systems. This result indicates that user-side recommender systems are feasible in principle. One the other hand, it raises new concerns regarding privacy (Section \ref{sec: recover}).
    \item We list three desirable properties of user-side recommender systems, i.e., consistency, soundness, and locality, and show that existing user-side recommender systems lack at least one of them (Section \ref{sec: principles}).
    \item We propose an algorithm for user-side recommender systems, \textsc{Consul}, which satisfies all three properties, i.e., consistency, soundness, and locality (Section \ref{sec: proposed}).
    \item We empirically validate that personal information can be reconstructed solely from raw recommendation results (Section \ref{sec: experiments-reverse}).
    \item We empirically demonstrate that our proposed method strikes excellent trade-off between effectiveness and communication efficiency (Section \ref{sec: experiments-performance}).
    \item We conduct case studies employing crowd-sourcing workers and show that user-side recommender systems can retrieve novel information that the provider's official recommender system cannot (Section \ref{sec: experiments-user}).
    \item We deploy \textsc{Consul} in the real-world Twitter environment and confirmed that \textsc{Consul} realizes an effective and efficient user-side recommender system. (Section \ref{sec: twitter}).
\end{itemize}

\begin{tcolorbox}[colframe=gray!20,colback=gray!20,sharp corners]
\textbf{Reproducibility}: Our code is available at \url{https://github.com/joisino/consul}.
\end{tcolorbox}

\section{Notations}

For every positive integer $n \in \mathbb{Z}_+$, $[n]$ denotes the set $\{ 1, 2, \dots n \}$.
A lowercase letter, such as $a$, $b$, and $c$, denotes a scalar, and a bold lower letter, such as $\bolda$, $\boldb$, and $\boldc$, denotes a vector. 
Let $\mathcal{I} = [n]$ denote the set of items, where $n$ is the number of items. Without loss of generality, we assume that the items are numbered with $1, \dots, n$. $K \in \mathbb{Z}_+$ denotes the length of a recommendation list. 

\section{Problem Setting} \label{sec: setting-consul}

\begin{table}[tb]
    \centering
    \caption{Notations.}
    \begin{tabular}{ll} \toprule
        Notations & Descriptions \\ \midrule
        $[n]$ & The set $\{ 1, 2, \dots, n \}$. \\
        $a, \bolda$ & A scalar and vector. \\
        $G = (V, E)$ & A graph. \\
        $\mathcal{I} = [n]$ & The set of items. \\
        $\mathcal{A}$ & The set of protected groups. \\
        $a_i \in \mathcal{A}$ & The protected attribute of item $i \in \mathcal{I}$. \\
        $\mathcal{H} \subset \mathcal{I}$ & The set of items that have been interacted with. \\
        $\mathcal{P}_{\text{prov}}$ & The provider's official recommender system. \\
        $K \in \mathbb{Z}_+$ & The length of a recommendation list. \\
        $\tau \in \mathbb{Z}_{\ge 0}$ & The minimal requirement of fairness. \\
        $d \in \mathbb{Z}_+$ & The number of dimensions of the embeddings. \\
        \bottomrule
    \end{tabular}
    \label{tab: notations-consul}
\end{table}

\textbf{Provider Recommender System.} In this study, we focus on item-to-item recommendations, following \cite{sato2022private}\footnote{Note that item-to-user recommender systems can also be built based on on item-to-item recommender systems by, e.g., gathering items recommended by an item-to-item recommender system for the items that the user has purchased.}. Specifically, when we visit the page associated with item $i$, $K$ items $\mathcal{P}_\text{prov}(i) \in \mathcal{I}^K$ are presented, i.e., for $k = 1, 2, \cdots, K$, $\mathcal{P}_\text{prov}(i)_k \in \mathcal{I}$ is the $k$-th relevant item to item $i$. For example, $\mathcal{P}_\text{prov}$ is realized in the ``Customers who liked this also liked'' form in e-commerce platforms. We call $\mathcal{P}_\text{prov}$ the \emph{service provider's official recommender system}. We assume that $\mathcal{P}_\text{prov}$ provides relevant items but is unfair and is a black-box system. The goal is to build a fair and white-box recommender system by leveraging the provider's recommender system.

\vspace{0.1in}
\noindent \textbf{Embedding Assumption.} We assume for the provider's recommender system that there exists a representation vector $\boldx_i \in \mathbb{R}^d$ of each item $i$, and that $\mathcal{P}_\text{prov}$ retrieves the top-$k$ nearest items from $\boldx_i$. We do not assume the origin of $\boldx$, and it can be a raw feature vector of an item, a column of user-item interaction matrix, hidden embedding estimated by matrix factorization, or hidden representation computed by neural networks. The embeddings can even be personalized. This kind of recommender system is one of the standard methods \citep{barkan2016item2vec, yao2018judging, sato2022private}. It should be noted that \emph{we assume that $\boldx_i \in \mathbb{R}^d$ cannot be observed} because such data are typically stored in confidential databases of the service provider, and the algorithm for embeddings is typically industrial secrets \citep{milano2020recommender}. For example, $\boldx_i$ may be a column of user-item interaction matrix, but a user cannot observe the interaction history of other users or how many times each pair of items is co-purchased. One of the main argument shown below is that even if $\{\boldx_i\}$ is not directly observable, we can ``reverse-engineer'' $\{\boldx_i\}$ solely from the outputs of the official recommender system $\mathcal{P}_\text{prov}$, which is observable by an end-user. Once $\{\boldx_i\}$ is recovered, we can easily construct user-side recommender systems by ourselves using standard techniques \citep{geyik2019fairness, liu2019personalized, zehlike2017fair}.

\vspace{0.1in}
\noindent \textbf{Sensitive Attributes.} We assume that each item $i$ has a discrete sensitive attribute $a_i \in \mathcal{A}$, where $\mathcal{A}$ is the set of sensitive groups. For example, in a talent market service, each item represents a person, and $\mathcal{A}$ can be gender or race. In a news recommender system, each item represents a news article, and $\mathcal{A}$ can be $\{$Republican, Democrat$\}$. 
We want a user-side recommender system that offers items from each group in equal proportions. Note that our algorithm can be extended to any specified proportion (e.g., so that demographic parity holds), but we focus on the equal proportion for simplicity. We assume that sensitive attribute $a_i$ \emph{can be observed}, which is the common assumption in \citep{sato2022private}. Admittedly, this assumption does not necessarily hold in practice. However, when this assumption violates, one can estimate $a_i$ from auxiliary information or recovered embedding $\hat{\boldx}_i$, and the estimation of the attribute is an ordinary supervised learning task and can be solved by off-the-shelf methods, such as neural networks and random forests. As the attribute estimation process is not relevant to the core of user-side recommender system algorithms, this study focuses on the setting where the true $a_i$ can be observed.

\vspace{0.1in}
\noindent \textbf{Other Miscellaneous Assumptions.} We assume that the user knows the set $\mathcal{H} \subset \mathcal{I}$ of items that he/she has already interacted with. This is realized by, for instance, accessing the purchase history page. In addition, we assume that $\mathcal{P}_\text{prov}$ does not recommend items in $\mathcal{H}$ or duplicated items. These assumptions are optional and used solely for technical reasons. If $\mathcal{H}$ is not available, our algorithm works by setting $\mathcal{H}$ as the empty set.

\vspace{0.1in}
\noindent The problem setting can be summarized as follows:

\begin{tcolorbox}[colframe=gray!20,colback=gray!20,sharp corners]
\noindent \uline{\textbf{User-side Recommender System Problem.}}\\
\textbf{Given:} Oracle access to the official recommendations $\mathcal{P}_{\text{prov}}$. Sensitive attribute $a_i$ of each item $i \in \mathcal{I}$.\\
\textbf{Output:} A user-side recommender system $\mathcal{Q}\colon \mathcal{I} \to \mathcal{I}^K$ that is fair with respect to $\mathcal{A}$.
\end{tcolorbox}

\section{Interested Users Can Recover Item Embeddings} \label{sec: recover}

We first see that a user without access to the database can recover the item embeddings solely from the recommendation results $\mathcal{P}_\text{prov}$.

\vspace{0.1in}
\noindent \textbf{Recommendation Networks.} A recommendation network is a graph where nodes represent items and edges represent recommendation relations. Recommendation networks have been traditionally utilized to investigate the properties of recommender systems \citep{cano2006topology, celma2008new, seyerlehner2009limitation}. They were also used to construct user-side recommender systems \citep{sato2022private}. Recommendation network $G = (V, E)$ we use in this study is defined as follows:
\begin{itemize}
    \item Node set $V$ is the item set $\mathcal{I}$.
    \item Edge set $E$ is defined by the recommendation results of the provider's recommender system. There exists a directed edge from $i \in V$ to $j \in V$ if item $j$ is included in the recommendation list in item $i$, i.e., $\exists k \in [K] \text{ s.t. } \mathcal{P}_\text{prov}(i)_k = j$.
    \item We do not consider edge weights. 
\end{itemize}
It should be noted that $G$ can be constructed solely by accessing $\mathcal{P}_\text{prov}$. In other words, an end-user can observe $G$. In practice, one can crawl the web pages to retrieve top-$k$ recommendation lists. We use unweighted recommendation networks because weights between two items, i.e., the similarity score of two items, are typically \emph{not} presented to users but only a list of recommended items is available in many web services. This setting makes it difficult to estimate the hidden relationship from the recommendation network. Although the original data that generate the recommendation results, e.g., how many times items $i$ and $j$ are co-purchased, cannot be observed, the recommendation results provide useful information on the similarity of items. A critical observation is that, from the embedding assumption we introduce in Section \ref{sec: setting-consul}, graph $G$ can be seen as the $k$-nearest neighbor ($k$-NN) graph of hidden item embeddings $\{\boldx_i\}$.

\vspace{0.1in}
\noindent \textbf{Embedding Recovery.} We discuss that the original embeddings can be recovered from the unweighted recommendation network based on the metric recovery theory \citep{hashimoto2015metric,luxburg2013density,terada2014local}.

First, it is impossible to exactly recover the original embeddings, but there exist some degrees of freedom such as translation, rotation, reflections, and scaling because such similarity transformations do not change the $k$-NN graph. However, similarity transformations are out of our interest when we construct recommender systems, and fortunately, the original embeddings can be recovered up to similarity transformations.

Hashimoto et al. \cite{hashimoto2015metric} considered metric recovery from the $k$-NN graph of a general point cloud $\{\boldx_i\} \subset \mathbb{R}^d$ assuming $\{\boldx_i\}$ is sampled from an unknown distribution $p(x)$ and showed that under mild assumptions of the embedding distribution, the stationary distribution of the simple random walk on the $k$-NN graph converges to $p(x)$ with appropriate scaling (Corollary 2.3 in \citep{hashimoto2015metric}) when $k = \omega(n^{\frac{2}{d+2}} \log^{\frac{d}{d+2}} n)$. Alamgir et al. \cite{alamgir2012shortest} showed that by setting the weight of the $k$-NN graph $G$ to the scaled density function, the shortest path distance on $G$ converges to the distance between the embeddings. Therefore, distance matrix $D \in \mathbb{R}^{n \times n}$ of $\{\boldx_1, \cdots, \boldx_n\}$ can be estimated from the $k$-NN graph (Theorem S4.5 in \citep{hashimoto2015metric}). From the standard argument on the multi-dimensional scaling \citep{sibson1979studies}, we can estimate the original embeddings $\{\boldx_1, \cdots, \boldx_n\}$ from the distance matrix up to similarity transformations with sufficient number of samples $n$. 

Terada et al. \cite{terada2014local} also considered embedding recovery from unweighted $k$-NN graphs\footnote{Note that the use of the these methods require knowledge of the dimension $d$ of the embeddings. However, when the dimensionality is not available, it can be estimated from the unweighted $k$-NN graph using the estimator proposed in \citep{kleindessner2015dimensionality}. For the sake of simplicity, we assume that the true number of dimensions $d$ is known.}. Theorem 3 in \citep{terada2014local} shows that their proposed method, local ordinal embeddings (LOE), recovers the original embeddings $\{\boldx_i\}$ from the $k$-NN graph up to similarity transformations. Although the original proof of Theorem 3 in \citep{terada2014local} relied on the result of \citep{luxburg2013density}, which is valid only for $1$-dimensional cases, the use of Corollary 2.3 in \citep{hashimoto2015metric} validates the LOE for general dimensions. The strength of LOE is that it can estimate the embeddings in a single step while the density estimation approach \citep{luxburg2013density, hashimoto2015metric} requires intermediate steps of estimating the density function and distance matrix. We use the LOE in our experiments.

In summary, the item embeddings $\{\boldx_i\}$ can be recovered solely from the recommendation network. Once embeddings $\{\boldx_i\}$ are recovered, the user has access to the same information as the service provider, after which it is straightforward to construct user-side recommender systems, e.g., by adopting a post-processing fair recommendation algorithm \citep{geyik2019fairness, liu2019personalized, zehlike2017fair}. We call this approach estimate-then-postprocessing (ETP).

These results demonstrate that the $k$-NN recommendation network, which is observable by an end-user, contains sufficient information for constructing user-side recommender systems. Therefore, from the perspective of information limit, the user-side recommender system problem is feasible.

To the best of our knowledge, this work is the first to adopt the embedding recovery theory to user-side recommender systems.

\vspace{0.1in}
\noindent \textbf{Security of Confidential Data.} Although it was not our original intent, the discussion above reveals a new concern regarding data privacy, i.e., confidential information about items may be unintentionally leaked from the recommender systems. For example, in a talent search service, a recommender system may use personal information, such as educational level, salary, and age, to recommend similar talents. This indicates that we can recover such data solely from the recommendation results, which we numerically confirm in the experiments. Therefore, service providers need to develop rigorous protection techniques to prevent leaks of personal information from the recommender system. This is beyond the scope of this work, and we leave this interesting problem as future work.
 
\vspace{0.1in}
\noindent \textbf{Limitation of ETP.} Although ETP is conceptually sound and provides a theoretical foundation of user-side recommender systems, it is impractical because it incurs high communication costs. LOE requires observing the entire $k$-NN graph, and so do other methods, including the path-based approach \citep{luxburg2013density} and global random walk approach \citep{hashimoto2015metric}. Indeed, von Luxburg et al. \cite{luxburg2013density} noted that ``It is impossible to estimate the density in an unweighted $k$-NN graph by local quantities alone.'' As the density estimation can be reduced to embedding recovery, it is also impossible to recover the embedding by local quantities of an unweighted $k$-NN graph. This means that ETP requires to download all item pages from the web service to construct user-side recommender systems, which incurs prohibitive communication costs when the number of items is large. From a practical point of view, it is not necessary to exactly recover the ground truth embeddings. It may be possible to derive useful user-side recommender systems by bypassing the reverse engineering. We will formulate the three desiderata and propose a method that provably satisfies the desiderata in the following sections.

\section{Design Principles} \label{sec: principles}

We formulate three design principles for user-side recommender systems in this section. We consider user-side recommender systems that take a parameter $\tau \in \mathbb{Z}_{\ge 0}$ to control the trade-off between fairness and performance. $\tau = 0$ indicates that the recommender system does not care about fairness, and an increase in $\tau$ should lead to a corresponding increase in fairness.

\vspace{0.1in}
\noindent \textbf{Consistency.} A user-side recommender system $\mathcal{Q}$ is consistent if nDCG of $\mathcal{Q}$ with $\tau = 0$ is guaranteed to be the same as that of the official recommender system. In other words, a consistent user-side recommender system does not degrade the performance if we do not impose the fairness constraint. Sato \cite{sato2022private} showed that \textsc{PrivateRank} is consistent.

\vspace{0.1in}
\noindent \textbf{Soundness.} We say a user-side recommender system is sound if the minimum number of items from each sensitive group is guaranteed to be at least $\tau$ provided $0 \le \tau \le K/|\mathcal{A}|$ and there exist at least $\tau$ items for each group $a \in \mathcal{A}$\footnote{Note that if $\tau > K/|\mathcal{A}|$ or there exist less than $\tau$ items for some group $a$, it is impossible to achieve this condition.}. In other words, we can provably guarantee the fairness of sound user-side recommender system by adjusting $\tau$. Sato \cite{sato2022private} showed that \textsc{PrivateRank} is sound.

\vspace{0.1in}
\noindent \textbf{Locality.} We say a user-side recommender system is local if it generates recommendations without loading the entire recommendation network. This property is crucial for communication costs. As we mentioned above, the ETP approach does not satisfy this property.

\vspace{0.1in}
\noindent \textbf{Inconsistency of \textsc{PrivateWalk}.} The following proposition shows that \textsc{PrivateWalk} is not consistent \footnote{This fact is also observed from the experimental results in \citep{sato2022private}.}

\begin{proposition}
\textsc{PrivateWalk} is not consistent.
\end{proposition}

\begin{proof}
We construct a counterexample where \textsc{PrivateWalk} with $\tau = 0$ generates different recommendations from the provider. Let $\mathcal{I} = [5]$ and $K = 2$,
\begin{align*}
\mathcal{P}_\text{prov}(i)_1 &= ((i + 3) \text{ mod } 5) + 1, \\
\mathcal{P}_\text{prov}(i)_2 &= (i \text{ mod } 5) + 1.
\end{align*}
The recommendations of the provider in item $3$ is $(2, 4)$. Let \textsc{PrivateWalk} select the first item (item $2$) as a result of randomness. Then, \textsc{PrivateWalk} accepts item $2$ and restarts the random walk. Let \textsc{PrivateWalk} select the first item (item $2$) again. \textsc{PrivateWalk} rejects item $2$ and continues the random walk. Let \textsc{PrivateWalk} select the first item (item $1$). Then, \textsc{PrivateWalk} accepts item $1$ and terminates the process. The final recommendation result is $(2, 1)$, which is different from the recommendations of the provider. If item $4$ is the ground truth item, the nDCG of \textsc{PrivateWalk} is less than the official recommender system.
\end{proof}

\begin{table}[tb]
\small
    \caption{Properties of user-side recommender systems. Postprocessing (PP) applies postprocessing directly to the official recommender system, which is not sound when the list does not contain some sensitive groups \citep{sato2022private}.}
    \centering
    \begin{tabular}{lccccc} \toprule
    & PP & \textsc{PrivateRank} & \textsc{PrivateWalk} & ETP & \textsc{CONSUL} \\ \midrule
    Consistent & \cmark & \cmark & \xmark & \cmark & \cmark \\
    Sound & \xmark & \cmark & \cmark & \cmark & \cmark \\
    Local & \cmark & \xmark & \cmark & \xmark & \cmark \\ \bottomrule
    \end{tabular}
    \label{tab: prop}
\end{table}

The properties of user-side recommender systems are summarized in Table \ref{tab: prop}. The existing methods lack at least one properties. We propose a method that satisfies all properties in the next section.

\section{Proposed Method} \label{sec: proposed}

\setlength{\textfloatsep}{5pt}
\begin{algorithm2e}[t]
\caption{\textsc{Consul}}
\label{algo: consul}
\DontPrintSemicolon 
\nl\KwData{Oracle access to $\mathcal{P}_{\text{prov}}$, Source item $i \in \mathcal{I}$, Protected attributes $a_i ~\forall i \in \mathcal{I}$, Minimum requirement $\tau$, Set $\mathcal{H}$ of items that user $i$ has already interacted with, Maximum length $L_{\text{max}}$ of search.}
\nl\KwResult{Recommended items $\mathcal{R} = \{j_k\}_{1 \le k \le K}$.}
    \nl Initialize $\mathcal{R} \leftarrow []$ (empty), $p \leftarrow i$ \;
    \nl $c[a] \leftarrow 0 \quad \forall a \in \mathcal{A}$ \tcp*{counter of sensitive groups}
    \nl $\mathcal{S} \leftarrow \text{Stack}([])$ \tcp*{empty stack}
    \nl \For{\textup{iter} $\gets 1$ \textbf{to} $L_{\text{max}}$}{
    \nl     \While{$p$ \textup{is already visited}}{

    \nl     \If{$|\mathcal{S}| = 0$}{
    \nl         \textbf{goto} line 21 \tcp*{no further items}
            }
    \nl     $p \leftarrow \mathcal{S} \text{.pop\_top}()$  \tcp*{next search node}
            }
    \nl     \For{$k \gets 1$ \textbf{to} $K$}{
    \nl         $j \leftarrow \mathcal{P}_{\text{prov}}(p)_{k}$ \;
    \nl         \If{$j \textup{ \textbf{not} \textbf{in} } \mathcal{R} \cup \mathcal{H}$ \textup{\textbf{and}} $\sum_{a \neq a_j} \max(0, \tau - c[a]) \le K - |\mathcal{R}| - 1$}{
    \nl             \tcc{$j$ can be safely added keeping fairness. Avoid items in $\mathcal{R} \cup \mathcal{H}$.}
    \nl             Push back $j$ to $\mathcal{R}$. \;
    \nl             $c[a_j] \leftarrow c[a_j] + 1$ \;
                }
    \nl         \If{$|\mathcal{R}| = K$}{
    \nl             \textbf{return} $\mathcal{R}$ \tcp*{list is full}
                }
            }
    \nl     \For{$k \gets K$ \textbf{to} $1$}{
    \nl         $\mathcal{S} \text{.push\_top}(\mathcal{P}_{\text{prov}}(p)_{k})$ \tcp*{insert candidates}
            }
        }
    \nl \While{$|\mathcal{R}| < K$}{
    \nl     $j \leftarrow$ \text{Uniform}($\mathcal{I}$) \tcp*{random item}
    \nl     \If{$j \textup{ \textbf{not} \textbf{in} } \mathcal{R} \cup \mathcal{H}$ \textup{\textbf{and}} $\sum_{a \neq a_j} \max(0, \tau - c[a]) \le K - |\mathcal{R}| - 1$}{
    \nl         Push back $j$ to $\mathcal{R}$. \;
    \nl         $c[a_j] \leftarrow c[a_j] + 1$ \;
            }
        }
    \nl \textbf{return} $\mathcal{R}$ \;
\end{algorithm2e}

In this section, we propose \textbf{CON}sistent, \textbf{S}o\textbf{U}nd, and \textbf{L}ocal user-side recommender system, \textsc{Consul}. \textsc{Consul} inherits the basic idea from \textsc{PrivateWalk}, i.e., close items on the recommendation network of the official recommender system should be similar, and \textsc{Consul} retrieves similar items using a walk-based algorithm. However, there are several critical differences between \textsc{Consul} and \textsc{PrivateWalk}. First, \textsc{Consul} runs a deterministic depth-first search instead of random walks. This reduces the variance of the algorithm by removing the stochasticity from the algorithm. Second, \textsc{Consul} uses all recommended items when it visits an item page, whereas \textsc{PrivateWalk} chooses only one of them. This feature improves efficiency by filling the recommendation list with short walk length. Third, \textsc{Consul} continues the search after it finds a new item, whereas \textsc{PrivateWalk} restarts a random walk. When items of the protected group are distant from the source node, \textsc{PrivateWalk} needs to travel long distances many times. By contrast, \textsc{Consul} saves such effort through continuation. Although each modification is small, the combined improvements make a considerable difference to the theoretical properties, as we show in Theorem \ref{thm: consul}, and empirical performances, as we will show in the experiments. We stress that the simplicity of \textsc{Consul} is one of the strengths of our proposed method.

\vspace{0.1in}
\noindent \textbf{Pseudo Code.} Algorithm \ref{algo: consul} shows the pseudo code. Lines 3--5 initialize the variables. In lines 11--18, items included in the recommendation list in item $p$ are added to the list as long as the insertion is safe. The condition $\sum_{a \neq a_j} \max(0, \tau - c[a]) \le K - |\mathcal{R}| - 1$ ensures soundness as we will show in the theorem below. In lines 19--20, the adjacent items are included in the search stack in the descending order. This process continues until $K$ items are found, $L_{\text{max}}$ nodes are visited, or no further items are found. In lines 21--25, the fallback process ensures that $K$ items are recommended. Note that this fallback process is usually skipped because $K$ items are found in the main loop.

\begin{theorem} \label{thm: consul}
\textsc{Consul} is consistent, sound, and local.
\end{theorem}

\begin{figure*}[tb]
\centering
\includegraphics[width=\hsize]{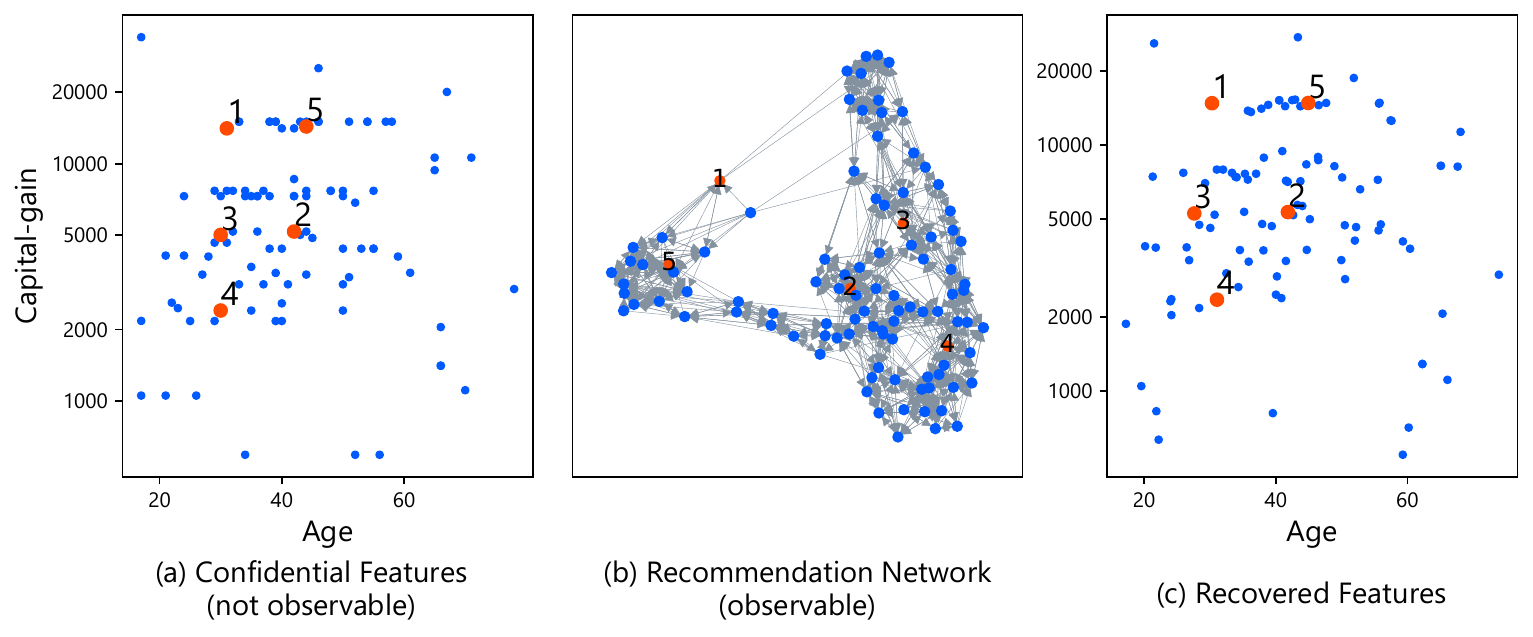}
\caption{\textbf{Feature Reverse Engineering.} \textbf{(Left)} The original features which are confidential. \textbf{(Middle)} The $k$-NN recommendation network revealed to the user. We visualize the graph with \texttt{networkx} package. \textbf{(Right)} The recovered feature solely from the recommendation results. Items $1$ to $5$ are colored red for visibility. All features are accurately recovered. These results show that the raw recommendation results contain sufficient information to build user-side recommender systems.}
\label{fig: reverse}
\end{figure*}

\begin{proof}
\noindent \textbf{Consistency.} If $\tau = 0$, $\sum_{a \neq a_j} \max(0, \tau - c[a]) = 0$ holds in line 13. Therefore, the condition in line 13 passes in all $K$ iterations in the initial node, and the condition in line 17 passes at the $K$-th iteration. The final output is $\mathcal{R} = \mathcal{P}_{\text{prov}}(i)$.

\vspace{0.1in}
\noindent \textbf{Soundness.} We prove by mathematical induction that \begin{align}\sum_{a \in \mathcal{A}} \max(0, \tau - c[a]) \le K - |\mathcal{R}| \label{eq: consul-proof-eq}\end{align} holds in every step in Algorithm \ref{algo: consul}. At the initial state, $|\mathcal{R}| = 0$, $c[a] = 0 ~\forall a \in \mathcal{A}$, and $\sum_{a \in \mathcal{A}} \max(0, \tau - c[a]) = \tau |\mathcal{A}|$. Thus, the inequality holds by the assumption $\tau |\mathcal{A}| \le K$. The only steps that alter the inequality are lines 15--16 and lines 24-25. Let $c, \mathcal{R}$ be the states before the execution of these steps, and $c', \mathcal{R}'$ be the states after the execution. We prove that $\sum_{a \in \mathcal{A}} \max(0, \tau - c'[a]) \le K - |\mathcal{R}'|$ holds assuming \begin{align}\sum_{a \in \mathcal{A}} \max(0, \tau - c[a]) \le K - |\mathcal{R}|, \label{eq: consul-proof-inductive}\end{align} i.e., the inductive hypothesis.  When these steps are executed, the condition \begin{align}\sum_{a \neq a_j} \max(0, \tau - c[a]) \le K - |\mathcal{R}| - 1 \label{eq: consul-proof-if}\end{align} holds by the conditions in lines 13 and 23. We consider two cases. (i) If $c[a_j] \ge \tau$ holds,
\begin{align*}
\sum_{a \in \mathcal{A}} \max(0, \tau - c'[a])
&\stackrel{\text{(a)}}{=} \sum_{a \neq a_j} \max(0, \tau - c'[a]) \\
&\stackrel{\text{(b)}}{=} \sum_{a \neq a_j} \max(0, \tau - c[a]) \\
&\stackrel{\text{(c)}}{\le} K - |\mathcal{R}| - 1 \\
&= K - |\mathcal{R}'|,
\end{align*}
where (a) follows $c[a_j] \ge \tau$, (b) follows $c'[a] = c[a] ~(\forall a \neq a_j)$, and (c) follows eq. \eqref{eq: consul-proof-if}. (ii) If $c[a_j] < \tau$ holds,
\begin{align*}
\sum_{a \in \mathcal{A}} \max(0, \tau - c'[a])
&\stackrel{\text{(a)}}{=} \sum_{a \in \mathcal{A}} \max(0, \tau - c[a]) - 1 \\
&\stackrel{\text{(b)}}{\le} K - |\mathcal{R}| - 1 \\
&= K - |\mathcal{R}'|,
\end{align*}
where (a) follows $c'[a_j] = c[a_j] + 1$ and $c[a_j] + 1 \le \tau$, and (b) follows eq. \eqref{eq: consul-proof-inductive}. In sum, eq. \eqref{eq: consul-proof-eq} holds by mathematical induction. When Algorithm \ref{algo: consul} terminates, $|\mathcal{R}| = K$. As the left hand side of eq. \eqref{eq: consul-proof-eq} is non-negative, each term should be zero. Thus, $c[a] = |\{i \in \mathcal{R} \mid a_i = a\}| \ge \tau$ holds for all $a \in \mathcal{A}$.

\vspace{0.1in}
\noindent \textbf{Locality.} \textsc{Consul} accesses the official recommender system in lines 12 and 20. As the query item $p$ changes at most $L_\text{max}$ times in line 10, \textsc{Consul} accesses at most $L_\text{max}$ items, which is a constant, among $n$ items.
\end{proof}

\vspace{0.1in}
\noindent \textbf{Time complexity.} The time complexity of \textsc{Consul} depends on the average length $L_{\text{ave}}$ of random walks, which is bounded by $L_{\text{max}}$. The number of loops in lines 11--18 is $O(K L_{\text{ave}})$. Although the condition in line 13 involves $|\mathcal{A}|$ terms, it can be evaluated in constant time by storing \[s \stackrel{\text{def}}{=} \sum_{a \in \mathcal{A}} \max(0, \tau - c[a])\] because \[\sum_{a \neq a_j} \max(0, \tau - c[a]) = s - \max(0, \tau - c[a_j]),\] the right hand side of which can be computed in constant time. The number of loops in lines 19-20 is also $O(K L_{\text{ave}})$. Therefore, the main loop in lines 6--20 runs in $O(K L_{\text{ave}})$ time if we assume evaluating $\mathcal{P}_{\text{prov}}$ runs in constant time. Assuming the proportion of each sensitive group is within a constant, the number of loops of the fallback process (lines 21--25) is $O(K |\mathcal{A}|)$ in expectation because the number of trials until the condition in line 23 passes is $O(|\mathcal{A}|)$. Therefore, the overall time complexity is $O(K (|\mathcal{A}| + L_{\text{ave}}))$ and is independent of the number $n$ of items. This is $K$ times faster than \textsc{PrivateWalk}. Besides, in practice, the communication cost is more important than the time complexity. The communication cost of \textsc{Consul} is $L_\text{ave}$, which is $K$ times faster than $O(K L_\text{ave})$ of \textsc{PrivateWalk}. We confirm the efficiency of \textsc{Consul} in experiments.

\section{Experiments} \label{sec: experiments}

We answer the following questions through the experiments.

\begin{itemize}
    \item (RQ1) Can users recover the confidential information solely from the recommendation network?
    \item (RQ2) How good a trade-off between performance and efficiency does \textsc{Consul} strike?
    \item (RQ3) Can user-side recommender systems retrieve novel information?
    \item (RQ4) Does \textsc{Consul} work in the real world?
\end{itemize}

\begin{table*}[t]
    \caption{Performance Comparison. Access denotes the average number of times each method accesses item pages, i.e., the number of queries to the official recommender systems. The less this value is, the more communication-efficient the method is. The best score is highlighted with bold. \textsc{Consul} is extremely more efficient than other methods while it achieves on par or slightly worse performances than Oracle and \textsc{PrivateRank}.}
    \centering
\begin{tabular}{lcc} \toprule
    & \multicolumn{2}{c}{Adult}\\
    \cmidrule(lr{1.0em}){2-3} 
    & Accuracy $\uparrow$ & Access $\downarrow$ \\ \midrule
    Oracle & \textbf{0.788} & $\infty$  \\
    \textsc{PrivateRank} & 0.781 & 39190 \\
    \textsc{PrivateWalk} & 0.762 & 270.5  \\
    \textsc{Consul} & 0.765 & \textbf{34.5} \\ \bottomrule
\end{tabular}
\newline
\vspace*{0.1 in}
\newline
\begin{tabular}{lcccccc} \toprule
    & \multicolumn{3}{c}{MovieLens (oldness)} & \multicolumn{3}{c}{MovieLens (popularity)} \\
    \cmidrule(lr{1.0em}){2-4} \cmidrule(lr{1.0em}){5-7}
    & Accuracy $\uparrow$ & Recall $\uparrow$ & Access $\downarrow$ & nDCG $\uparrow$ & Recall $\uparrow$ & Access $\downarrow$ \\ \midrule
    Oracle & \textbf{0.057} & $\infty$ & \textbf{0.034} & \textbf{0.064} & $\infty$ \\
    \textsc{PrivateRank} & 0.0314 & 0.055 & 1682 & \textbf{0.034} & 0.062 & 1682 \\
    \textsc{PrivateWalk} & 0.0273 & 0.049 & 154.2 & 0.029 & 0.054 & 44.0 \\
    \textsc{Consul} & \textbf{0.0321} & \textbf{0.057} & \textbf{19.6} & 0.033 & 0.060 & \textbf{4.6} \\ \bottomrule
\end{tabular}
\newline
\vspace*{0.1 in}
\newline
\begin{tabular}{lcccccc} \toprule
    & \multicolumn{3}{c}{Amazon} & \multicolumn{3}{c}{LastFM} \\
    \cmidrule(lr{1.0em}){2-4} \cmidrule(lr{1.0em}){5-7}
    & nDCG $\uparrow$ & Recall $\uparrow$ & Access $\downarrow$ & nDCG $\uparrow$ & Recall $\uparrow$ & Access $\downarrow$ \\ \midrule
    Oracle & \textbf{0.0326} & \textbf{0.057} & $\infty$ & \textbf{0.0652} & \textbf{0.111} & $\infty$ \\
    \textsc{PrivateRank} & 0.0325 & \textbf{0.057} & 1171 & 0.0641 & 0.107 & 1507 \\
    \textsc{PrivateWalk} & 0.0217 & 0.048 & 135.7 & 0.0424 & 0.080 & 74.1 \\
    \textsc{Consul} & 0.0310 & 0.052 & \textbf{7.3} & 0.0639 & 0.107 & \textbf{6.5} \\ \bottomrule
\end{tabular}
    \label{tab: performance}
\end{table*}

\subsection{(RQ1) Feature Reverse Engineering} \label{sec: experiments-reverse}

\noindent \textbf{Setup.} We use Adult dataset \citep{dua2017uci}\footnote{\url{https://archive.ics.uci.edu/ml/datasets/adult}} in this experiment. This dataset contains demographic data such as age, sex, race, and income. Although this is originally a census dataset, we use this with talent search in mind. Specifically, we regard a record of a person as an item and construct a provider's recommender system that recommends similar people on each person's page. We use age and capital-gain, which can be confidential information, as features. We remove the highest and lowest values as outliers because they are clipped in the dataset. We take a logarithm for capital-gain because of its high dynamic range. We normalize the features and recommend $K$-nearest neighbor people with respect to the Euclid distance of the features. We recover the features solely from the recommendation results using LOF \citep{terada2014local}.

\vspace{0.1in}
\noindent \textbf{Results.} The left panel of Figure \ref{fig: reverse} visualizes the original data, which are not revealed to us. The middle panel of Figure \ref{fig: reverse} is the observed $k$-NN recommendation network, which is the only information we can use in this experiment. The right panel of Figure \ref{fig: reverse} is the recovered features. We apply the optimal similarity transformation that minimizes the L2 distance to the original embedding for visibility. We also highlight items $1, 2, 3, 4,$ and $5$ in red. These results show that the confidential information is accurately recovered. These results validate that the recommendation network, which is observable by an end-user, contains sufficient information for building user-side recommender systems.

\vspace{0.1in}
\noindent \textbf{Discussion on the Similarity Transformation.} As we discussed in Section \ref{sec: recover}, the similarity transformation factor (e.g., rotation) cannot be determined solely from the $k$-NN graph. When we need to recover exact information, we may be able to determine the similarity transformation factor by registering a few items with known features to the item database, e.g., signing up at the talent search service. As the degree of freedom is at most $6$ dimensions in the case of two features, a few items suffice. As the exact recover is not the original aim of this study, we leave exploring more realistic approaches for recovering the exact information as future work.

\subsection{(RQ2) Performance} \label{sec: experiments-performance}

\noindent \textbf{Setup.} We use Adult dataset, MovieLens100k \citep{harper2016movielens}, Amazon Home and Kitchen \citep{he2016ups, mcauley2015image}, and LastFM \footnote{\url{https://grouplens.org/datasets/hetrec-2011/}} datasets following \citep{sato2022private}.

\noindent \textbf{Adult dataset.} In this dataset, an item represents a person, and the sensitive attribute is defined by sex. We use the nearest neighbor recommendations with demographic features, including age, education, and capital-gain, as the provider's official recommender system. The label of an item (person) represents whether the income exceeds \$50\,000 per year. The accuracy for item $i$ represents the ratio of the recommended items for item $i$ that have the same label as item $i$. The overall accuracy is the average of the accuracy of all items.

\noindent \textbf{MovieLens dataset.} In this dataset, an item represents a movie. We consider two ways of creating protected groups, (i) oldness: We regard movies released before 1990 as a protected group, and (ii) popularity: We regard movies with less than $50$ reviews as the protected group. We use Bayesian personalized ranking (BPR) \citep{rendle2009bpr} for the provider's recommender system, where the similarity of items is defined by the inner product of the latent vectors of the items, and the top-$K$ similar items are recommended. We use the default parameters of Implicit package\footnote{\url{https://github.com/benfred/implicit}} for BPR. We measure nDCG@$K$ and recall@$k$ as performance metrics following previous works \citep{krichene2020sampled, he2017neural, rendle2009bpr, sato2022private}. Note that we use the full datasets to compute nDCG and recall instead of employing negative samples to avoid biased evaluations \citep{krichene2020sampled}.

\noindent \textbf{LastFM and Amazon dataset.} In these datasets, an item represents a music and a product, respectively. We regard items that received less than $50$ interactions as a protected group. We extract $10$-cores for these datasets by iteratively discarding items and users with less that $10$ interactions. We use BPR for the provider's official recommender system, as on the MovieLens dataset. We use nDCG@$K$ and recall@$k$ as performance metrics.

In all datasets, we set $K = 10$ and $\tau = 5$, i.e., recommend $5$ protected items and $5$ other items. Note that all methods, \textsc{Consul} and the baselines, are guaranteed to generate completely balanced recommendations, i.e., they recommend $5$ protected items and $5$ other items, when we set $K = 10$ and $\tau = 5$. We checked that the results were completely balanced. So we do not report the fairness scores but only report performance metrics in this experiment. 

\begin{figure}[tb]
\centering
\includegraphics[width=\hsize]{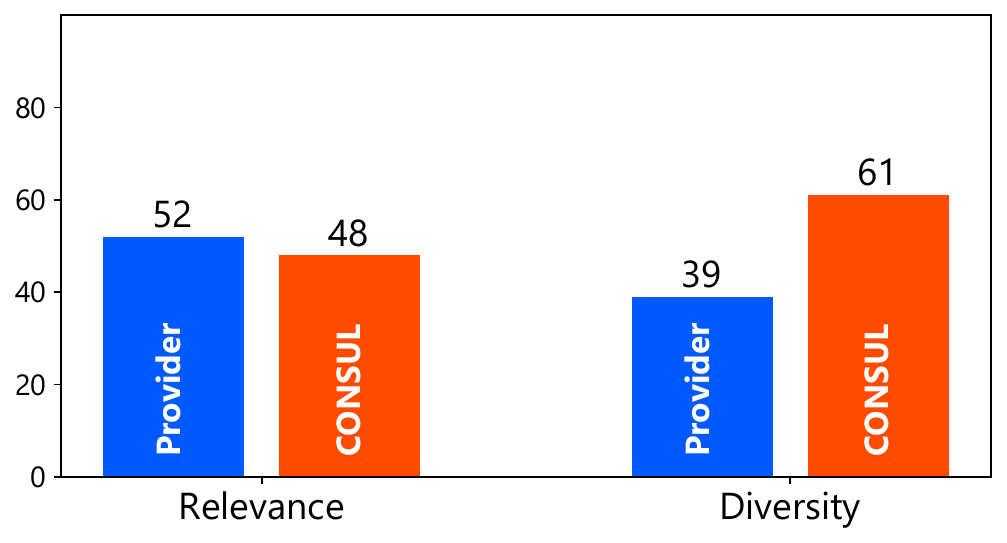}
\caption{\textbf{User Study}. Each value shows the number of times the method is chosen by the crowd workers in $100$ comparisons. Relevance: The crowd workers are asked which recommendation is more relevant to the source item. Diversity: The crowd workers are asked which recommendation is more diverse. \textsc{Consul} is on par or slightly worse than the provider in terms of relevance, while it clearly has more diversity.}
\label{fig: user}
\end{figure}

\vspace{0.1in}
\noindent \textbf{Methods.} We use \textsc{PrivateRank} and \textsc{PrivateWalk} \citep{sato2022private} as baseline methods. We use the default hyperparameters of \textsc{PrivateRank} and \textsc{PrivateWalk} reported in \citep{sato2022private}. Note that Sato \cite{sato2022private} reported that \textsc{PrivateRank} and \textsc{PrivateWalk} were insensitive to the choice of these hyperparameters. In addition, we use the oracle method \citep{sato2022private} that uses the complete similarity scores used in the provider's recommender system and adopts the same fair post-processing as \textsc{PrivateRank}. The oracle method uses hidden information, which is not observable by end-users, and can be seen as an ideal upper bound of performance. We tried to apply ETP but found that LOE \citep{terada2014local} did not finish within a week even on the smallest dataset, MovieLens 100k. As ETP is at most as communication efficient as \textsc{PrivateRank}, we did not pursue the use of ETP further in this experiment.

\vspace{0.1in}
\noindent \textbf{Results.} Table \ref{tab: performance} reports performances and efficiency of each method. First, \textsc{Consul} is the most efficient method and requires $10$ times fewer queries to the official recommender system compared with \textsc{PrivateWalk}. Note that \textsc{PrivateWalk} was initially proposed as an efficient method in \citep{sato2022private}. Nevertheless, \textsc{Consul} further improves communication cost by a large margin. Second, \textsc{Consul} performs on par or slightly worse than Oracle and \textsc{PrivateRank}, whereas \textsc{PrivateWalk} degrades performance in exchange for its efficiency. In sum, \textsc{Consul} strikes an excellent trade-off between performance and communication costs.

\vspace{0.1in}
\noindent \textbf{Discussion on the Real-time Inference.} It is noteworthy that \textsc{Consul} makes recommendations with only $6$ to $7$ accesses to item pages in Amazon and FastFM. This efficiency enables it to recommend items in real-time when the user visits an item page. This property of \textsc{Consul} is important because nowadays many recommendation engines are very frequently updated, if not in real time, and maintaining the recommendation graph is costly. \textsc{Consul} can make a prediction on the fly without managing the recommendation graph offline. By contrast, \textsc{PrivateWalk} requires close to 100 accesses, which prohibit real-time inference and thus requires to manage the recommendation graph offline. This is a clear advantage of \textsc{Consul} compared to the existing methods.

\begin{table}[t]
    \caption{Recommendations for ``Terminator 2: Judgment Day.'' The common items are shown in black, and the differences are shown in blue and red for visibility. In this example, \textsc{Consul} was chosen by the crowd workers in terms of \emph{both} relevance and diversity. \textsc{Consul} retrieves classic sci-fi movies in addition to contemporary sci-fi movies while Provider recommends only contemporary movies.}
    \centering
    \scalebox{0.95}{
    \begin{tabular}{cc} \toprule
        \multicolumn{2}{c}{Terminator 2: Judgment Day (1991)} \\ \midrule
        Provider & \textsc{Consul} \\ \midrule
        Total Recall (1990) & Total Recall (1990) \\
        The Matrix (1990) & The Matrix (1990)  \\
        The Terminator (1990) & The Terminator (1990) \\
        {\color[HTML]{005AFF} Jurassic Park (1993)} & {\color[HTML]{FF4B00} Alien (1979)} \\
        {\color[HTML]{005AFF} Men in Black (1997)} & {\color[HTML]{FF4B00} Star Wars: Episode IV (1977)} \\
        {\color[HTML]{005AFF} The Fugitive (1993)} & {\color[HTML]{FF4B00} Star Trek: The Motion Picture (1979)} \\ \bottomrule
    \end{tabular}
    }
    \label{tab: example}
\end{table}

\subsection{(RQ3) User Study} \label{sec: experiments-user}

We conducted user studies of \textsc{Consul} on Amazon Mechanical Turk.

\vspace{0.1in}
\noindent \textbf{Setup.} We used the MovieLens 1M dataset in this experiment. We chose a movie dataset because crowd workers can easily interpret the results without much technical knowledge \citep{serbos2017fairness}. We choose $100$ movies with the most reviews, i.e., $100$ most popular items, as source items. We use BPR for the provider's recommender system as in the previous experiment. A movie is in a protected group if its release date is more that $10$ years earlier or later than the source items. We set $K = 6$ and $\tau = 3$ in this experiment. We conduct two user studies. We show the crowd workers the two recommendation lists side by side as in Table \ref{tab: example}, and in the first study, we ask the workers which list is more relevant to the source item, and in the second study, we ask them which list is more divergent in terms of release dates. We show the list without any colors or method names, and thus, the workers do not know which method offers which list. To mitigate the position bias, we show the provider's recommender system on the left panel for $50$ trials, and \textsc{Consul} on the left panel for $50$ trials. The two experiments are run in different batches so that one result does not affect the other result.

\vspace{0.1in}
\noindent \textbf{Results.} Figure \ref{fig: user} reports the number of times each method was chosen. This shows that \textsc{Consul} is on par or slightly worse than the provider's recommender system in terms of relevance. Meanwhile, the crowd workers found \textsc{Consul} to be more divergent. This shows that \textsc{Consul} indeed improves fairness (here in terms of release dates) and recommends divergent items. We stress that \textsc{Consul} does not necessarily compensate the relevance to improve fairness. Table \ref{tab: example} shows an actual example where the crowd workers found \textsc{Consul} to be \emph{both} relevant and divergent. \textsc{Consul} succeeded in recommending \emph{classic} sci-fi movies, whereas the provider's recommended only \emph{contemporary} movies, which have indeed high scores from the perspective of collaborative filtering. This result shows that \textsc{Consul} can retrieve information that the provider recommender system cannot.

We note that the official and user-side recommender systems are not alternatives to each other. Rather, a user can consult both official and user-side recommender systems. The diversity of \textsc{Consul} is preferable even if it degrades the relevance to some extent because official and user-side recommender systems can complement the defects of one another.

\subsection{(RQ4) Case Study in the Real World} \label{sec: twitter}

\begin{table}[t]
    \caption{Recommendations for ``Hugh Jackman'' on the real-world Twitter environment. The provider's official system recommends only male users and is not fair with respect to gender. Both \textsc{PrivateWalk} and \textsc{Consul} succeed completely balanced recommendations. All the recommended users by \textsc{Consul} are actors/actresses and related to the source user, Hugh Jackman. Besides, the number of accesses to the official system is much less than in \textsc{PrivateWalk}.}
    \centering
    \scalebox{0.8}{
    \begin{tabular}{ccc} \toprule
        \multicolumn{3}{c}{Hugh Jackman} \\ \midrule
        Provider & \textsc{PrivateWalk} & \textsc{Consul} \\
        N/A & 23 accesses & \textbf{5 accesses} \\ \midrule
        Chris Hemsworth ({\color[HTML]{FF4B00} man}) & Ian McKellen ({\color[HTML]{FF4B00} man}) & Chris Hemsworth ({\color[HTML]{FF4B00} man}) \\
        Chris Pratt ({\color[HTML]{FF4B00} man}) & Zac Efron ({\color[HTML]{FF4B00} man}) & Chris Pratt ({\color[HTML]{FF4B00} man})  \\
        Ian McKellen ({\color[HTML]{FF4B00} man}) & Seth Green ({\color[HTML]{FF4B00} man}) & Ian McKellen ({\color[HTML]{FF4B00} man}) \\
        Zac Efron ({\color[HTML]{FF4B00} man}) & Dana Bash ({\color[HTML]{005AFF} woman}) & Brie Larson ({\color[HTML]{005AFF} woman}) \\
        Patrick Stewart ({\color[HTML]{FF4B00} man}) & Lena Dunham ({\color[HTML]{005AFF} woman}) & Danai Gurira ({\color[HTML]{005AFF} woman}) \\
        Seth Rogen ({\color[HTML]{FF4B00} man}) & Jena Malone ({\color[HTML]{005AFF} woman}) & Kat Dennings ({\color[HTML]{005AFF} woman}) \\ \bottomrule
    \end{tabular}
    }
    \label{tab: twitter}
\end{table}

Finally, we show that \textsc{Consul} is applicable to real-world services via a case study. We use the user recommender system on Twitter, which is in operation in the real-world. We stress that we are not employees of Twitter and do not have any access to the hidden data stored in Twitter. The results below show that we can build a recommender system for Twitter even though we are not employees but ordinary users. The sensitive attribute is defined by gender in this experiment. Table \ref{tab: twitter} shows the results for the account of ``Hugh Jackman.'' While the official recommender system is not fair with respect to gender, \textsc{Consul}'s recommendations are completely balanced, and \textsc{Consul} consumes only $5$ queries, which is $4$ times more efficient than \textsc{PrivateWalk}. Note that ETP and \textsc{PrivateRank} require to download all users in Twitter, and thus they are infeasible for this setting. These results highlight the utility of \textsc{Consul} in real-world environments.

\section{Related Work}

\noindent \textbf{User-side Information Retrieval.} User-side (or client-side) information retrieval \citep{sato2022private, sato2022retrieving, sato2022clear} aims to help users of web services construct their own information retrieval and recommender systems, whereas traditional methods are designed for the developers of the web services. This setting has many unique challenges including limited resources and limited access to databases. Our study focuses on user-side recommender systems, a special case of user-side information retrieval problems. User-side information retrieval systems are closely related to focused crawling  \citep{chakrabarti1999focused, mccallum2000automating, johnson2003evolving, baezayates2005crawling, guan2008guide}, which aims to retrieve information of specific topics \citep{chakrabarti1999focused, mccallum2000automating}, popular pages \citep{baezayates2005crawling}, structured data \citep{meusel2014focused}, and hidden pages \citep{barbosa2007adaptive}. The main difference between focused crawling and user-side information retrieval is that focused crawling consumes a great deal of time, typically several hours to several weeks, whereas user-side information retrieval aims to conduct real-time inference. Our proposed method is highly efficient as shown in the experiments and enables inference even when costly crawling is prohibitive.

\vspace{0.1in}
\noindent \textbf{Fairness in Recommender Systems.} As fairness has become a major concern in society \citep{united2014big, executive2016big}, many fairness-aware machine learning algorithms have been proposed \citep{hardt2016equality, kamishima2012fairness, zafar2017fairness}. In particular, fairness with respect to gender \citep{zehlike2017fair, singh2018fairness, xu2020algorithmic}, race \citep{zehlike2017fair, xu2020algorithmic}, financial status \citep{fu2020fairness}, and popularity \citep{mehrotra2018towards, xiao2019beyond} is of great concern. In light of this, many fairness-aware recommendation algorithms have been proposed \citep{kamishima2012enhancement, yao2017beyond, biega2018equity, milano2020recommender, sato2022enumerating}. Fair recommender systems can be categorized into three groups \citep{bruke2017mltisided}. C-fairness concerns fairness for users; it, for example, ensures that female and male users should be treated equally. P-fairness concerns fairness for items; it, for example, ensures that news related to Republicans and Democrats are treated equally. CP-fairness is the combination of the two. In this study, we focus on P-fairness following \citep{sato2022private}. A notable application of P-fair recommender systems is in the job-market \citep{geyik2019fairness, geyik2018building}, where items correspond to job-seekers, and the sensitivity attributes correspond to gender and race. In contrast to traditional P-fair recommender systems \citep{ekstrand2018exploring, beutel2019fairness, mehrotra2018towards, liu2019personalized}, which are designed for the developers of the services, our proposed method is special in that it is designed for the users of the services. Note that fairness is closely related to topic diversification \citep{ziegler2005improving} by regarding the topic as the sensitive attribute, and we consider the diversity of recommended items in this study as well.

\vspace{0.1in}
\noindent \textbf{Steerable Recommender Systems.} The reliability of recommender systems has attracted a lot of attention \citep{tintarev2007survey, balog2019transparent}, and steerable recommender systems that let the users modify the behavior of the system have been proposed \citep{green2009generating, balog2019transparent}. User-side recommender systems also allow the users modify the recommendation results. However, the crucial difference between steerable and user-side recommender systems are that steerable recommender systems must be implemented by a service provider, whereas user-side recommender systems can be built by arbitrary users even if the official system is an ordinary (non-steerable) one. Therefore, user-side recommender systems can expand the scope of steerable recommender systems by a considerable margin \citep{sato2022private}.

\vspace{0.1in}
\noindent \textbf{Metric Recovery.} We utilize the results on metric recovery when we derive the feasibility of user-side recommender systems. Recovering hidden representations has a long history, including the well-known multi-dimensional scaling \citep{kruskal1964multidimensional, agarwal2007generalized} and Laplacian Eigenmaps \citep{belkin2003laplacian}. In particular, our discussions rely on the metric recovery theory on unweighted $k$-NN graphs \citep{alamgir2012shortest, hashimoto2015metric, luxburg2013density, terada2014local}. To the best of our knowledge, this work is the first to connect the theory to security, fairness, or user-side aspects of recommender systems.

\section{Conclusion}

We first saw that interested users can estimate hidden features of items solely from recommendation results. This theory indicates that the raw recommendation results contain sufficient information to build user-side recommender systems and elucidates the rigorous feasibility of building user-side recommender systems without using log data. However, this approach is not practical due to its high communication costs. To design practical user-side recommender systems, we proposed three desirable properties of user-side recommender systems: consistency, soundness, and locality. We found that existing user-side recommender systems lack at least one property. We, therefore, proposed \textsc{Consul}, a user-side recommender system that satisfies the three desired properties. We then demonstrated empirically that confidential features can indeed be recovered solely from recommendation results. We also confirmed that our proposed method is much more communication efficient than the existing methods while retaining a high degree of accuracy. Finally, we conducted user studies with crowd workers and confirmed that our proposed method provided diverse items compared with provider's official recommender systems.

\chapter{Retrieving Black-box Optimal Images from External Databases}


\begin{table*}[tb]
    \centering
    \caption{Symbols, Definitions, and Examples.}
    \scalebox{0.6}{
    \begin{tabular}{lll} \toprule
        \textbf{Symbol} & \textbf{Definition} & \textbf{Example} \\ \midrule
        $\mathcal{X}$ & The set of images in the image database & The set of all images uploaded on Flickr \\
        $\mathcal{T}$ & The set of tags & The set of all tags on Flickr \\
        $f\colon \mathcal{X} \to \mathbb{R}$ & The objective black-box function & A deep neural network that computes a preference score of an image \\
        $\mathcal{O}\colon \mathcal{T} \to \mathcal{X} \times 2^{\mathcal{T}}$ & The oracle for image search & A wrapper of \texttt{flickr.photos.search} API \\
        $\mathcal{T}_\text{ini} \subseteq \mathcal{T}$ & Initially known tags & $100$ popular tags obtained by \texttt{flickr.tags.getHotList} API \\
        $B \in \mathbb{Z}_+$ & A budget of oracle calls & $B = 500$ \\ \bottomrule
    \end{tabular}
    }
    \label{tab: symbol}
\end{table*}

\section{Introduction}

As the amount of information on the Internet continues to drastically increase, information retrieval algorithms are playing more important roles. In a typical situation, a user of a system issues a query by specifying keywords, and an information retrieval algorithm retrieves the optimal items with respect to the query words. Here, the retrieval algorithm is designed by the service provider, not by the user. The uses of information systems have become divergent, and various retrieval algorithms have therefore been proposed, e.g., a cross-modal image search \cite{cao2016deep, kordan2018deep}, complex query retrieval \cite{nie2012harvesting}, and conversational recommender systems \citep{sun2018conversational, christakopoulou2016towards}.

Deep neural networks have achieved state-of-the-art performances in computer vision tasks \cite{krizhevsky2012imagenet, he2016deep}, notably image retrieval \cite{babenko2015aggregating, gordo2016deep, bell2015learning, niu2018neural}. In a conventional setting of image retrieval, algorithms assume that they have full access to the image database. A straightforward method under this setting is to evaluate all images in the database and return the one with the maximum score. When the image database is extremely large, two-step methods are used, i.e., a handful of images are retrieved through fast retrieval algorithms such as a nearest neighbor search, and the results are ranked using sophisticated algorithms. Hash coding further improves the effectiveness and efficiency \cite{lai2015simultaneous, liu2016deep}.

In this study, we consider information retrieval under a different scenario. Whereas most existing studies have focused on how a service provider can improve the search algorithms, we focus on how a user of a service can effectively exploit the search results. Specifically, we consider a user of a service builds their own scoring function. The examples of the scoring functions are as follows.

\noindent \textbf{Example 1 - Favorite Image Retrieval:} A user trains deep neural networks using a collection of favorite images found in different services and wants to retrieve images with similar properties in a new service.

\noindent \textbf{Example 2 - Similar Image Retrieval:} Deep convolutional neural networks are known to have a superior ability to extract useful image features \cite{babenko2015aggregating, gatys2016image, bell2015learning}. Although some online services provide similar image search engines, users do not have full control of the search. For example, even if the service provides a similar texture-based image search engine, some users may want to retrieve similar images based on the semantics. The user-defined score function allows image searches at a finer granularity.

\noindent \textbf{Example 3 - Fair Image Retrieval:} A search engine can be unfair to some protected attributes. For example, when we search for images through a query ``president,'' an image search engine may retrieve only male president images \cite{singh2018fairness}. Some users may want to use their own scoring function, e.g., scoring male and female images equally.

Recent advancements in deep learning, such as self-supervised contrastive learning \cite{chen2020simple, he2020momentum} and meta-learning \cite{finn2017model}, enable the training of deep neural networks with a few labeled samples. In addition, many pre-trained models for various vision tasks have been released on the Internet. Machine learning as a service platform, such as Microsoft Azure Machine Learning Studio and AWS Machine Learning, has also made it easy to build a machine learning model. These techniques and services enable an individual to easily build their own black-box scoring function. Suppose a user has already built a black-box function she wants to optimize. How can she retrieve optimal images with respect to this function from an image database on the Internet? 

Under this setting, an individual cannot evaluate all images in the image database because it contains a significant number of images. It is also impossible for an individual to build a search index (e.g., a hash index) because of both limited access to the database and the insufficient computational resources of the individual. These limitations prohibit the use of standard image retrieval algorithms.

In this chapter, we assume that a user can access the database through a search query alone and that a tight query budget exists. In addition, we assume little knowledge about the database to apply our method to a new environment. We formulate the problem through the lens of the multi-armed bandit problem and propose a query efficient algorithm, \textsc{Tiara}, with the aid of pre-trained representative word embeddings. We confirm the effectiveness of the proposed method under various settings, including the online Flickr environment.

It is noteworthy that Bayesian optimization and optimization on deep neural networks \cite{erhan2009visualizing, simonyan2014deep, nguyen2016synthesizing} also aim to optimize black-box or complex functions. However, these methods assume a continuous and simple optimization domain, typically the entire Euclidean space $\mathbb{R}^d$, a hypercube $[0, 1]^d$, or a unit ball $\{x \mid \|x\| \le 1\}$. By contrast, we aim to \emph{retrieve} optimal images from the \emph{fixed database}. Therefore, the optimization domain is a discrete set of images. Although there are several methods for discrete black-box optimization \cite{baptista2018baysian, oh2019combinatorial}, they also assume that the optimization domain is simple, e.g., $\{0, 1\}^d$, or at least they have full access to the optimization domain. Under our setting, an algorithm does not even know the entire optimization domain, i.e., the image database. This limitation causes numerous challenges, as we describe in the following sections.

The contributions of this chapter are summarized as follows:

\begin{itemize}
    \item We formulate the black-box optimal image retrieval problem. This problem examines how a user of an online service can effectively exploit a search engine, whereas most existing studies focus on how a service provider can improve a search engine.
    \item We propose \textsc{Tiara}, an effective method for this problem for the first time. \textsc{Tiara} is a general algorithm that works in various situations and retrieves optimal black-box images with few API queries.
    \item We investigate the effectiveness of \textsc{Tiara} using many real-world data. Notably, we conduct ``in the wild'' experiments on the real-world Flickr environment and confirm that \textsc{Tiara} can be readily used in real-world applications.
\end{itemize}

\noindent \textbf{Reproducibility:} Our code is available at \url{https://github.com/joisino/tiara}.

\section{Background}

\subsection{Notations}

Let $\mathcal{X}$ be a set of images in the image database. We do not know the exact set of $\mathcal{X}$. Let $f: \mathcal{X} \to \mathbb{R}$ be a black-box function that evaluates the value of an image. For example, $f$ measures the preference of the user in favorite image retrieval and measures the similarity in similar image retrieval. The notations are summarized in Table \ref{tab: symbol}.

\subsection{Problem Formulation}

In this section, we formulate the problem of black-box optimal image retrieval. We observe that many image databases, such as Flickr and Instagram, support (hash) tag-based search. We assume that we can search for images by specifying a tag through an API. For example, in the Flickr case, we use the \texttt{flickr.photos.search} API.

For a formal discussion, we generalize the function of tag search APIs. Let $\mathcal{T} \subset \Sigma^*$ be the set of tags, where $\Sigma^*$ is the set of strings. We formalize a tag search API as a (randomized) oracle $\mathcal{O}$ that takes a tag as input and returns an image with the tag. We assume that $\mathcal{O}$ always returns different images even when we query the same tag twice. In addition to the image itself, we assume that the oracle $\mathcal{O}$ returns the set of tags of the image. Therefore, for any tag $t \in \mathcal{T}$, $\mathcal{O}(t) \in \mathcal{X} \times 2^\mathcal{T}$. Let $\mathcal{O}(t).\texttt{image} \in \mathcal{X}$ denote the returned image and $\mathcal{O}(t).\texttt{tags} \in 2^\mathcal{T}$ denote the returned tags. As the oracle returns an image with the query tag, $t \in \mathcal{O}(t).\texttt{tags}$ always holds.

We find that myriad tags exist in real-world services, and we cannot know the entire tag set. Therefore, we assume that we know only a fraction $\mathcal{T}_\text{ini}$ of the tag set. Here, $\mathcal{T}_\text{ini}$ can be constructed by browsing the online service or retrieving popular tags through an API, e.g., \texttt{flickr.tags.getHotList}, in the Flickr example.

We assume that there is a budget $B \in \mathbb{Z}_+$ for the oracle call. For example, in the Flickr case, There is an API rate limit of $3600$ calls per hour. Thus, it is natural to assume that we can use at most $3600$ API calls in one task. If we retrieve many images within a short period of time, the API limits become tighter for each retrieval. To summarize, the black-box optimal image retrieval problem is formalized as follows.

\begin{tcolorbox}[colframe=gray!20,colback=gray!20,sharp corners]
\textbf{Black-box optimal image retrieval.}

\textbf{Given:}
\begin{itemize}
    \item Black box function $f\colon \mathcal{X} \to \mathbb{R}$
    \item Oracle $\mathcal{O}\colon \mathcal{T} \to \mathcal{X} \times 2^{\mathcal{T}}$
    \item Known tags $\mathcal{T}_\text{ini} \subseteq \mathcal{T}$
    \item Budget $B \in \mathbb{Z}_+$
\end{itemize}

\textbf{Goal:} Find an image $x \in \mathcal{X}$ in the image database with as high an $f(x)$ as possible within $B$ accesses to the oracle.
\end{tcolorbox}

\section{Proposed Method}

In this section, We introduce our proposed method, \textsc{Tiara} (\underline{T}ag-based \underline{i}m\underline{a}ge \underline{r}etriev\underline{a}l).

\subsection{Bandit Formulation} \label{sec: bandit}

We first propose regarding the black-box optimal image retrieval problem as a multi-armed bandit problem \cite{lattimore2020bandit, slivkins2019introduction}. Specifically, we regard a tag as an arm, the budget $B$ as the time horizon, and the objective value $f(\mathcal{O}(t).\texttt{image})$ as the reward when we choose arm $t$. This formulation enables us to use off-the-shelf multi-armed bandit algorithms, such as UCB \cite{auer2002finite}, $\varepsilon$-greedy, and Thompson sampling \cite{thompson1933likelihood}. This formulation is the basics of our proposed algorithm. However, there are several challenges to a black-box optimal image retrieval problem. First, we do not initially know all arms but only a fraction $\mathcal{T}_\text{ini}$ of arms. Considering $\mathcal{T}_\text{ini}$ alone leads to suboptimal results because $\mathcal{T}_\text{ini}$ does not contain the best arm in general. Second, myriad arms exist, and the number $|\mathcal{T}|$ of arms is larger than budget $B$ in practice. Standard multi-armed bandit algorithms first explore all arms once. However, they are unable to even finish this initial exploration phase under a tight budget constraint.

The first problem is relatively easy to solve. We obtain tags $\mathcal{O}(t).\texttt{tags}$ when we choose arm $t$. These tags may contain new tags. We can add such tags into the known tag set and gradually grow the known set. A good bandit algorithm will choose relevant tags, and the returned tags will contain relevant tags. For example, suppose that the black-box function $f$ prefers cat images. A good bandit algorithm will choose ``cat'' and ``animal'' tags and obtain cat images accompanied by many cat-related tags. Even if the budget is so tight that we cannot know all tags within the budget, a good bandit algorithm ignores irrelevant tags and prioritizes the collection of many relevant tags. The second problem is essential and difficult to solve. We investigate how to improve the query efficiency using tag embeddings in the following sections.

\vspace{0.1in}
\noindent \textbf{Discussion on the max-bandit problem.} It should be noted that standard bandit algorithms aim to maximize a cumulative reward or regret, whereas the black-box optimal image retrieval problem aims to find an image with the maximum value. These goals are not exactly the same. The problem of obtaining a maximum reward is known as a max $K$-armed bandit \cite{cicirello2005max} or extreme bandit \cite{carpentier2014extreme}. However, we adhere to the standard bandit setting in this study for the following reasons: First, we assume that the reward distribution is Gaussian-like, whereas existing max-$K$-armed bandit algorithms \cite{cicirello2005max, carpentier2014extreme} mainly assume that rewards are drawn from extreme value distributions, such as the GEV distribution. We find that this is not the case in the applications considered herein, i.e., image retrieval. Second, under our setting, an arm with a high expected value often leads to a high maximum value because such an arm usually represents the concept captured by the black-box function $f$. By contrast, max-bandit algorithms focus on detecting heavy tail arms with possibly low average rewards. We find that this is not necessary for our setting. Finally, we experimentally confirm that our proposed method already shows a superior empirical performance even though it does not directly maximize the maximum objective value. We leave the exploration of the max-bandit algorithms in our setting for future work.

\subsection{Tag Embedding} \label{sec: embedding}

\setlength{\textfloatsep}{5pt}
\begin{algorithm2e}[t]
\caption{\textsc{Tiara}}
\label{algo}
\DontPrintSemicolon 
\nl\KwData{Known tags $\mathcal{T}_\text{ini} \subseteq \mathcal{T}$, Black box function $f\colon \mathcal{X} \to \mathbb{R}$, Budget $B \in \mathbb{Z}_+$, Oracle $\mathcal{O}\colon \mathcal{T} \to \mathcal{X} \times 2^{\mathcal{T}}$, Tag Embedding $\{\boldv_t \in \mathbb{R}^d \mid t \in \mathcal{T}\}$, Regularization coefficient $\lambda \in \mathbb{R}_+$, Exploration coefficient $\alpha \in \mathbb{R}_+$.}
\nl\KwResult{An image $x \in \mathcal{X}$ in the image database with as high $f(x)$ as possible.}
    \nl $\boldA \leftarrow \lambda \boldI_{d \times d}$ \tcp*{Initialize $\boldA$}
    \nl $\boldb \leftarrow \mathbf{0}_{d}$  \tcp*{Initialize $\boldb$}
    \nl $\mathcal{T}_\text{known} \leftarrow \mathcal{T}_\text{ini}$  \tcp*{Initialize with the initial tags}
    \nl \For{$i \gets 1$ \KwTo $B$}{
    \nl     $s_t \leftarrow \boldv_t^\top \boldA^{-1} \boldb + \alpha \sqrt{\boldv_t^\top \boldA^{-1} \boldv_t} \quad \forall t \in \mathcal{T}_\text{known}$ \tcp*{Scores}
    \nl     $t_i \leftarrow \argmax_{t \in \mathcal{T}_\text{known}} s_t$ \tcp*{Choose the best arm}
    \nl     $r_i \leftarrow \mathcal{O}(t_i)$ \tcp*{Issue an query}
    \nl     \For{$t \in r_i.\texttt{tags}$}{
    \nl         $\boldA \leftarrow \boldA + \boldv_t \boldv_t^\top$ \tcp*{Update $\boldA$}
    \nl         $\boldb \leftarrow \boldb + f(r_i.\texttt{image}) \cdot \boldv_t$ \tcp*{Update $\boldb$}
    \nl         $\mathcal{T}_\text{known} \leftarrow \mathcal{T}_\text{known} \cup \{t\}$ \tcp*{Insert $t$ to $\mathcal{T}_\text{known}$}
            }
        }
    \nl \textbf{return} $\argmax_{x \in \{r_i.\texttt{image} \mid i = 1, \cdots, B\}} f(x)$ \tcp*{Best image}
\end{algorithm2e}

Because we cannot choose each arm even once on average, we need to estimate the value of each arm without observing the reward from it. We estimate the value of one arm from the rewards obtained from the other arms. The key is how to define the similarities of the arms. As the challenge here, we assume a new image database, and owing to the tight API limit, it is difficult to learn the similarities from the current environment in an online manner. To tackle this problem, we utilize external resources. We find that a tag in an image database is usually described through natural language and typically is a word or composition of words. We use a pre-trained word embedding \cite{mikolov2013efficient, pennington2014glove} to define the similarities between tags. Specifically, we first decompose a tag into a bag of words by non-alphabetical symbols, e.g., a white space.
We use the mean of the word embeddings in the bag of words as the tag embedding.

Owing to the tag embeddings, we can infer the value of an arm from similar arms. For example, if the reward from a ``cat'' tag is high (resp. low), we can assume the black-box function is highly (resp. rarely) related to cats, and we can infer that similar tags with similar embeddings, such as ``Aegean cat'' and ``animal,'' are also valuable (resp. irrelevant) without actually querying them.

\subsection{Tiara}

Our proposed method combines the aforementioned tag embeddings with a bandit algorithm. Specifically, we utilize LinUCB \cite{li2010contextual}. Because there are no contextual features, we use only features of arms. We define the feature of an arm as the tag embedding introduced in Section \ref{sec: embedding} and apply LinUCB to this feature. We call this variant \textsc{Tiara}-S, where S stands for ``single'' and ``simple.'' However, we found that there are too many tags, and LinUCB is still inefficient for learning a reward because of the tight query budget and limited training samples. To improve the query efficiency, we use another signal from the oracle. The oracle $\mathcal{O}$ returns not only an image but also the tags of the returned image. We assume that these tags have similar average rewards and add these tags into the training dataset. Specifically, when we query tag $t \in \mathcal{T}$, we use $\{(s, f(\mathcal{O}(t).\texttt{image})) \mid s \in \mathcal{O}(t).\texttt{tags}\}$ as training data, whereas \textsc{Tiara}-S uses only $\{(t, f(\mathcal{O}(t).\texttt{image}))\}$. As we will see in the experiment, this technique significantly improves query efficiency and performance. Algorithm \ref{algo} shows the pseudo-code of \textsc{Tiara}.

\vspace{0.1in}
\noindent \textbf{Time Complexity.} The bottleneck of the computation is in Line 7. The inverse matrix $\boldA^{-1}$ can be efficiently computed in an iterative manner using the Sherman–Morrison formula, i.e.,
\[ (\boldA + \boldv_t \boldv_t^\top)^{-1} = \boldA^{-1} - \frac{\boldA^{-1} \boldv_t \boldv_t^\top \boldA^{-1}}{1 + \boldv_t^\top \boldA^{-1} \boldv_t}. \]
This technique reduces the cubic dependence on the number of dimensions to the quadratic dependence. Therefore, the computation is in Line 7 takes $O(|\mathcal{T}_\text{known}| d^2)$ time. Let $T_\text{max}$ be the maximum number of tags of an image. The value of $T_\text{max}$ is typically $10$ to $100$. Because $|\mathcal{T}_\text{known}|$ increases by at most $T_\text{max}$ in an iteration, the total time complexity is $O((|\mathcal{T}_\text{ini}| + B T_\text{max}) B d^2)$ in the worst case. In our problem setting, $B$ is small (e.g., in the hundreds) because of the tight API limitation, and $|\mathcal{T}_\text{ini}| \approx 100$ and $d \approx 300$ are also small during the experiments. Therefore, the oracle calls in Line 9 become a bottleneck in the wall-clock time because it requires communication to the Internet. As the oracle calls are common in all methods, \textsc{Tiara} is sufficiently efficient. When the black-box function is complex, its evaluation can be a computational bottleneck as well. \textsc{Tiara} evaluates $f$ as few as $B$ times, which is the same as with other baseline methods. In addition, when the efficiency is insufficient, we can speed up \textsc{Tiara} through a lazy variance update \cite{desautels2012parallelizing} and by applying sub-sampling heuristics.\footnote{We do not adopt such techniques in the experiments.}

\vspace{0.1in}
\noindent \textbf{Discussion on graph-feedback bandits.} There are several multi-armed bandit algorithms with a so-called feedback graph setting \cite{mannor2011from, alon2013from, kocak2014efficient}. This is an intermediate setting between the bandit and full feedback. Specifically, it assumes that there is an underlying graph, where a node represents an arm, and when we choose arm $t$, we can also observe the rewards of the neighboring arms. The underlying graph can be directed and time-varying. Therefore, if we define the neighbor of $t$ as $\mathcal{O}(t).\texttt{tags}$, our problem is seen as a variant of the feedback graph setting. However, we do not employ the feedback graph framework for the following reasons: First, existing graph feedback bandit algorithms require the number of neighbors of each neighbor \cite{mannor2011from, alon2013from, kocak2014efficient}, i.e., $|f(\mathcal{O}(s).\texttt{tags})|, \forall s \in f(\mathcal{O}(t).\texttt{tags})$ in our setting. We need additional API queries to compute these values. Such additional queries are prohibitive because the API limit is tight under our setting. Second, graph feedback bandit algorithms assume that the reward feedback of neighboring arms is an unbiased estimate of the true rewards \cite{mannor2011from, alon2013from, kocak2014efficient}. However, $f(\mathcal{O}(t).\texttt{image})$ is not unbiased for arm $s \in \mathcal{O}(t).\texttt{tags}$. Therefore, we cannot enjoy the theoretical guarantees of the graph feedback bandits. Third, the improvement obtained by the feedback graph framework is on the order of $O(\sqrt{\alpha/n})$ \cite{kocak2014efficient}, where $n$ is the number of arms, and $\alpha$ is the size of the maximum independent set of the feedback graph. Under our setting, the sizes of the feedback graphs and $\alpha$ are still large in practice. How to leverage the existing graph-feedback bandit algorithms for our setting is an interesting area of future study.

\vspace{0.1in}
\noindent \textbf{Discussion on the reward estimation model.} \textsc{Tiara} uses LinUCB and a linear model for estimating the reward from the feature vector. In general, we can use any model for prediction. We use a linear model owing to the following reasons: First, complex models are difficult to train with a tight sample budget. Second, LinUCB empirically performs quite well under various tasks \cite{li2010contextual, krishnamurthy2016contextual} regardless of its suboptimal theoretical regret. Third, we use the average word embeddings as the feature vector of an arm. It has been confirmed that word embeddings are representative and that the weighted average of the word embeddings with a cosine similarity or liner models produce superior performances in many NLP tasks \cite{arora2017simple, shen2018baseline} and sometimes outperform even neural network-based methods.

\subsection{Visualization and Interpretation}

Although not our original goal, \textsc{Tiara} provides an interpretation of the black-box function $f$ as a byproduct. 

Deep neural networks suffer from interpretability issues for reliable decision making. There have been many interpretation methods developed for deep neural networks \cite{simonyan2014deep, sundararajan2017axiomatic}. We consider model-level interpretability \cite{simonyan2014deep, nguyen2016synthesizing, yuan2020xgnn}, i.e., interpreting what function each model represents. It is unclear what function each model represents by simply looking at the model parameters of deep models. Even if the model is prepared by the user, deep models sometimes behave in unexpected ways \cite{goodfellow2015explaining, madry2018towards}. Input instances that produce high values can be regarded as representations of the model \cite{simonyan2014deep, nguyen2016synthesizing, yuan2020xgnn}. \textsc{Tiara} can retrieve such images from external image databases. Compared to methods that use a fixed dataset, the search space of \textsc{Tiara} is extremely large. Therefore, there are more chances that relevant images will be found. In addition, tag scores of \textsc{Tiara} also provide another interpretation of the black-box function through words. We show that these tags are also beneficial for interpretability by visualizing word clouds in the experiments. The tag-based interpretability is beneficial for further exploration as well. When the result is unsatisfactory, the user can continue manual exploration from the tags with high scores.

\section{Experiments}

\begin{figure*}[t]
\centering
\includegraphics[width=\hsize]{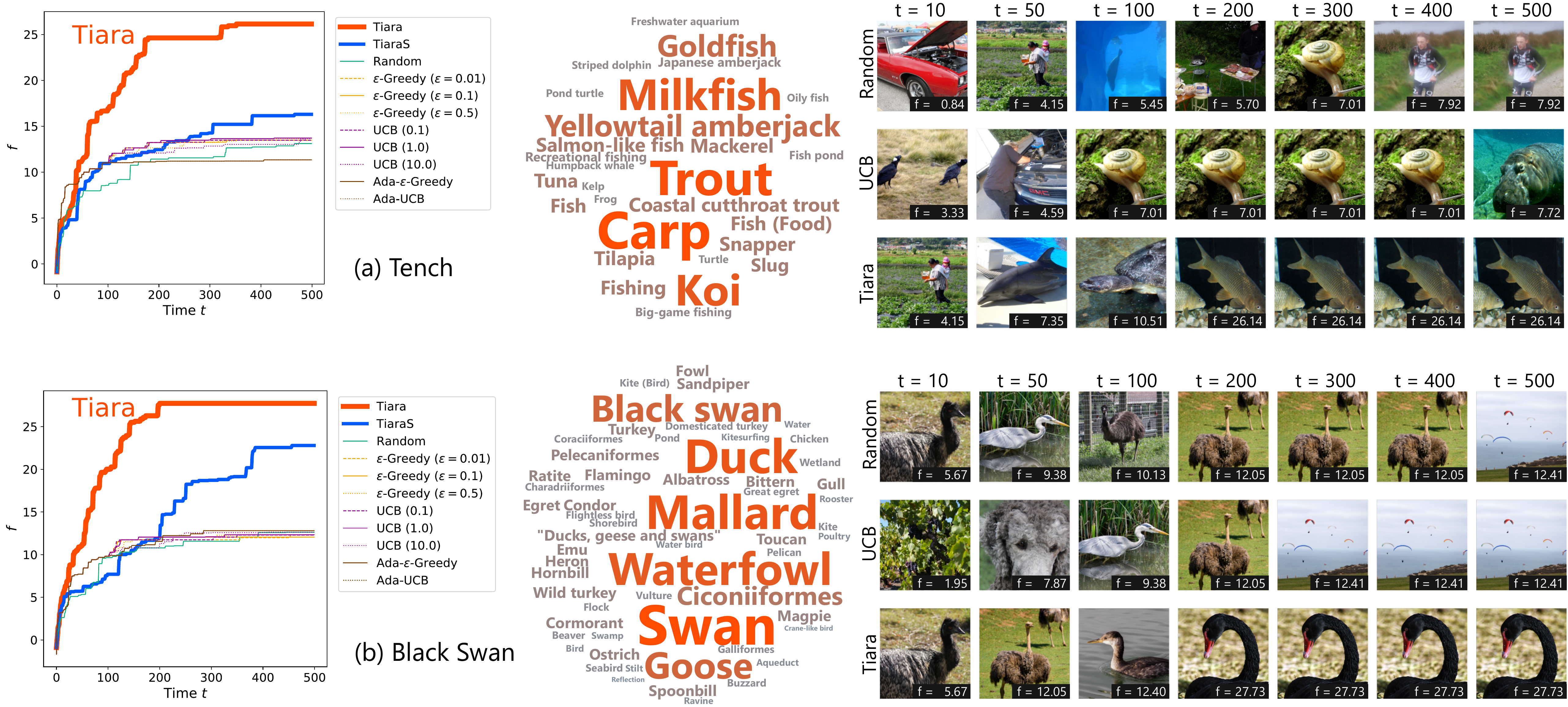}
\caption{Open Image Environment. (Top) Tench, (Bottom) Black Swan. (Left) Learning curves averaged over $10$ independent runs. (Mid) Word clouds that represent the scores computed by \textsc{Tiara} at the end of the last iteration. These visualizations provide an interprtability of the black-box function. (Right) Best images found at each iteration.}
\label{fig: openimage-tiara}
\end{figure*}

\begin{table*}[t]
    \centering
    \caption{Open Image Environment. Each score is the highest $f(x)$ found by an algorithm. The scores are averaged over $10$ independent runs, and the standard deviations are also reported. The highest score is highlighted by bold in each column.}
    \scalebox{0.65}{
\begin{tabular}{lccccc} \toprule
Method $\backslash$ Class & Tench & Black Swan & Tibetan Terrier & Tiger Beetle & Academic Gown \\ \midrule
Random & 13.11 $\pm$ 5.95 & 12.62 $\pm$ 4.83 & 13.16 $\pm$ 1.92 & 8.21 $\pm$ 1.85 & 12.63 $\pm$ 4.88 \\ 
$\epsilon$-Greedy ($\epsilon = 0.01$) & 13.48 $\pm$ 6.27 & 12.30 $\pm$ 4.94 & 13.76 $\pm$ 2.16 & 7.47 $\pm$ 2.53 & 14.83 $\pm$ 5.80 \\
$\epsilon$-Greedy ($\epsilon = 0.1$) & 13.48 $\pm$ 6.27 & 12.06 $\pm$ 4.99 & 13.76 $\pm$ 2.16 & 8.54 $\pm$ 3.19 & 14.83 $\pm$ 5.80  \\
$\epsilon$-Greedy ($\epsilon = 0.5$) & 13.64 $\pm$ 6.07 & 12.30 $\pm$ 4.96 & 13.76 $\pm$ 2.16 & 9.88 $\pm$ 3.39 & 13.10 $\pm$ 4.77 \\
UCB ($\alpha = 0.1$) & 13.48 $\pm$ 6.27 & 12.30 $\pm$ 4.94 & 13.76 $\pm$ 2.16 & 7.47 $\pm$ 2.53 & 14.83 $\pm$ 5.80 \\
UCB ($\alpha = 1.0$) & 13.71 $\pm$ 5.99 & 12.33 $\pm$ 4.94 & 13.76 $\pm$ 2.16 & 9.27 $\pm$ 3.73 & 14.83 $\pm$ 5.80 \\
UCB ($\alpha = 10.0$) & 13.68 $\pm$ 5.98 & 12.62 $\pm$ 4.83 & 13.76 $\pm$ 2.16 & 8.21 $\pm$ 1.85 & 13.05 $\pm$ 4.58 \\
Ada-$\epsilon$-Greedy ($\epsilon = 0.1$) & 11.34 $\pm$ 5.84 & 12.81 $\pm$ 4.97 & 13.28 $\pm$ 1.97 & 8.70 $\pm$ 1.36 & 15.62 $\pm$ 3.38 \\
Ada-UCB ($\alpha = 1.0$) & 11.34 $\pm$ 5.84 & 12.81 $\pm$ 4.97 & 13.28 $\pm$ 1.97 & 8.70 $\pm$ 1.36 & 15.62 $\pm$ 3.38 \\
TiaraS & 16.30 $\pm$ 8.05 & 22.79 $\pm$ 7.73 & 12.67 $\pm$ 4.23 & 15.13 $\pm$ 2.88 & 25.40 $\pm$ 1.11 \\
Tiara & \textbf{26.14 $\pm$ 0.00} & \textbf{27.73 $\pm$ 0.00} & \textbf{15.92 $\pm$ 0.00} & \textbf{16.91 $\pm$ 0.00} & \textbf{25.77 $\pm$ 0.00} \\ \bottomrule
    \end{tabular}
    }
\newline
\vspace*{0.1 in}
\newline
    \scalebox{0.65}{
\begin{tabular}{lccccccccccc} \toprule
Method $\backslash$ Class & Cliff Dwelling & Hook & Paper Towel & Slot Machine & Water Tower & Average \\ \midrule
Random & 12.63 $\pm$ 5.31 & 11.58 $\pm$ 1.51 & 10.40 $\pm$ 1.01 & 9.29 $\pm$ 1.68 & 12.09 $\pm$ 1.64 & 11.57 $\pm$ 0.89 \\ 
$\epsilon$-Greedy ($\epsilon = 0.01$) & 13.32 $\pm$ 6.07 & 11.17 $\pm$ 1.85 & 9.79 $\pm$ 0.84 & 8.55 $\pm$ 1.02 & 13.25 $\pm$ 1.41 & 11.79 $\pm$ 0.99 \\
$\epsilon$-Greedy ($\epsilon = 0.1$) & 13.06 $\pm$ 6.26 & 11.17 $\pm$ 1.85 & 9.66 $\pm$ 0.83 & 8.84 $\pm$ 0.88 & 13.25 $\pm$ 1.41 & 11.86 $\pm$ 1.08 \\
$\epsilon$-Greedy ($\epsilon = 0.5$) & 13.36 $\pm$ 6.02 & 11.64 $\pm$ 1.96 & 9.74 $\pm$ 0.95 & 8.42 $\pm$ 1.16 & 13.06 $\pm$ 1.63 & 11.89 $\pm$ 1.07 \\
UCB ($\alpha = 0.1$) & 13.32 $\pm$ 6.07 & 11.17 $\pm$ 1.85 & 9.79 $\pm$ 0.84 & 8.55 $\pm$ 1.02 & 13.25 $\pm$ 1.41 & 11.79 $\pm$ 0.99 \\
UCB ($\alpha = 1.0$) & 13.32 $\pm$ 6.07 & 11.17 $\pm$ 1.85 & 10.14 $\pm$ 0.98 & 8.42 $\pm$ 1.17 & 13.25 $\pm$ 1.41 & 12.02 $\pm$ 1.14 \\
UCB ($\alpha = 10.0$) & 13.36 $\pm$ 6.02 & 11.61 $\pm$ 2.00 & 9.77 $\pm$ 1.25 & 8.70 $\pm$ 1.50 & 13.25 $\pm$ 1.41 & 11.80 $\pm$ 0.87 \\
Ada-$\epsilon$-Greedy ($\epsilon = 0.1$) & 10.74 $\pm$ 2.13 & 11.09 $\pm$ 1.72 & 9.93 $\pm$ 0.81 & 8.43 $\pm$ 0.88 & 10.25 $\pm$ 1.34 & 11.22 $\pm$ 1.08 \\
Ada-UCB ($\alpha = 1.0$) & 10.74 $\pm$ 2.13 & 11.09 $\pm$ 1.72 & 9.93 $\pm$ 0.81 & 8.43 $\pm$ 0.88 & 10.25 $\pm$ 1.34 & 11.22 $\pm$ 1.08 \\
TiaraS & \textbf{22.21 $\pm$ 0.00} & 12.21 $\pm$ 1.65 & 9.73 $\pm$ 1.69 & 12.60 $\pm$ 2.56 & 14.81 $\pm$ 1.08 & 16.39 $\pm$ 1.21 \\
Tiara &\textbf{22.21 $\pm$ 0.00} & \textbf{14.06 $\pm$ 0.16} & \textbf{22.13 $\pm$ 0.00} & \textbf{16.77 $\pm$ 0.30} & \textbf{15.69 $\pm$ 0.00} & \textbf{20.33 $\pm$ 0.03} \\ \bottomrule
    \end{tabular}
    }
    \label{tab: openimage}
\end{table*}

We investigate the performance of \textsc{Tiara} using various real-world datasets.

\subsection{Experimental Setups}

\subsubsection{Baselines}

We use the following baselines.

\begin{itemize}
    \item \textbf{Random} queries random known tags.
    \item \textbf{$\varepsilon$-Greedy} is a bandit algorithm. This algorithm chooses the best tag with the highest empirical mean reward with a probability of $1 - \varepsilon$ and chooses a random tag with probability $\varepsilon$. The candidate pool is set to $\mathcal{T}_\text{ini}$.
    \item \textbf{UCB} is a bandit algorithm. The score of tag $t$ is the empirical mean reward plus $\alpha \sqrt{1/n_t}$, where $n_t$ is the number of observations from tag $t$, and $\alpha$ is a hyperparameter. UCB chooses the tag with the highest score. The candidate pool is set to $\mathcal{T}_\text{ini}$.
    \item \textbf{Ada-$\varepsilon$-Greedy} is a variant of $\varepsilon$-Greedy. This algorithm inserts new tags to the candidate pool $\mathcal{T}_\text{known}$ when new tags are found and chooses a tag from $\mathcal{T}_\text{known}$.
    \item \textbf{AdaUCB} is a variant of UCB. This algorithm inserts new tags to the candidate pool $\mathcal{T}_\text{known}$ when new tags are found and chooses a tag from $\mathcal{T}_\text{known}$.
    \item \textbf{\textsc{Tiara}-S} is a variant of \textsc{Tiara} that uses only the query tag for training.
\end{itemize}

\subsubsection{Hyperparameters}

We use $\lambda = 1$ and $\alpha = 0.01$ for \textsc{Tiara} and \textsc{Tiara}-S across all settings without further tuning. We will show that \textsc{Tiara} is insensitive to the choice of these hyperparameters over orders of magnitude in Section \ref{exp: sensitivity}. We use $300$-dimensional GloVe\footnote{\url{https://nlp.stanford.edu/projects/glove/}} trained using six billion Wikipedia 2014 + Gigaword 5 tokens for the word embeddings.

We report the performance with various hyperparameters for the baseline methods. Note that hyperparameter tuning is prohibitive in practice because of the tight API limit. If we apply hyperparameter tuning, we should use these query budgets for the main task instead. Therefore, this setting is slightly advantageous for the baseline methods.

\subsection{Open Images dataset}

\begin{figure}[t]
\centering
\includegraphics[width=\hsize]{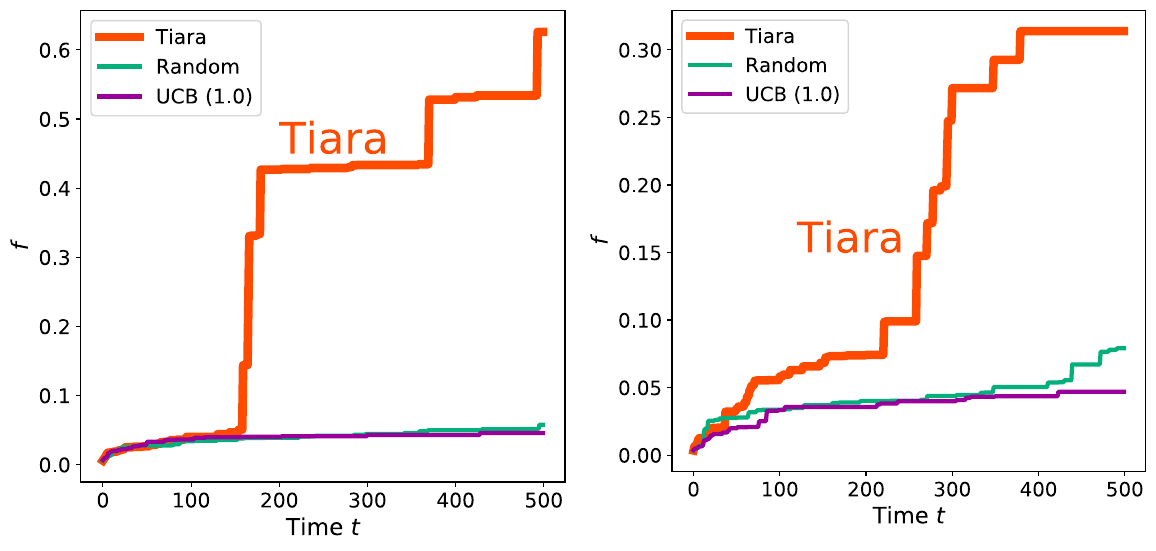}
\caption{Safebooru Environment. Learning curves averaged over $10$ independent runs for two source images. Visual results are available at the technical report \cite{sato2021retrieving} \url{http://arxiv.org/abs/2112.14921}.}
\label{fig: miku}
\end{figure}

We use the Open Images Dataset V6 \cite{kuznetsova2020open} to construct the environment for the first testbed. Each image in this dataset has multi-class annotations, such as ``cat,'' ``Aegean cat,'' and ``Pumpkin pie.'' An image has $8.8$ classes on average. We use these classes as tags. We utilize ResNet18 \cite{he2016deep} trained with ImageNet for the black-box functions. Specifically, for each class $c$ of ImageNet, we use the pre-softmax logit for $c$ as the back-box function. Among $1000$ classes of ImageNet, we use $10$ classes $c = 1, 101, \cdots, 901$, i.e., Tench, Black Swan, Tibetan Terrier, Tiger Beetle, Academic Gown, Cliff Dwelling, Hook, Paper Towel, Slot Machine, and Water Tower. Note that the retrieval algorithms do not know these class names, i.e., the function is a black-box. We use these class names for only evaluations and interpretations of the results. Note also that the set of tags (i.e., classes of the Open Image Dataset) differs from the set of classes of ImageNet.

We subsample $10{,}000$ test images from the Open Images Dataset and construct an environment. When we query tag (class) $t$, this environment returns a random image with class $t$ and the set of classes this image belongs to. We choose $100$ random tags as the initial known tags $\mathcal{T}_\text{ini}$ and set the budget to $B = 500$.

We run ten trials with different seeds. Table \ref{tab: openimage} reports the means and standard deviations of the best $f$ for each method within $B$ queries. The last column reports the average of ten classes. These results show that \textsc{Tiara} performs the best under all settings, and \textsc{Tiara}-S performs second best on average.

The middle panels in Figures \ref{fig: openimage-tiara} show the word clouds \footnote{\url{https://github.com/amueller/word_cloud}} generated by \textsc{Tiara}. The size and color of a tag represent the score of the tag at the final iteration. These scores provide interpretations for the black-box functions. For example, in the case of tench, fish-related tags, such as Trout, Carp, and Milkfish, have high scores. In the black swan case, bird-related tags, such as Swan, Waterfowl, and Duck, have high scores. We stress again that \textsc{Tiara} does not use the ground truth class name but instead treats $f$ as a black-box function. Even when the ground truth class name is unavailable to us, the word cloud generated by \textsc{Tiara} and the retrieved images indicate that $f$ is fish-related in the first example and bird-related in the second example with significant interpretability.

\begin{figure*}[t]
\centering
\includegraphics[width=\hsize]{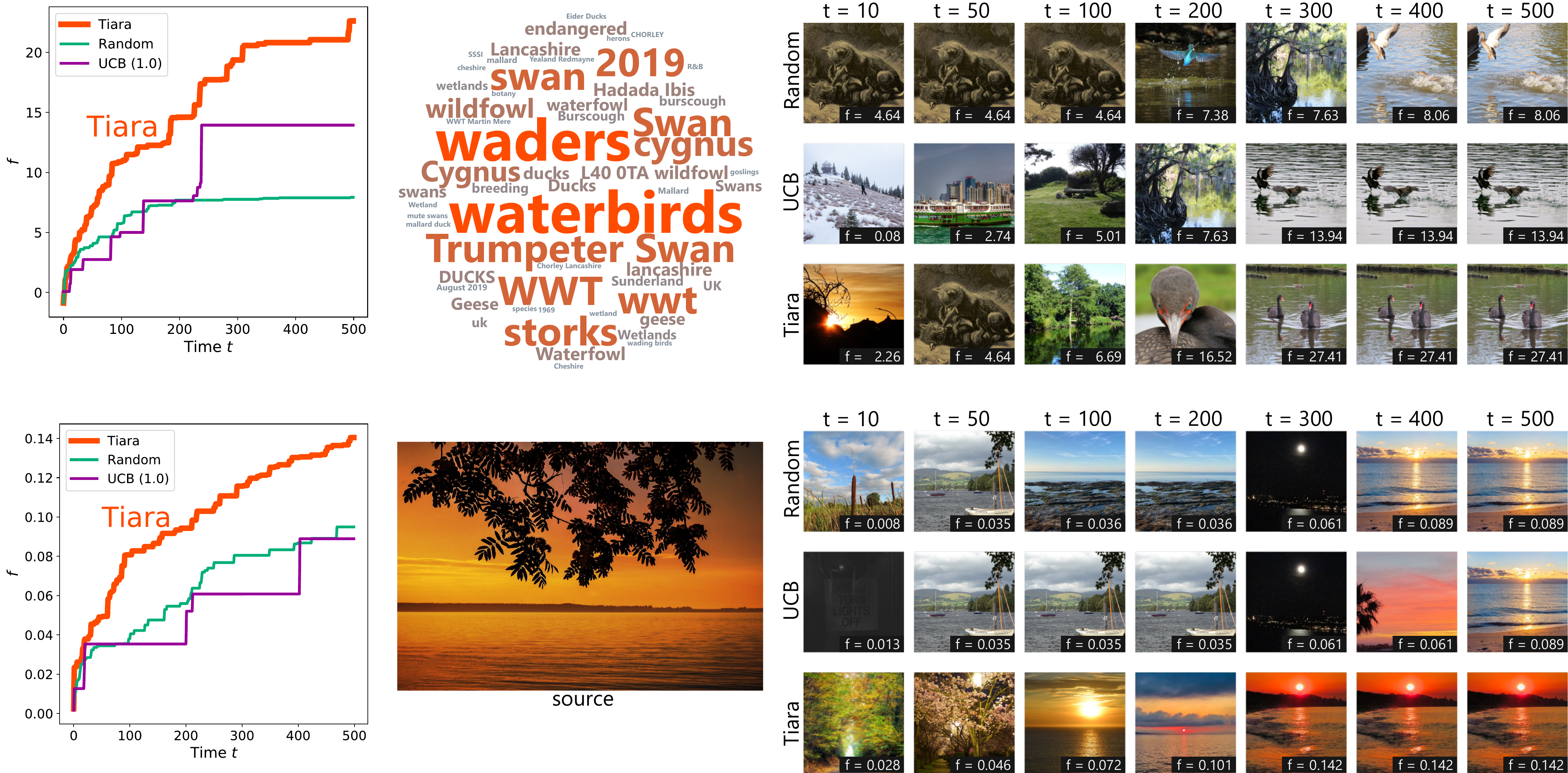}
\caption{Flickr Environment. (Top) Black swan. (Bottom) Similar image retrieval. (Left) Learning curves averaged over ten independent runs. (Top middle) Word cloud generated by \textsc{Tiara}. (Bottom middle) Source image. (Right) Best images found duing each iteration.}
\label{fig: flickr-tiara}
\end{figure*}

\subsection{Safebooru Environment}

The aim of this section is to confirm that \textsc{Tiara} is effective even in a completely different domain. We use Safebooru as a testbed. We use a dump retrieved on June 7, 2019 \footnote{\url{https://www.kaggle.com/alamson/safebooru}}. Each image has $34$ tags on average, such as ``smile,'' ``long hair,'' and ``blonde hair.'' We use illustration2vec \cite{saito2015illustration} to construct the black-box functions.

We use the latest $100{,}000$ images from this dataset and construct an environment. As a result of the broken links, this environment contains $81{,}517$ images in total. When we query tag $t$, this environment returns a random image with tag $t$ and the set of tags this image belongs to. We choose $100$ random tags as the initial known tags $\mathcal{T}_\text{ini}$ and set the budget to $B = 500$.

We conduct semantically similar image search experiments. We emphasize that although this environment does not provide an official image-based search, our algorithm enables such an image-based search. Given a source image $s$, we compute $4096$-dimensional embedding $\boldv_s \in \mathbb{R}^{4096}$ using the pre-trained illustration2vec model \footnote{\url{https://github.com/rezoo/illustration2vec}}. We use the Gaussian kernel between the embeddings of an input image $x$ and the source image as the black-box function, i.e., $f(x) = \exp(-\|\boldv_s - \boldv_x\|^2/\sigma^2)$, where $\boldv_x \in \mathbb{R}^{4096}$ is the embedding of input $x$ computed using the illustration2vec model, and $\sigma$ is the bandwidth of the Gaussian kernel. We set $\sigma = 100$ during this experiment. The more semantically similar the input image is to the source, the higher the value taken by this function.

We use two illustrations (visually shown in the technical report \cite{sato2021retrieving}) as the source images. The first source image is in the image database, and \textsc{Tiara} succeeds in retrieving the same image from the database. Even the second best image retrieved by \textsc{Tiara} is semantically similar to the source image, i.e., it depicts the same character with cherry blossom motifs, and has a higher objective value than the images retrieved by the baseline methods. For the second case, the source image is not in the database. Although not exactly the same, \textsc{Tiara} succeeds in retrieving a semantically similar image with the same characters and motifs, i.e., horn, rabbit, and costume.

These results show the flexibility of our framework such that \textsc{Tiara} can be applied to not only photo-like image databases but also illustration-like image databases using appropriate black-box functions. e.g., the illustration2vec model.

\subsection{Flickr Environment}

The aim of this section is to confirm that \textsc{Tiara} is effective and readily applicable to real-world environments. We use the online Flickr environment in operation as a testbed. We also use ResNet18 trained with ImageNet to construct the black-box functions. We implement oracle $\mathcal{O}$ by combining \texttt{flickr.photos.search} and \texttt{flickr.tags.getListPhoto} APIs. We use the \texttt{license='9,10'} option such that it returns only public domain or CC0 images. We also set the budget to $B = 500$. Because Flickr contains as many as \emph{ten billion} images, it is challenging to find relevant images within $B = 500$ queries. Each image has $15$ tags on average, including ``summer,'' ``water,'' and ``sea.'' 

We conduct two experiments in this environment. First, we conduct the same experiment as in the open image environment. The top row of Figure \ref{fig: flickr-tiara} shows the result for the black swan class. Compared to Figure \ref{fig: openimage-tiara}, \textsc{Tiara} in this environment learns more slowly than in the open image environment. We hypothesis that this is because the Flickr environment contains noisy tags annotated by users, which are occasionally described in foreign languages, whereas the open image environment contains only clean tags that indicate solid categories judged by annotators. Nevertheless, \textsc{Tiara} succeeds in retrieving black swan images in several hundred queries.

Second, we conduct similar image retrieval experiments as in the safebooru environment. We use the output of the penultimate layer of the pre-trained ResNet18 for the image embedding. We use the Gaussian kernel between the embeddings of an input image $x$ and the source image $s$ as the black-box function, i.e., $f(x) = \exp(-\|\boldv_s - \boldv_x\|^2/\sigma^2)$, where $\boldv_x \in \mathbb{R}^{512}$ is the embedding of input $x$ computed using ResNet18, and $\sigma$ is the bandwidth of the Gaussian kernel. We set $\sigma = 10$ in this experiment. Figure \ref{fig: flickr-tiara} (bottom, center) shows the source image. As the right panels show, \textsc{Tiara} succeeds in retrieving semantically similar images depicting a sunrise over the sea.

\begin{figure*}[t]
\centering
\includegraphics[width=0.32\hsize]{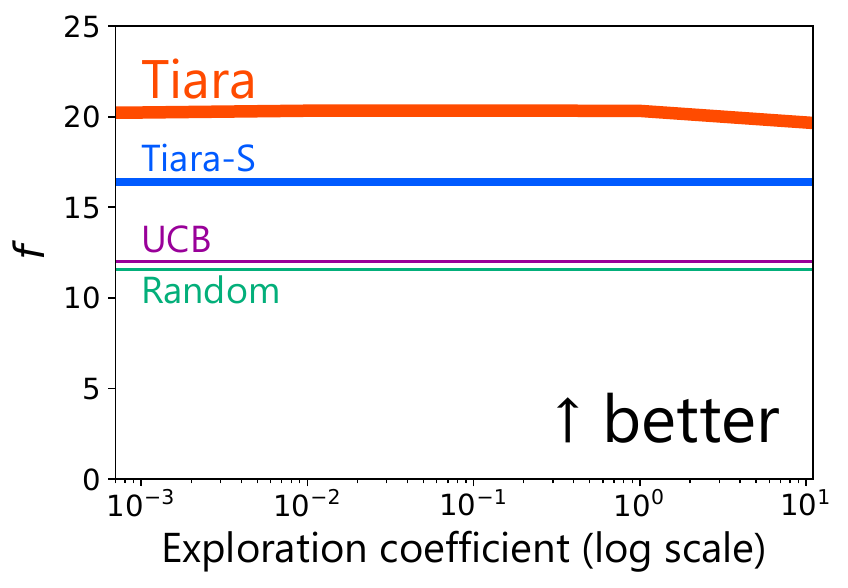}
\includegraphics[width=0.32\hsize]{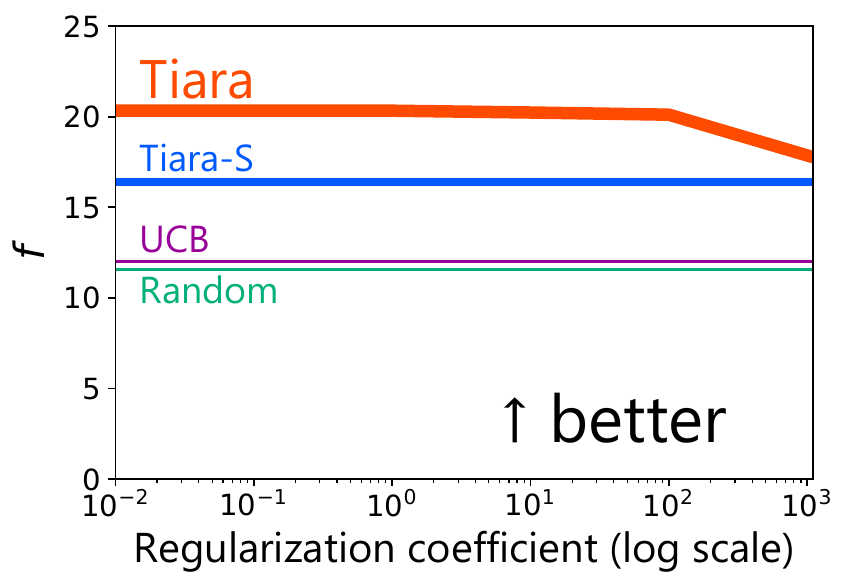}
\includegraphics[width=0.32\hsize]{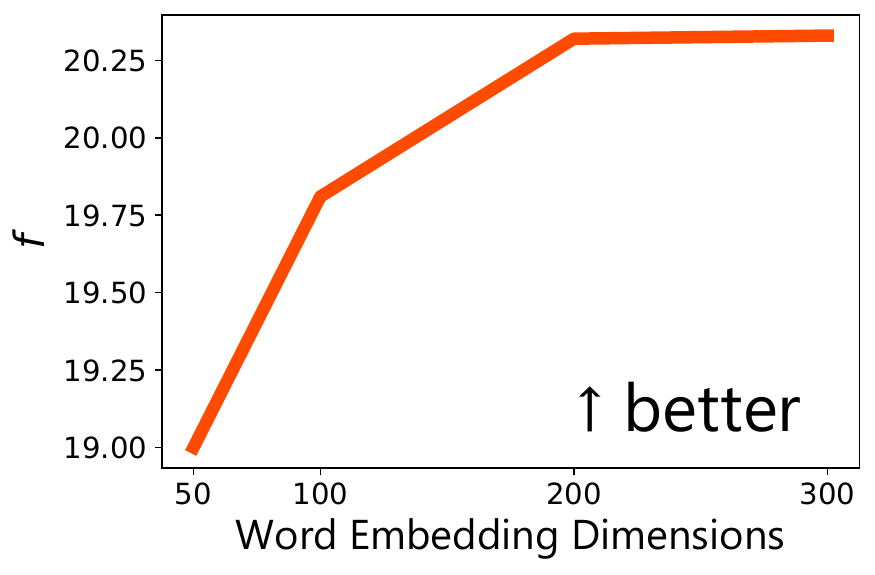}
\caption{Hyperparameter sensitivity. As the default settings, we set the exploration coefficient as $\alpha = 0.01$, regularization coefficient as $\lambda = 1.0$, and word embedding as $300$-dimensional Glove and change one configuration in each panel. These results show that \textsc{Tiara} is stable with respect to the hyperparameter choices across several orders of magnitude.}
\label{fig: hyperparameter}
\end{figure*}

\subsection{Hyperparameter Sensitivity} \label{exp: sensitivity}

We investigate the hyperparameter sensitivity of \textsc{Tiara} in this section. \textsc{Tiara} has two hyperparameters, i.e., exploration coefficient $\alpha$ and regularization coefficient $\lambda$. \textsc{Tiara} also has a choice of word embeddings. It is crucial to ensure stability with respect to the hyperparameter choices because tuning the hyperparameters in a new environment is difficult with tight API budgets. The left panel of Figure \ref{fig: hyperparameter} shows the average performance of \textsc{Tiara} in the Open Image Dataset environment with various values of $\alpha$ maintaining $\lambda = 1.0$ (default value). The $x$-axis is plotted on a log scale. We show the performances of \textsc{Tiara}-S, UCB, and Random for reference. In this plot, we do not tune the hyperparameters of \textsc{Tiara}-S accordingly to maintain the conciseness of the plot. This plot shows that \textsc{Tiara} is stable with respect to $\alpha$ over several orders of magnitude. The middle panel of Figure \ref{fig: hyperparameter} shows the average performance with various values of $\lambda$ maintaining $\alpha = 0.01$ (default value). The $x$-axis is plotted on a log scale. The result shows that \textsc{Tiara} is stable with respect to $\lambda$ over several orders of magnitude. The right panel of Figure \ref{fig: hyperparameter} shows the average performance with various dimensions of the GloVe word embedding maintaining $\alpha = 0.01$ and $\lambda = 1.0$ (default values). This result shows that \textsc{Tiara} performs better with higher dimensional embeddings, i.e., more expressive embeddings. It also indicates that \textsc{Tiara} indeed utilizes the word embedding geometry. We recommend using high-dimensional and strong word embeddings for \textsc{Tiara}.

\section{Related Work}

\subsection{Image Retrieval}

Image retrieval has been studied for decades. The main concern is how to retrieve relevant images \cite{babenko2015aggregating, kato2008can, leiva2011query} with efficiency \cite{lai2015simultaneous, liu2016deep}. Deep neural networks have been preferred for image retrieval in recent years because they can extract rich features including texture \cite{cimpoi2015deep}, style \cite{gatys2016image}, and semantics \cite{babenko2015aggregating}. Hash-based methods have also been extensively studied for their efficiency \cite{lai2015simultaneous, liu2016deep}. In addition, many extensions have been studied, such as multi-modal image search \cite{cao2016deep, kordan2018deep} and contextual retrieval \cite{xie2019improving}. We investigated the image retrieval problem from a different perspective. Specifically, we consider how an outsider can retrieve desirable images from an \emph{external} image database with as few queries as possible.

\subsection{Web Crawling}

Our problem setting can be seen as a crawling problem. Developing efficient web crawlers is a long-standing problem in the literature \cite{cho1998efficient, chakrabarti1999focused, diligenti2000focused, cho2003effective, castillo2005effective, pham2019bootstrapping}. In particular, focused crawling \cite{chakrabarti1999focused, mccallum2000automating, johnson2003evolving, baezayates2005crawling, guan2008guide} is relevant to our problem setting. Focused crawling aims to efficiently gather relevant pages by skipping irrelevant pages.

However, there are several differences between our study and existing focused crawling. First, focused crawling is used to search web pages by following WWW hyperlinks, whereas we search for \emph{images} from an \emph{external database} utilizing the tag search oracle. Thus, the existing crawling methods are not directly applicable to our setting. Second, we assume the API limit is extremely tight, e.g., $500$, whereas existing focused crawlers typically visit hundreds of thousands of pages. Third, existing focused crawlers focus on retrieving pages of a specific topic \cite{chakrabarti1999focused, mccallum2000automating}, popular pages \cite{baezayates2005crawling}, pages with structured data \cite{meusel2014focused}, or hidden web pages \cite{barbosa2007adaptive}. By contrast, deep convolutional neural networks in our setting realize rich vision applications, as we show in the experiments. These applications are qualitatively different from the existing focused crawlers and are valuable in their own right.

\subsection{Private Recommender Systems}

Private recommender systems \cite{sato2022private} aim to build a fair recommender system when the service provider does not offer a fair system. Our problem can also be seen as constructing a private recommender system when the black-box function is a recommendation score. There are several differences between our approach and private recommender systems. First, the existing methods, \textsc{PrivateRank} and \textsc{PrivateWalk}, assume the use of an item recommendation oracle, which is unavailable under our setting. Second, \textsc{Tiara} considers content-based retrieval, whereas \textsc{PrivateRank} and \textsc{PrivateWalk} mainly focus on a collaborative recommendation scenario. Third, our application is not limited to fairness, and we showed promising applications, including semantically similar image retrieval.

\section{Discussion and Limitations}

Although we use deep neural networks as a black-box function, our framework and the proposed method are not limited to deep neural networks. For example, we can use human judgment as the black-box function, i.e., when we evaluate $f(x)$, we actually ask a human viewer how much he/she likes image $x$. Because our proposed method requires several hundred evaluations for a single run, the current method is too inefficient for human-in-the-loop experiments. More efficient methods are important for developing such intriguing applications.

In this work, we have assumed that $f$ is a well-behaved function. It may be interesting to use buggy $f$ functions. For example, when we debug a deep neural network model $f$, applying \textsc{Tiara} to $f$ may reveal what $f$ has already learned and has not yet. We hypothesize that there are much room for further applications of \textsc{Tiara}.

We have assumed that the only way to access the image database is the tag search oracle $\mathcal{O}$. Although our method is general owing to this formulation, many other APIs are available in real applications, such as user-based searches, popularity ranking, and collaborative filtering recommendations. Utilizing richer information to enhance performance and query efficiency is important in practice. At the other extreme, dropping the assumptions of tag affinity and tag search APIs to make the method applicable to broader databases is also an intriguing direction.

Although we have focused on image retrieval problems, our formulation is not limited to such problems and can be applied to other domains such as music, document, and video retrieval problems. Exploring further applications of our framework is left as future work.

\section{Conclusion}

In this chapter, we formulated the problem of optimal image retrieval with respect to a given black-box function from an external image database. This problem enables each user to retrieve their preferable images from the Internet, even if the image database does not provide such features. We combined a bandit formulation with pre-trained word embeddings and proposed an effective retrieval algorithm called \textsc{Tiara}. Finally, we confirmed the effectiveness of \textsc{Tiara} using three environments, including an online Flickr environment.

\chapter{CLEAR: A Fully User-side Image Search System}


\section{Introduction}

A massive amount of information is uploaded on the Internet every day, and it becomes more important yet difficult to search for desired information from the flood of information. Thus, many functionalities of search engines have been called for, including multi-modal search \cite{cheng2009sketch, zha2009visual, cao2016deep, kordan2018deep} and fairness-aware systems \cite{singh2018fairness, biega2018equity, sato2022enumerating}. However, many information retrieval systems on the Internet have not adopted rich functionalities, and they usually accept simple text queries only. Even if a user of the service is unsatisfied with a search engine and is eager to enjoy additional functionalities, what he/she can do is limited. In many cases, he/she continues to use the unsatisfactory system or leaves the service.

\begin{figure}[p]
  \centering
    \includegraphics[width=\hsize]{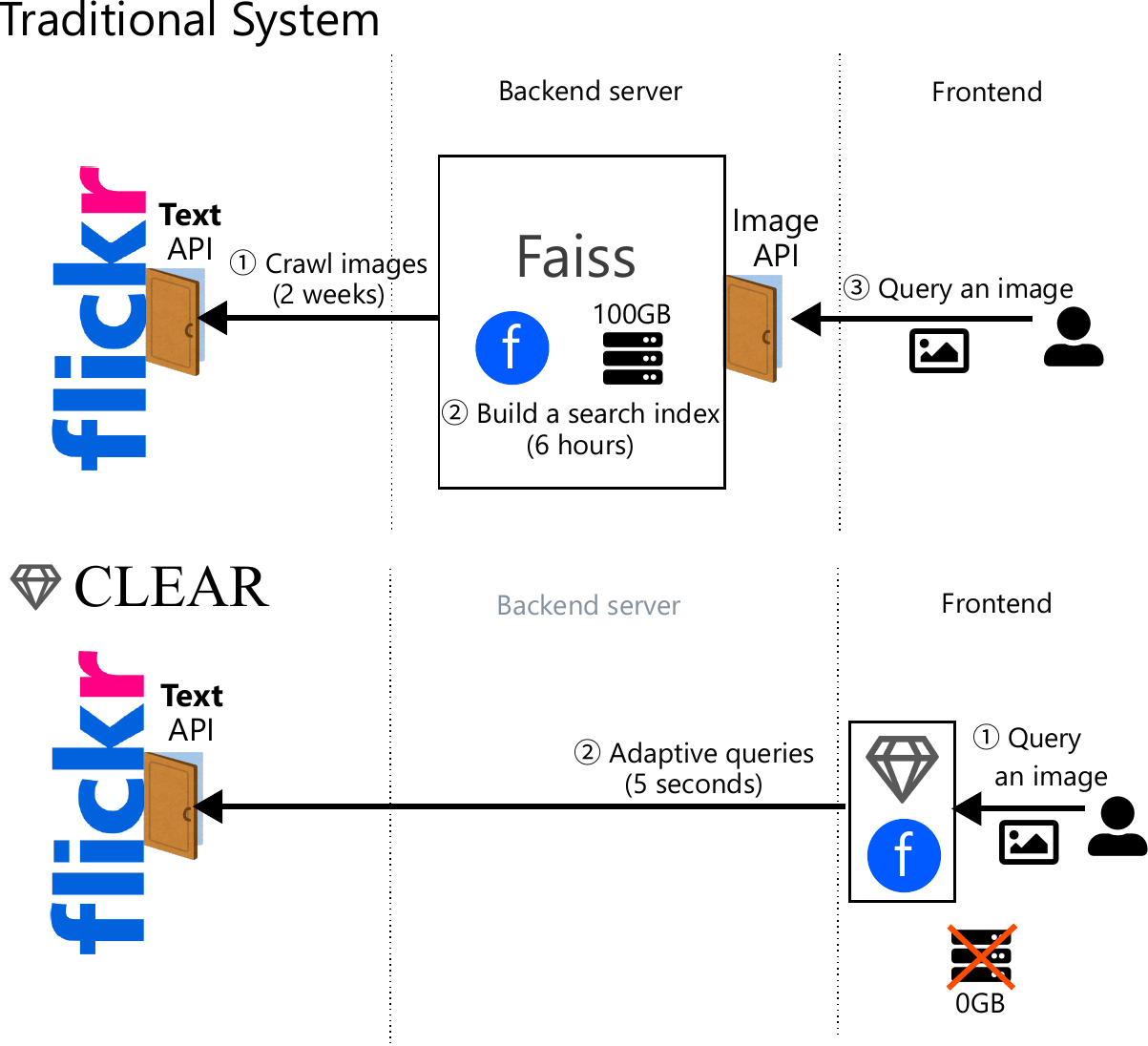}
\caption{Flickr does not provide an official similar image search engine or corresponding API. It is not straightforward for an end-user to build their own search system. (Top) A traditional system would first crawl images on Flickr and build a search index. It typically takes a few weeks. (Bottom) Our proposed approach does not require initial crawling or indexing. CLEAR issues adaptive queries to the text-based image search API, which is available. It realizes an image-based search combining several effective text queries. The f mark represents the score function. In the traditional system, the score function is tightly connected to the search index. It is costly to change it. By contrast, the score function is on the frontend in CLEAR. It can be seamlessly redefined.}
 \label{fig: illust}
\end{figure}

Let us consider Flickr\footnote{\url{https://www.flickr.com/}}, an image hosting service, as an example. Suppose we are enthusiastic users of Flickr. Although there are many amazing photos and communities on Flickr, we are not satisfied with the search engine of Flickr. Specifically, it accepts only text queries, and we cannot conduct a similar image search on Flickr. One possible yet rare solution to realize a similar image search is to join SmugMug, the owner company of Flickr, and implement the feature as an official software engineer. However, this choice is not realistic for an ordinary user.

Another possible solution is to build a search system by ourselves. We can crawl images from Flickr, build a search index by faiss \cite{johnson2019billion}, and host the search engine by ourselves. However, there are more than five billion photos on Flickr, and it typically takes a few weeks to crawl Flickr even if we selectively crawl pages with focused crawling techniques \cite{chakrabarti1999focused, mccallum2000automating, johnson2003evolving, pham2019bootstrapping}. Besides, it also requires a costly backend server to host a search engine. Therefore, this choice is neither practical for many users.

Recently, more practical user-side information retrieval algorithms have been proposed \cite{sato2022private, sato2022retrieving}. They run on a browser and do not need a backend server. However, the evaluations of existing works are conducted offline, and it is still not clear if this concept works in practice. Indeed, the official implementation of Tiara \cite{sato2022retrieving}\footnote{\url{https://github.com/joisino/tiara}}, a user-side image search algorithm, requires a few minutes for a single search, which significantly degrades user experiences and is not practical for a real-time search system. In this demonstration, we simplify and optimize Tiara and thereby show a practical user-side image search system online. \emph{This is the first practical implementation of a fully-user side search system}. We named our system CLEAR, which stands for \underbar{CL}ient-side s\underbar{EAR}ch. The features of CLEAR are as follows.
\begin{description}
\item[Lightweight.] As mentioned above and illustrated in Figure 1, traditional search systems require crawling the Web and building a search index, which takes many weeks and consumes much computational and network resources. By contrast, CLEAR does not require initial crawling or indexing, i.e., CLEAR incurs zero overhead. Each user can fork, deploy, and use CLEAR instantly. It is also easy to redefine their own scoring function. Thus, end-users can easily enjoy their own search systems.
\item[Fast Iteration.] When we update the feature extractor or scoring function, an ordinary system would need to rebuild the index. By contrast, as CLEAR does not use any search indices in the first place, we can seamlessly update the scoring function just by rewriting the JavaScript snippet. It accelerates the development of the desirable scoring function.
\item[Privacy-aware.] If we adopted a traditional system configuration with a backend server, we could snoop on the uploaded images in theory. By contrast, as CLEAR runs totally on the user-side, users do not have to worry about privacy issues.
\end{description}
In addition, we stress that user-side systems do not need to be used privately. Rather, each user can publish their own search engines with their own scoring function and their own interface. As CLEAR does not need a backend server, hosting a user-side system is much easier. For example, their own search system can be deployed on a static page hosting service such as Amazon S3 and GitHub Pages. We expect many characteristic user-side search systems will appear on the Internet, and it will become easy to find favorite engines. We hope CLEAR facilitates this trend.

Another use case of CLEAR is prototyping by official/unofficial developers. Even if their goal is to build a traditional system with a backend server, CLEAR accelerates the design of the score function. Besides, as CLEAR can be instantly deployed, interface designers can try and develop the system, and they can carry out user studies \emph{before} the backend engineers complete the system.

The online demo is available at \url{https://clear.joisino.net}. The source code is available at \url{https://github.com/joisino/clear}.

\section{Related Work}

User-side information retrieval systems \cite{sato2022private, sato2022retrieving, sato2022principled} enable each user to build their own system, whereas traditional systems are developed by service developers. For example, private recommender systems \cite{sato2022private} and Consul \cite{sato2022principled} turn recommendation results into fair and/or diverse ones even if the official recommender systems are not fair or diverse. Tiara \cite{sato2022retrieving} is the most relevant work to this demonstration. It realizes image retrieval based on user-defined score functions. The critical difference between Tiara and this demonstration is that the original paper of Tiara conducted only batch evaluations. Indeed, the official implementation of Tiara requires a few minutes for a single search on Flickr, which is too slow for real-time demonstration. We optimize the system and realize a real-time user-side search system.

A steerable system \cite{green2009generating, balog2019transparent} also allows users to customize the system. However, the critical difference with user-side systems is that steerable systems are implemented by official developers, and users cannot enjoy this feature if the official system is not steerable in the first place. By contrast, user-side systems turn ordinary systems into steerable ones on the user-side.

Another relevant realm is focused crawling \cite{chakrabarti1999focused, baezayates2005crawling, barbosa2007adaptive, guan2008guide, meusel2014focused}. In contrast to exhaustive crawling, focused crawling aims to retrieve only relevant information. Although these techniques accelerate crawling, they still require several hours to several weeks. Therefore, most users cannot afford to adopt such systems due to time, network, or computational resources. In stark contrast, our system does not require initial crawling at all.

\section{CLEAR}

\begin{figure}[t]
  \centering
    \includegraphics[width=\hsize]{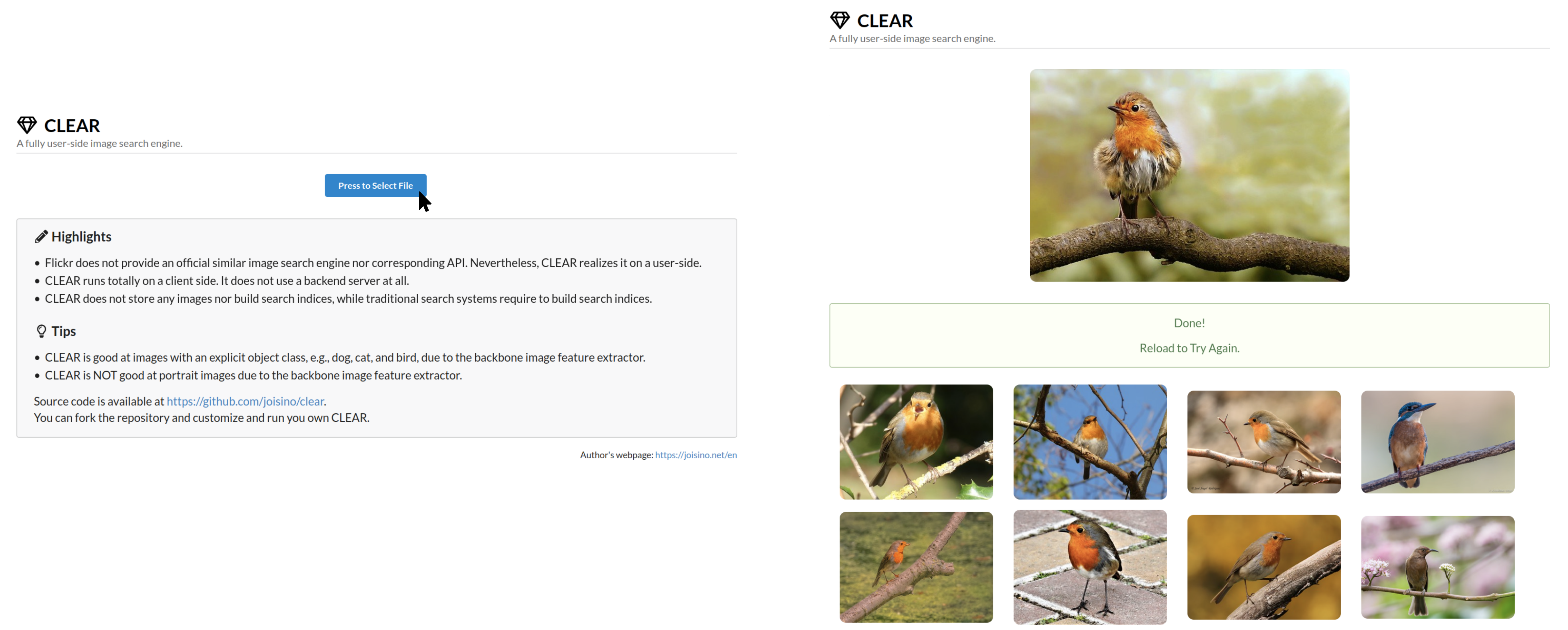}
\caption{Interface of CLEAR. Upload an image, and then CLEAR retrieves similar images from Flickr. The functionality of CLEAR is ordinary. The highlight lies rather in how it is realized and how easy deployment is. (Left) User interface. (Right) Search results.}
 \label{fig: interface}
\end{figure}

\subsection{User Interface}

The interface of CLEAR is simple (Figure \ref{fig: interface}). A user clicks the blue button and uploads an image; then CLEAR retrieves similar images from Flickr. The functionality is an ordinal similar image search. We stress that the highlight of CLEAR lies rather in how it is realized. It is remarkable that CLEAR works similarly to an ordinary system with a backend server.

\subsection{Algorithm}

\begin{figure}[t]
  \centering
    \includegraphics[width=\hsize]{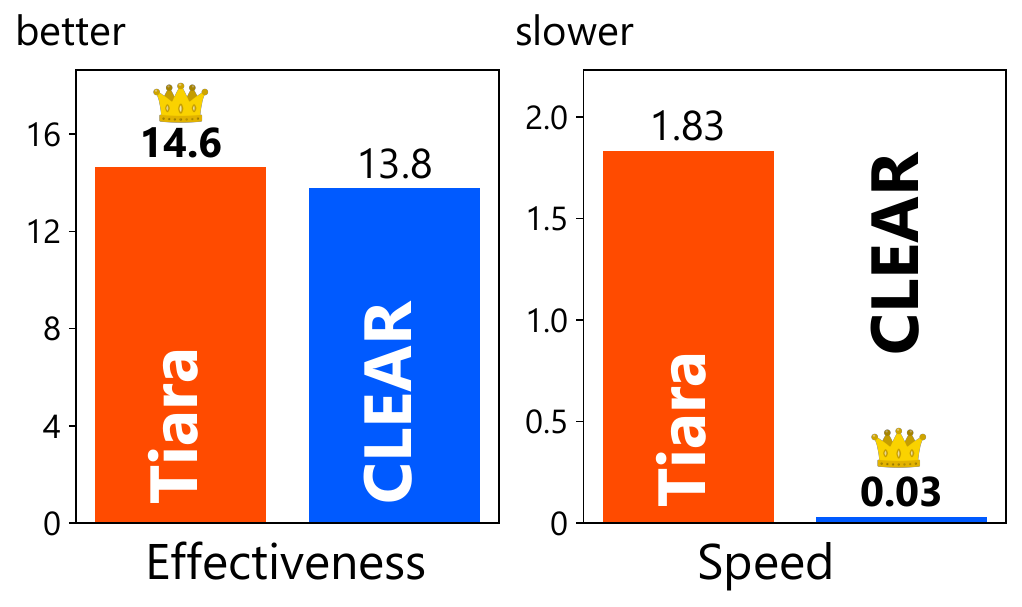}
\caption{The greedy algorithm is fast while the degradation of performance is slight. The better one is highlighted with bold style and a crown mark on each panel. (Left) Each value is the average score of the retrieved images. The higher, the better the retrieved images are. This shows that the greedy algorithm slightly degrades the performance. We confirmed that the retrieved images by both methods were visually comparable. (Right) The average running time in seconds. The greedy method is $60$ times more efficient.}
 \label{fig: experiments}
\end{figure}

\noindent \textbf{Preliminary: Tiara.} We first introduce Tiara \cite{sato2022retrieving}, on which CLEAR is based. Tiara assumes that we have a black-box score function $f$ that takes an image as input and outputs a real-valued score. This score function is typically realized by an off-the-shelf feature extractor such as ResNet and MobileNet. The aim of Tiara is to find images with high scores from an external image database, such as Flickr. The critical assumption is that we do not have direct access to the external database, but we can access it via a tag-based search API only. The core idea of Tiara is that it regards a tag as an arm and the score function $f$ as a reward and formulates this problem as a multi-armed bandit algorithm. Tiara uses the linear bandit algorithm \cite{li2010contextual} with a UCB-style acquisition function \cite{auer2002using}. Tiara utilizes the GloVe word embeddings \cite{pennington2014glove} for tag features to cope with the cold-start problem. Tiara realized an effective similar image search on Flickr with a similarity function $f$.

\vspace{0.1in}
\noindent \textbf{CLEAR.} We design the algorithm for CLEAR based on Tiara. Although the use of the bandit algorithm adopted in Tiara strikes a good tradeoff between effectiveness and efficiency, it consumes too much time, e.g., a few minutes, for a real-time search. Besides, the GloVe embeddings consume about $1$ GB, which makes it difficult to download the program on the client-side. We find that the greedy algorithm that always draws the current best arm is still effective while it's much more efficient. The greedy algorithm corresponds to setting $\lambda = 0.0$ in Tiara, i.e., no exploration. The greedy algorithm is efficient because (i) it does not need to compute the matrix inverse for computing the UCB score, (ii) it does not consume the query budget for exploration, and (iii) it does not require downloading word embeddings. Indeed, exploration may lead to better search results in the long run, but the real-time search in CLEAR runs only several iterations, and thus the greedy approach is sufficient. Besides, we found that the response of an API call was stochastic, and therefore, the greedy algorithm also did some exploration rather than querying only one arm.

To quantify the effect of using the greedy algorithm, we conducted an offline evaluation using the Open Image Dataset \cite{OpenImages}. The experimental setup is the same as in Tiara \cite{sato2022retrieving}. Specifically, we use a subset of the Open Image Dataset with $100\,000$ images as a virtual image repository and use the labels of an image as a set of tags. We use the same score function as Tiara, i.e., the logit of ResNet. We set the query budget as $100$ in this experiment. Figure \ref{fig: experiments} shows that the greedy algorithm slightly degrades the performance while it is $60$ times more efficient. Tiara is effective when there is no tight time constraint. As we call for real-time inference, we adopt the greedy algorithm in the demonstration.

\subsection{Implementation}

\sloppy CLEAR is implemented with the React framework. It adaptively calls \texttt{flickr.photos.search} API via the Axios library. This API takes text as input and returns a set of images and their attached tags. CLEAR manages the average score of each tag and queries the tag with the highest average score iteratively. CLEAR shows the returned images in the order of scores. CLEAR's score function is \begin{align}
f(x) \stackrel{\text{def}}{=} \exp(g(x)^\top g(s) / 1000),
\end{align} where $g$ is the MobileNetv2 \cite{sandler2018mobilenetv2} feature extractor and $s$ is the source image uploaded by the user. The higher the inner product similarity to the source image, the higher this score is. The feature extractor is implemented with the TensorFlow.js library \cite{smilkov2019tensorflow}. To improve the effectiveness and user experience, CLEAR adopts the following strategies.
\begin{description}
\item[Initial queries.] At first, CLEAR does not have any candidate tags. To generate initial queries, CLEAR first obtains class names of the input image using the MobileNetv2 classifier and uses top-$k$ class names for initial queries. This is in contrast to Tiara, which uses a fixed set of initial tags.
\item[Parallelization.] As an API call is the bottleneck of wall-time consumption, CLEAR issues several API calls simultaneously and aggregates results. Specifically, CLEAR selects top-$10$ tags and issues $10$ queries at once. The score evaluation is also parallelized. This is in contrast to Tiara, which issues a single query at once.
\item[Date Specification.] The \texttt{flickr.photos.search} API can specify the dates of images, and the default setting is the latest images. We found that the default setting was not effective and reduced diversity. For example, we upload a cat image, and ``cat'' and ``calico cat'' tags are the identified best tags. CLEAR queries both ``cat'' and ``calico cat'' tags, and the set of returned images would largely overlap each other. To overcome this issue, we randomly set the dates for each tag so that the returned images do not overlap one another.
\item[Real-time display.] Although CLEAR is efficient, it still takes tens of seconds to complete a single search, which degrades the user experience. To mitigate this, CLEAR shows the intermediate results, and thereby, users can see the first results in a few seconds. This improves the user experience much.
\end{description}

We also note that the search results of CLEAR are stochastic. First, the Flickr API is stochastic because many images are being uploaded and some images are being deleted. Second, CLEAR's algorithm itself is stochastic in specifying the dates and selecting tags. As the former one is difficult to remove, we keep the entire system stochastic. The demerit of this choice would be the lack of reproducibility. However, the merit is that a user can try the same query many times and obtain fresh information every time. This improves the chance of serendipitous findings. 

\subsection{Bespoke System}

The source code of CLEAR is available at \url{https://github.com/joisino/clear}. An interested user can fork the repository and build their own system. As mentioned earlier, CLEAR is easy to deploy and host because it is a fully-user side system. The score function is defined in \texttt{src/score.js}. A user can rewrite the function and instantly tries their own scoring function. A small change in the score function may change the behavior of the system, and a user may be able to find a search engine that fits his/her preference.
\begin{itemize}
    \item \texttt{getFeature} function defined in \texttt{src/score.js} computes feature vectors for both the source image and retrieved images. The embedding layer is defined in \texttt{embeddingName}. Users can use other embedding layers for feature extraction.
    \item \texttt{embs2score} function defined in \texttt{src/score.js} computes scores. Higher is better. This demonstration uses \\ \texttt{Math.exp(emb1.mul(emb2).sum().dataSync()[0] / 1000)}. One can try other functions, e.g., the Gaussian kernel \\ \texttt{Math.exp(- emb1.squaredDifference(emb2).sum() \\ .dataSync()[0] / 1000)}.
\end{itemize}
Although we focused on a similar image search system in this demonstration, the score function needs not to measure similarity. If one trains a neural network using his/her favorite images so that preferable images have high scores and uses it as the score function, then the resulting system searches for preferable images. If a fairness-aware score function is used, a fairness-aware search system is realized. Another interesting setting is to use a black-box (possibly buggy) neural network as the score function as proposed in Tiara \cite{sato2022retrieving}. CLEAR finds a set of images that activate the neural network. Such instances visually show what the neural network at hand represents \cite{simonyan2014deep, nguyen2016synthesizing, yuan2020xgnn}. Traditional interpretation methods rely on a fixed set of datasets such as ImageNet, but CLEAR can retrieve images from Flickr, which hosts as many as five billion images. In sum, the creativity for the design of the score function is open to each user.

Users can change the target services on which the search system is built by writing a wrapper in \texttt{src/flickr.js}. The requirement for the target service is that it accepts text and returns a set of images. Users can search images from the Internet using their own score functions if Google Image Search is used instead of Flickr.

As CLEAR does not rely on any backend servers or search indices, one can seamlessly use the system after one changes the score function. One can also change the search target from Flickr to other services by writing a wrapper in \texttt{src/flickr.js}.

Note also that the online demo we provide limits the number of queries because many users may use it simultaneously. A bespoke system allows using more query budgets if needed.

\section{Demonstration}

\begin{figure}[p]
  \centering
    \includegraphics[width=\hsize]{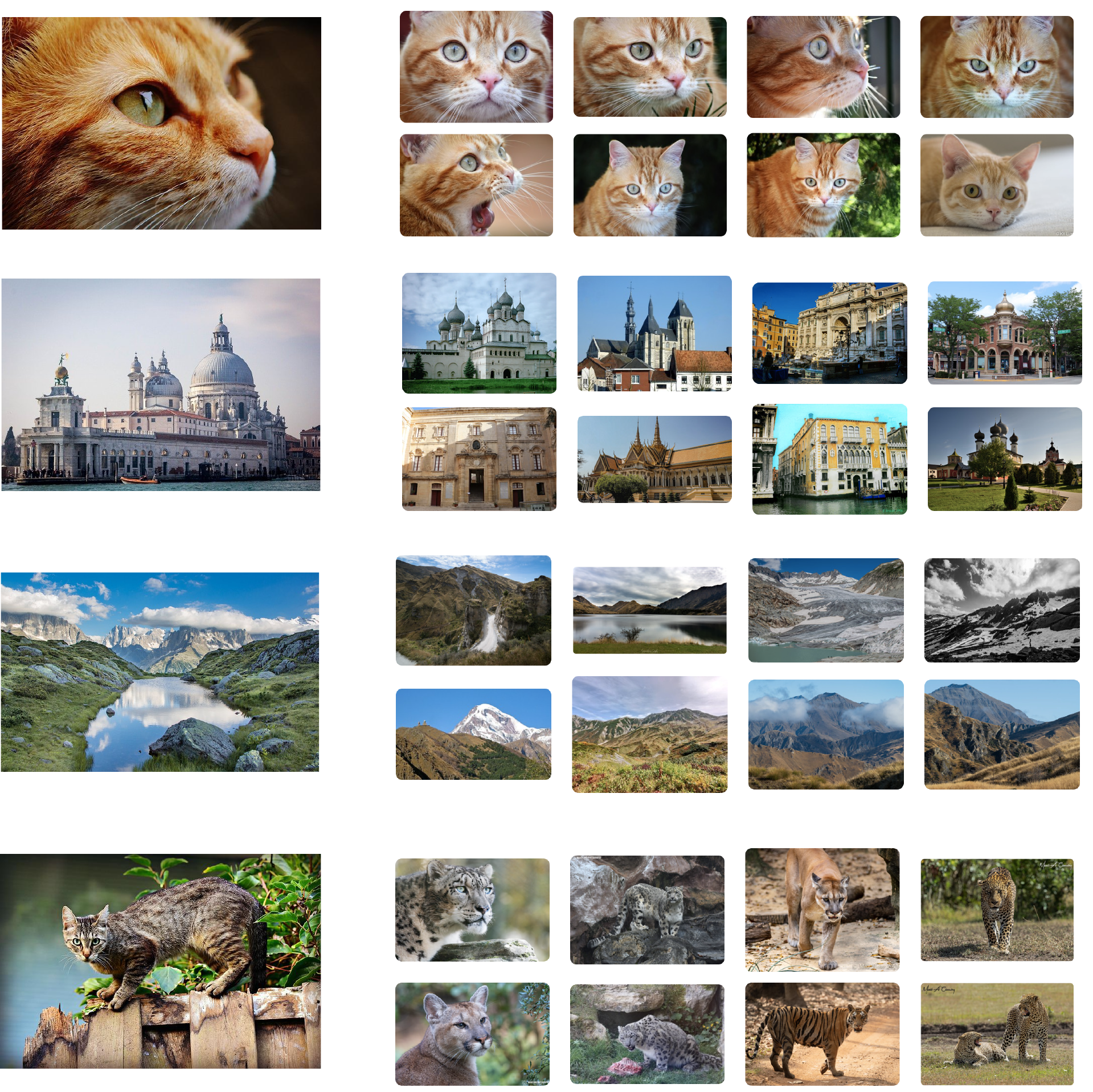}
\caption{Results. (Left) Source images. (Right) Top-$8$ retrieved images by CLEAR. (Bottom) A failure example. Although the input image is a cat, the retrieved images are snow leopards.}
 \label{fig: result-clear}
\end{figure}

The online demo is available at \url{https://clear.joisino.net}. Users can upload images, and the system retrieves similar images from Flickr. The functionality is simple. We stress again that Flickr does not provide an official similar image search system, and therefore users could not enjoy this feature so far. The highlight is that although the author is not employed by Flickr nor has privileged access to Flickr's server, CLEAR realizes the new feature of Flickr on a user-side.

\vspace{0.1in}
\noindent \textbf{Results.} Figure \ref{fig: result-clear} shows results of the demo program. It can be observed that CLEAR retrieves visually similar images to the source images. The bottom row shows a failure example. Although the input image is a cat, CLEAR confused this image with a snow leopard. Once CLEAR becomes confident to some extent, it is difficult to escape from the failure mode because it adopts the greedy algorithm without explicit exploration. Improving stability while keeping efficiency is an important future direction.

\section{Conclusion}

This demonstration provides the first practical user-side image search system, CLEAR. The proposed system is implemented on the client-side and retrieves similar images from Flickr, although Flickr does not provide an official similar image search engine or corresponding API. CLEAR adaptively generates effective text queries and realizes a similar image search. CLEAR does not require initial crawling, indexing, or any backend servers. Therefore, an end-user can deploy their own search engine easily.

\chapter{Active Learning from the Web}

\section{Introduction}

The revolution of deep learning has greatly extended the scope of machine learning and enabled many real-world applications such as computer-aided diagnosis \cite{chen2016mitosis, du2018breast,yu2018recurrent}, anomaly detection, \cite{akcay2018ganomaly,ruff2020deep} and information retrieval \cite{mcauley2015image,song2016deep}. However, the high cost of annotating data still hinders more applications of machine learning. Worse, as the number of parameters of deep models is larger than traditional models, it consumes more labels and, thus, more costs \cite{hestness2017deep,kaplan2020scaling,rosenfeld2020constructive,zhai2022scaling,geirhos2021partial}. Many solutions have been proposed for alleviating the cost of labeling.

First, many unsupervised and self-supervised methods that do not require manual labeling have been proposed \cite{chen2020simple,he2020momentum,baevski2020wav2vec}. However, the applications of pure unsupervised methods are still limited, and many other applications still call for human annotations to some extent. Specifically, unsupervised representation learning offers task-agnostic representations, and we need to extract task-specific information from them using labeled data.

Another popular approach is pre-training. Large scale supervised datasets such as ImageNet \cite{russakovsky2015imagenet} and JFT-3B \cite{zhai2022scaling} and/or unsupervised and self-supervised methods \cite{chen2020simple,devlin2019bert} can be used for pre-training. These methods significantly reduce the number of required labels. However, we still need several labels, even with such techniques. These techniques are orthogonal to the problem setting considered in this chapter, and our approach can be used with such pre-training techniques. In particular, we focus on labeling data after applying such techniques that reduce the sample complexity.

The most direct approach to reducing the burden of labeling would be active learning \cite{settles2009active,lewis1994sequential,tong2001support,gal2017deep}. Among many approaches to active learning \cite{angluin1987queries,atlas1989training}, we focus on pool-based active learning, which gathers unsupervised data and iteratively selects data to be labeled so that the number of required labels is minimized. Many pool-based active learning methods have been proposed for decades \cite{lewis1994sequential,tong2001support,gal2017deep,raj2022convergence}. Most researches on active learning focus on the procedure after collecting the pool of data; in particular, most works focus on the criterion of selecting data, i.e., the acquisition function of active learning \cite{gal2017deep,joshi2009multi,hoi2006batch}. However, how to define the pool of data has been less explored. Collecting unsupervised data is not straightforward. Naive crawling on the Web may collect only negative or irrelevant samples and harms the performance of active learning greatly. Making a very large pool on a local server may consume a lot of communication costs and time, and it is not economical or infeasible to build a very large pool in the traditional approach. For example, we build a middle-size pool with $10^5$ images in the experiments, and it consumes about 20 GB of storage. Building a ten or hundred times larger pool on a local server or on cloud storage may not be economical, especially when the target task is small. The cost scales linearly with the size of the pool. It is infeasible to host a pool with $10^{10}$ images in any case.

\begin{figure*}[tb]
\centering
\includegraphics[width=\hsize]{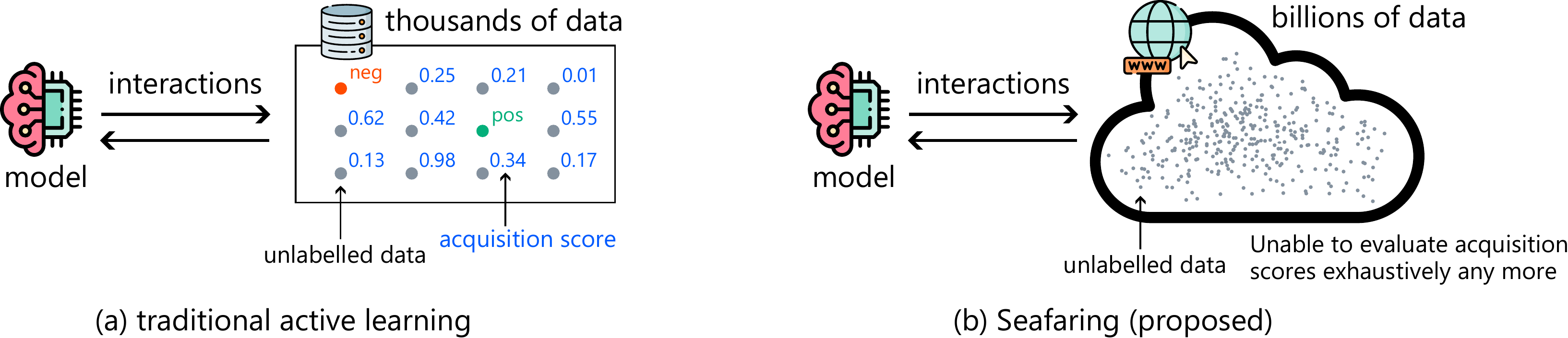}
\caption{\textbf{Illustrations of traditional active learning and Seafaring.} Traditional methods use a local database as the pool of active learning. Seafaring uses a huge database on the Internet as the pool of active learning. As the pool of Seafaring is extremely large, it is more likely that informative data exist, but it is challenging to find it efficiently.}
\label{fig: illustration}
\end{figure*}

\begin{table}[tb]
    \centering
    \caption{Notations.}
    \begin{tabular}{ll} \toprule
        Notations & Descriptions \\ \midrule
        $\mathcal{X}$ & The input space, i.e., images. \\
        $\mathcal{Y}$ & The label space, i.e., $\{0, 1\}$. \\
        $\mathcal{U} \subset \mathcal{X}$ & The pool of active learning. \\
        $\mathcal{U}_{\text{Flickr}} \subset \mathcal{X}$ & The set of all images on Flickr. \\
        $\mathcal{L} \subset \mathcal{X}$ & The set of initial labelled data. \\
        $\mathcal{O}\colon \mathcal{X} \to \mathcal{Y}$ & The labeling oracle. \\
        $f_{\theta}\colon \mathcal{X} \to [0, 1]^\mathcal{Y}$ & The classification model. \\
        $a\colon \mathcal{X} \to \mathbb{R}$ & The acquisition function. \\
        $B \in \mathbb{Z}_+$ & The labeling budget of active learning. \\
        \bottomrule
    \end{tabular}
    \label{tab: notations}
\end{table}

In this chapter, we regard items on the Web as a very large pool of unlabelled data and advocate the use of this large pool for active learning (Figure \ref{fig: illustration}). In the experiments, we use all the images on Flickr as the pool. This pool contains more than ten billion ($10^{10}$) images, which is several orders of magnitude larger than the pools that have ever been seen in the literature. As the pool is extremely large, it is likely that informative items exist in the pool. Besides, our approach does not require the user to explicitly build the pool on a local server before the application of active learning, but it retrieves informative items on the fly during active learning. Therefore, it further reduces the cost of preparation of data in addition to the reduction of the number of required labels.

However, they are many challenges to realizing this goal.
\begin{itemize}
\item As there are too many items in the pool, it is infeasible to compute the scores of all items exhaustively.
\item We do not even have the list of items on the Web or access to the database directly. The only way we access items is via the use of limited search queries.
\end{itemize}
We propose Seafaring (\underbar{Sea}rching \underbar{f}or \underbar{a}ctive lea\underbar{r}n\underbar{ing}), which uses a user-side information retrieval algorithm to overcome these problems. User-side search algorithms were originally proposed for building customized search engines on the user's side \cite{bra1994information,sato2022retrieving}. In a nutshell, our proposed method uses the acquisition function of active learning as the score function of the user-side engine and retrieves informative items to be labeled efficiently without direct access to external databases.

The contributions of this chapter are summarized as follows.
\begin{itemize}
\item We advocate the use of items on the Web as the pool of active learning. This realizes several orders of magnitude larger pools than existing approaches and enables us to find informative examples to be labeled.
\item We propose Seafaring, an active learning method that retrieves informative items from the Web.
\item We show the effectiveness of our proposed methods with a synthetic environment and real-world environment on Flickr, comparing with the existing approach that uses a small pool.
\item We show a supervising effect of our simple approach, i.e., Seafaring retrieves positive labels from a myriad of irrelevant data.
\end{itemize}

\begin{tcolorbox}[colframe=gray!20,colback=gray!20,sharp corners]
\textbf{Reproducibility}: Our code is publicly available at \url{https://github.com/joisino/seafaring}.
\end{tcolorbox}

\section{Problem Setting} \label{sec: setting}

\begin{table}[tb]
    \centering
    \caption{Size of the pools of unlabelled data. The pools of existing works contain hundreds of thousands of data. The pool of Seafaring is several orders of magnitude larger than pools that have ever been seen in the literature.}
    \begin{tabular}{ll} \toprule
        Dataset & Pool Size \\ \midrule
        Reuters \cite{tong2001support} & 1,000 \\
        NewsGroup \cite{tong2001support} & 500 \\
        ImageCLEF \cite{hoi2006batch} & 2,700 \\
        Pendigits \cite{joshi2009multi} & 5,000 \\
        USPS \cite{joshi2009multi} & 7,000 \\
        Letter \cite{joshi2009multi} & 7,000 \\
        Caltech-101 \cite{joshi2009multi} & 1,500 \\
        CACD \cite{wang2017cost} & 40,000 \\
        Caltech-256 \cite{wang2017cost} & 24,000 \\
        MNIST \cite{gal2017deep} \cite{kirsch2019batchbald} & 50,000 \\
        ISBI 2016 \cite{gal2017deep} & 600 \\
        EMNIST \cite{kirsch2019batchbald} & 94,000 \\
        CINIC-10 \cite{kirsch2019batchbald} & 160,000 \\
        MNIST \cite{beluch2018power} & 2,000 \\
        CIFAR-10 \cite{beluch2018power} & 4,000 \\
        CIFAR-10 \cite{beluch2018power} & 20,000 \\
        Diabetic R. \cite{beluch2018power} & 30,000 \\
        ImageNet \cite{beluch2018power} & 400,000 \\
        Synthetic (ours) & 100,000 \\
        Flickr (ours) & $\ge$ 10,000,000,000 \\
        \bottomrule
    \end{tabular}
    \label{tab: size}
\end{table}

\begin{figure*}[tb]
\centering
\includegraphics[width=\hsize]{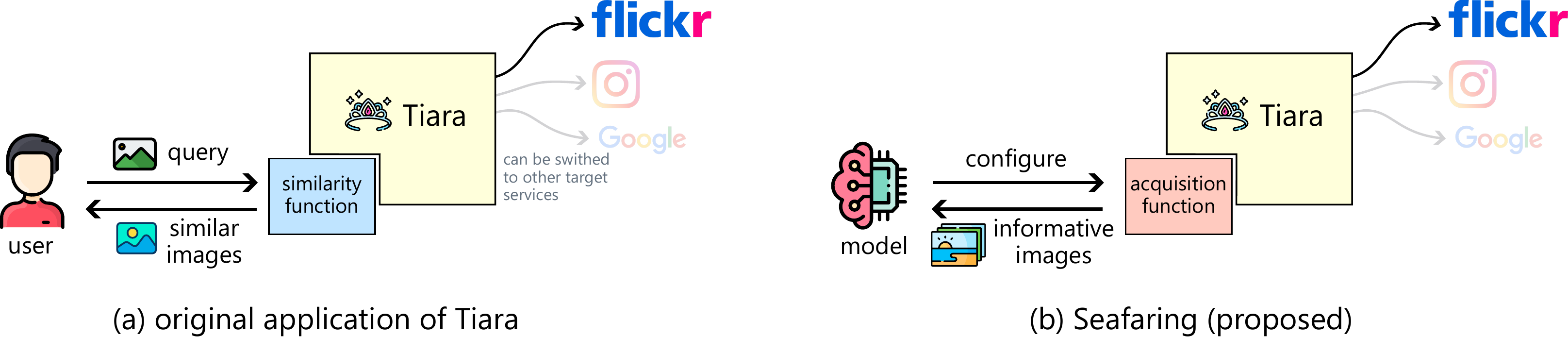}
\caption{\textbf{Illustrations of Tiara and Seafaring.} Originally, user-side search engines and Tiara were proposed to build similar image search engines on the user's side for the Web services without official similar image search engines. The advantages of these methods are that they work with user-defined score functions and they work without any privileged access to the database. Seafaring uses Tiara in active learning settings. Seafaring uses an acquisition function of active learning as the score function of the user-side search engine and retrieves informative images for the current model.}
\label{fig: tiara}
\end{figure*}

Let $\mathcal{X}$ and $\mathcal{Y}$ be the input and label domains, respectively. In this chapter, we focus on the binary classification of image data, though our approach can be extended to other tasks. Therefore, $\mathcal{X}$ is the set of images and $\mathcal{Y} = \{0, 1\}$. Let $\mathcal{U} = \{x_i\} \subset \mathcal{X}$ denote the set of unlabelled data, let $\mathcal{L} = \{(x_i, y_i)\} \subset \mathcal{X} \times \mathcal{Y}$ denote the set of initial labelled data. Let $\mathcal{O}\colon \mathcal{X} \to \mathcal{Y}$ be the labeling oracle function, which takes an image as input and outputs the ground truth label. In reality, the oracle $\mathcal{O}$ asks human annotators to reveal the ground truth label. The evaluation of the oracle $\mathcal{O}$ is costly because it requires human resources. We want to evaluate the oracle function as few times as possible. We limit the number of the evaluation by $B \in \mathbb{Z}_+$. $B$ is called the budget of active learning.

Active learning involves $B$ iterations. Let $\mathcal{L}_0 = \emptyset$ be the initial set of selected data. In $i$-th iteration, the model $f_{\theta}$ is trained with the current set of labelled images $\mathcal{L} \cup \mathcal{L}_{i-1}$, and an active learning method selects a unlabelled sample $x_i \in \mathcal{U}$ from the pool. Then, the oracle (e.g., human annotators) reveals the label $y_i = \mathcal{O}(x_i)$. The set of labelled data is updated by setting $\mathcal{L}_i = \mathcal{L}_{i-1} \cup \{(x_i, y_i)\}$. After $B$ iterations, the performance of the model is evaluated with the obtained training data $\mathcal{L} \cup \mathcal{L}_{B}$. The aim of active learning is to obtain a good model with a small $B$.

Typical active learning methods select data based on acquisition function $a\colon \mathcal{X} \to \mathbb{R}$ and set $x_i = \argmax_{x \in \mathcal{U}} a(x)$. The design of the acquisition function depends on the method, and many acquisition functions have been proposed, such as the entropy function \cite{settles2009active}, least confidence \cite{settles2009active}, margin sampling \cite{settles2009active}, expected error reduction \cite{roy2001toward}, expected model change \cite{settles2007multuple}, and ensemble-based scores \cite{beluch2018power}. The family of acquisition functions, including entropy, least confidence, and margin sampling, is sometimes referred to as uncertainty sampling, and they are all equivalent in the binary classification setting \cite{settles2009active,nguyen2022how}. In this chapter, we do not pursue the design of the acquisition function but use the existing one. Instead, we focus on the design of the pool $\mathcal{U}$.

\section{Proposed Method}

In this section, we introduce our proposed method, Seafaring, which explores the sea of data on the Internet in search of informative data.

Seafaring uses all the images in an image database on the Internet as the pool $\mathcal{U}$. In the experiments, we use the image database $\mathcal{U}_{\text{Flickr}}$ of Flickr as the pool. It should be noted that our method is general, and other databases such as Instagram $\mathcal{U}_{\text{Instagram}}$ and DeviantArt $\mathcal{U}_{\text{DeviantArt}}$ can be used as the target database depending on the target task. We focus on Flick because it contains sufficiently many images for general tasks, and its API is easy to use in practice. As $\mathcal{U}_{\text{Flickr}}$ contains more than $10$ billion images of various motifs and styles, it is likely that relevant and informative images exist in them. Existing methods use smaller pools, which contain thousands of items (Table \ref{tab: size}). As these pools are task-specific, existing active learning methods need to design the pool for each task. Existing methods for active learning assume that a task-specific pool is available for free, but we emphasize that this is not the case in many cases, and designing and building task-specific pools is costly. For example, suppose we build an image recommender system and train a model that classifies images that a user likes for each user. It is not obvious what kind of images each user likes beforehand, and preparing an appropriate pool of unlabelled data is a difficult task. The pool becomes very large if we gather a pile of images from the Web without any care, and traditional active learning is infeasible in such a case. By contrast, Seafaring handles an extremely large pool, and we can use it for many general tasks without any additional effort. Thus Seafaring further reduces the burden of data preparation. In addition, if we update the model regularly with continual or life-long learning methods \cite{li2016learning,zenke2017continual,wu2019large}, existing active learning methods need to maintain the pool. By contrast, Seafaring does not require the users to maintain the pool because the users of Flickr add new images to the Flickr database $\mathcal{U}_{\text{Flickr}}$ every day, and the pool is maintained up-to-date automatically and autonomously. Thus Seafaring always uses an up-to-date database without any special effort.

In each iteration of active learning, Seafaring searches the best image with respect to the acquisition function $a$ from $\mathcal{U}_{\text{Flickr}}$ and sets $x_i \leftarrow \argmax_{x \in \mathcal{U}_{\text{Flickr}}} a(x)$. However, as $\mathcal{U}_{\text{Flickr}}$ is extremely large, we cannot evaluate $a(x)$ for each image $x \in \mathcal{U}_{\text{Flickr}}$. It is impossible to even download all images in $\mathcal{U}_{\text{Flickr}}$ or even list or enumerate $\mathcal{U}_{\text{Flickr}}$. Employees of Flickr could build bespoke search indices and databases for active learning purposes, but as we and most readers and users of Seafaring are not employees of Flickr, the approach of building an auxiliary index of the Flickr database is infeasible.

To overcome this issue, Seafaring uses a user-side search algorithm, namely, Tiara \cite{sato2022retrieving}. Tiara retrieves items from external databases on the Internet. The original motivation of Tiara is that most users of Web services are not satisfied with the official search engine in terms of its score function or its interface, and Tiara offers a way for users to build their own search engines for existing Web services. Seafaring uses Tiara differently. Seafaring sets the acquisition function $a$ of active learning as the score function of Tiara and thereby retrieves informative images for model training (not for users). Figure \ref{fig: tiara} shows an illustration of Seafaring and the difference between the original application of Tiara and Seafaring. As Tiara works without privileged access to external databases (i.e., $\mathcal{U}_{\text{Flickr}}$), and ordinary users can run it, it suits our situation well.

We introduce the algorithm of Tiara briefly for the completeness of this chapter. Tiara takes a score function $s\colon \mathcal{X} \to \mathbb{R}$ as input. The aim of Tiara is to retrieve an image $x^*$ that maximizes $s$ from an external database. Tiara assumes that a tag or text-based search system that takes text as input and returns a set of images is available for the external database. In Flickr, Tiara and Seafaring use \texttt{flickr.photos.search} API\footnote{\url{https://www.flickr.com/services/api/flickr.photos.search.html}} for the tag-based search system. Tiara formulates the search problem as a multi-armed bandit problem where a tag is an arm, and the score function is the reward. Tiara uses LinUCB algorithm \cite{li2010contextual} to solve the multi-armed bandied problem with pre-trained tag features based on GloVe \cite{pennington2014glove}. Intuitively, Tiara iteratively queries tags (exploration), evaluates images, finds promising tags, and finds out (sub-)optimal images by querying promising tags (exploitation). Please refer to the original paper \cite{sato2022retrieving} for more details.

\begin{figure*}[tb]
\centering
\includegraphics[width=\hsize]{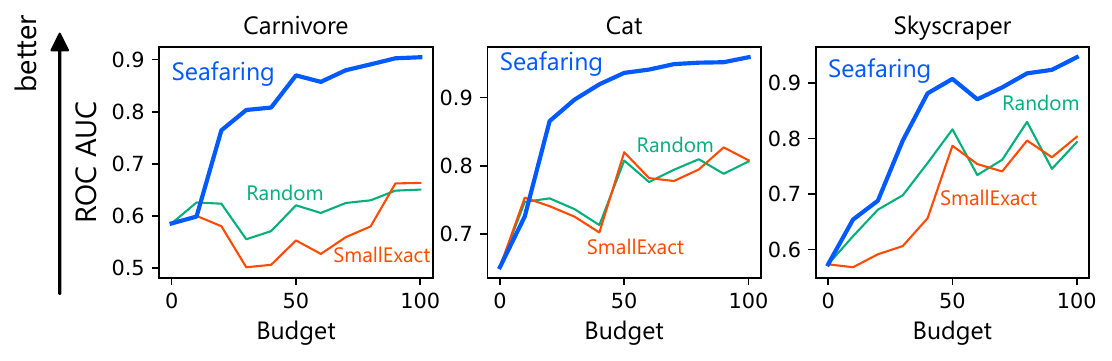}
\caption{\textbf{Results of OpenImage dataset.} Each curve reports the performance of each method. Point $(x, y)$ means that the method achieves $y$ AUC within $x$ access to the label oracle. Higher is better. These results show that Seafaring is more effective than the baseline methods.}
\label{fig: openimage-seafaring}
\end{figure*}

Seafaring uses the exponential of the entropy function \begin{align} \label{eq: acquisition}
    a(x) = s(x) = \exp\left(- \gamma \sum_{i \in \mathcal{Y}} f_\theta(x)_i \log f_\theta(x)_i\right)
\end{align}
as the acquisition function and the score function of Tiara. $f_\theta(x)_i \in [0, 1]$ is the output probability of class $i$ by the model $f_{\theta}$. $\gamma \in \mathbb{R}$ is a hyperparameter representing the temperature. We confirmed that the algorithm was insensitive to the choice of this parameter and set $\gamma = 4$ throughout the experiments. The entropy function is one of the most popular acquisition functions of active learning \cite{holub2008entropy,nguyen2022how}. Intuitively, if the entropy function is high, the model $f_\theta$ is not certain about the label of this data, so labeling this data is likely to provide much information for the model. Our acquisition function $a$ is monotone with respect to the entropy function. We apply the exponential function because data with high scores are more important than data with low scores, e.g., whether the entropy is $0.8$ or $0.9$ is more important than whether the entropy is $0.2$ or $0.3$ for searching high entropy data. The exponential function magnifies the difference in the high score regime. It should be noted that monotone transformation does not change the results if we evaluate the scores of all data and select the exact maximum data, but the results of Seafaring change by a monotone transformation because of the exploration-exploitation trade-off of the bandit algorithm and the inexact search scheme. Seafaring is general and can be combined with any other acquisition functions such as the ensemble score \cite{beluch2018power} and BALD \cite{gal2017deep}. As the design of the acquisition function is not the primal interest of this chapter, we leave further exploration of acquisition functions for future work.

Seafaring is indeed a simple application of Tiara to active learning. We emphasize that the formulation of Web-scale active learning is the core interest of this chapter, and the optimization method is secondary. Note also that we tried a more complicated method in preliminary experiments, but we found out that the vanilla method was already effective and had preferable features, as we analyze in the experimental section. So we keep our proposed method as simple as possible for usability and customizability.

\section{Experiments}

We show that Seafaring is more effective than the traditional approach that uses a small and fixed pool of unlabelled data. We conduct control experiments where the baseline method uses the same settings to contrast a huge pool with a small pool. Specifically, we let the baseline method use the same acquisition function and samples from the same database as those of Seafaring. In the analysis, we also show that Seafaring automatically balances the ratio of labels without any special mechanism.

\subsection{Experimental Setups}

We aim at training ResNet-18 \cite{he2016deep} for various tasks. One shot of an experiment is conducted as follows: Initially, one positive sample and one negative sample are given, and they form the initial training set $\mathcal{L}$. Through the experiments, an active learning method can evaluate the oracle $\mathcal{O}$ at most $B$ times. We set $B = 100$ throughout the experiments. In a loop of an experiment, the ResNet model $f_{\theta}$ is initialized by the ImageNet pre-trained weights\footnote{\texttt{ResNet18\_Weights.IMAGENET1K\_V2} in \url{https://pytorch.org/vision/stable/models.html}.}, and then trained with $\mathcal{L}$. The training procedure uses the stochastic gradient descent (SGD) with learning rate $0.0001$ and momentum $0.9$ and with $100$ epochs. Then, the evaluation program measures the accuracy of the model using test data. The accuracy of the model is measured by ROC-AUC. Note that the results of the evaluation are not fed back to the active learning method, but only the evaluator knows the results so as not to leak the test data. After the evaluation, an active learning method selects one unlabeled sampled $x_i$ from the pool using the trained model and an acquisition function. After the selection of the sample, the label oracle $\mathcal{O}(x_i)$ is evaluated and the label $y_i$ is notified to the active learning method, and $(x_i, y_i)$ is inserted to $\mathcal{L}$. This loop is continued until the method exhausts the budget, i.e., $B$ times. The goal of an active learning method is to maximize the accuracy of the model $f_{\theta}$ during the loop and at the end of the loop. 

It should be noted that the ResNet-18 model is relatively larger than the models used in existing works \cite{wang2017cost,gal2017deep}, where shallow multi-layered perceptron and convolutional neural networks are used, and it further introduces the challenge in addition to the huge pool because each evaluation of the acquisition function costs more.

We use two baseline methods. The \textbf{SmallExact} method is an active learning method with the traditional approach. Specifically, it first samples $1000$ unlabelled data from the database that Seafaring uses, i.e., $\mathcal{U}_{\text{Flickr}}$ or $\mathcal{U}_{\text{OpenImage}}$, where $\mathcal{U}_{\text{OpenImage}}$ is the pool of the synthetic database we will use in the next subsection. SmallExact uses these data for the pool of active learning. It uses the same acquisition function (Eq. \eqref{eq: acquisition}) as Seafaring to make the conditions the same and to clarify the essential difference between our approach and the existing approach. The only differences between Seafaring and SmallExact are (i) the size of the unlabelled pool and (ii) Seafaring carries out an inexact search, while SmallExact evaluates the scores of the data in the pool exhaustively and exactly selects the best data from the pool. The exact search approach of SmallExact is infeasible if the pool or the model is large. In the following, we show that Seafaring performs better than SmallExact even though Seafaring does not carry out an exact search. The \textbf{Random} method selects a random sample from the same pool of Seafaring uses, i.e., $\mathcal{U}_{\text{Flickr}}$ or $\mathcal{U}_{\text{OpenImage}}$, in each loop. Both Seafaring and Random use the same pool, but Random does not use active learning. The difference between the performances of Random and Seafaring shows the benefit of active learning.

We conduct the experiments on NVIDIA DGX-1 servers, and each experiment runs on a V100 GPU.

\subsection{OpenImage Dataset}

\begin{figure*}[tb]
\centering
\includegraphics[width=\hsize]{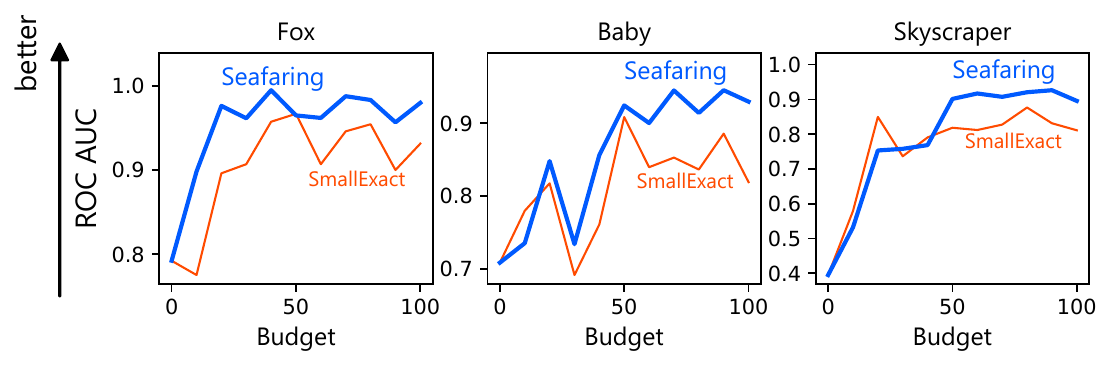}
\caption{\textbf{Results of Flickr dataset.} Each curve reports the performance of each method. Higher is better. These results show that Seafaring is more effective than the traditional approach with a small pool in the real-world Flickr environment.}
\label{fig: flickr-seafaring}
\end{figure*}

We first construct a synthetic and middle-size image database using the OpenImage Dataset \cite{kuznetsova2020open} to quantitatively evaluate the methods, confirm the general trends, and analyze the results. We also aim to make the experiments easy to reproduce by keeping the size of the environment not so large. We first sample $100\,000$ images from the OpenImage Dataset. Let $\mathcal{U}_{\text{OpenImage}}$ denote the set of these images. It should be noted that even $\mathcal{U}_{\text{OpenImage}}$ is larger than many of the pools used in existing work (Table \ref{tab: size}). We use the tag annotations of the OpenImage Dataset to construct a virtual search engine for Tiara, following the original paper \cite{sato2022retrieving}. We define a classification task with a tag $t$. In each task, the images with the tag $t$ are positive samples, and the other images are negative samples. The goal of the task is to train a model that classifies positive and negative samples accurately.

We set the number of iterations of LinUCB in Tiara to $1000$ in this experiment. The number indicates the number of evaluations of the model $f_{\theta}$ and the number of search queries in each loop of active learning. Although the labeling oracle $\mathcal{O}$ is the most costly operation, the evaluation of the model $f_{\theta}$ is also costly when the model is large. So it is desirable to keep this number small. As the SmallExact method evaluates the acquisition scores of all of the $1000$ samples in the pool in each iteration, the number of evaluations of the model is the same for both methods in this experiment.

We run each task $5$ times with different random seeds and report the average performances. Figure \ref{fig: openimage-seafaring} shows the accuracies of Seafaring and the baseline methods for three tasks, t = Carnivore, t = Cat, and t = Skyscraper. Seafaring outperforms both baselines in all tasks. In t = Carnivore, both baselines fail to find relevant images at all, and the performance does not increase with more labels. By contrast, Seafaring efficiently finds out informative data, and the performance quickly grows with small labels in all tasks.

\subsection{Flickr Environment}

We use the real-world Flickr database $\mathcal{U}_{\text{Flickr}}$ as the pool in this experiment. As we have mentioned before, $\mathcal{U}_{\text{Flickr}}$ contains more than ten billion images, and this is several orders of magnitude larger than the largest pool of active learning in the literature (Table \ref{tab: size}). We emphasize that we are not employees of Flickr nor have privileged access to the Flickr database. Nevertheless, we show that Seafaring retrieves informative images from the Flickr database.

We define a classification task with a tag $t$ but in a different way from the OpenImage experiment. Specifically, we gather $10$ images  $\mathcal{X}_t$ with tag $t$ from the OpenImage dataset, extract features of images using pre-trained ResNet18, and define images whose maximum cosine similarity between $\mathcal{X}_t$ is larger than a threshold as positive samples, and other images as negative samples. Note that $\mathcal{X}$ is not revealed to us and the learner, but only the evaluator knows it. Intuitively, images with features of tag $t$ are positive samples, and other images are negative samples. It should be noted that Flickr also provides tags for images, but we found that they are sometimes missing and noisy. For example, some skyscraper images are not tagged as a skyscraper, i.e., positive samples are wrongly labeled as negatives, and the evaluation becomes not reliable. So we do not use the tags in Flickr to define the task but adopt the similarity-based tasks. 

In this experiment, we set the number of iterations of LinUCB to $100$ to reduce the number of API queries to the Flicker server. As the SmallExact method cannot reduce the number of evaluations but conducts an exact search, we keep the number of evaluations of the SmallExact method $1000$. Therefore, in this experiment, the SmallExact method evaluates the model $10$ times many times in each loop of active learning, and the setting is slightly advantageous to the SmallExact method. We show that Seafaring nevertheless performs better than SmallExact.

Figure \ref{fig: flickr-seafaring} shows the accuracies of Seafaring and SmallExact for three tasks, t = Skyscraper, t = Fox, and t = Baby. Seafaring outperforms SmallExact in all tasks.

\begin{figure}[tb]
\centering
\includegraphics[width=\hsize]{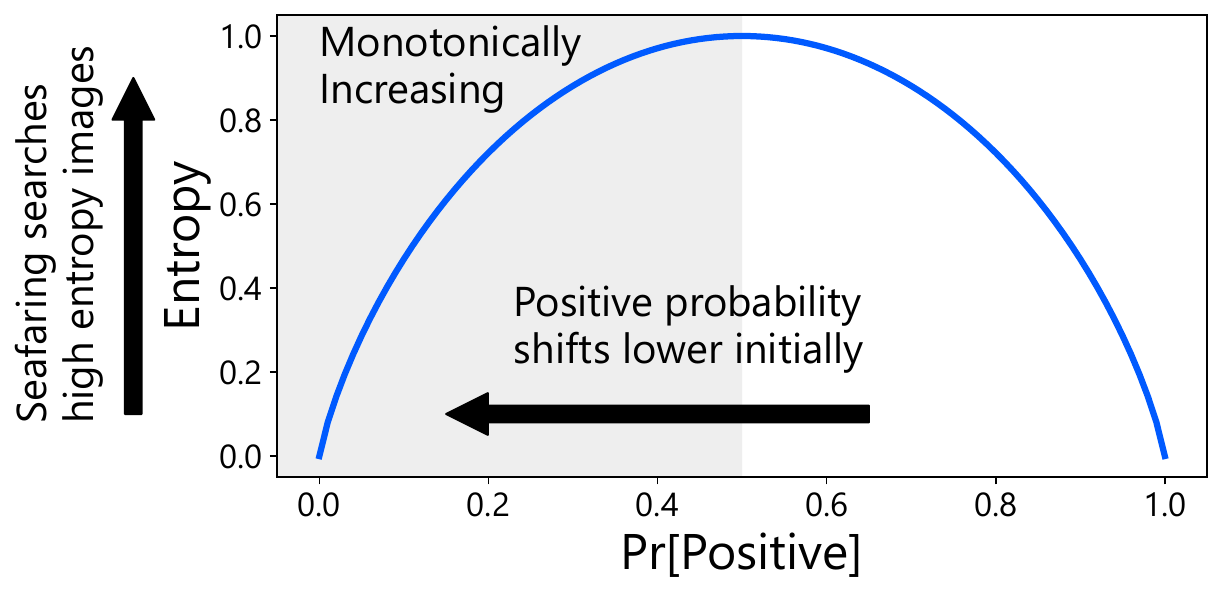}
\caption{\textbf{Illustration of the entropy function.} In the region of $\text{Pr}[\text{Positive}] \le 0.5$, the entropy function is monotonically increasing. In the initial phase of active learning, the model's output probability shifts lower, and the positive probability is less than $0.5$ for most data. Maximizing the entropy is equivalent to maximizing the positive probability in this regime.}
\label{fig: entillust}
\end{figure}

\subsection{Analysis: Seafaring Retrieves Positive Data When They Lack}

As the pool is huge and diverse, most data in the pool are irrelevant to the task, and only a fraction of the data are positive or difficult negative cases. If we randomly sample data from the pool, the obtained labels are negative with high probability, and we obtain little information. For example, in the classification of Carnivore images, most data, such as airplane, car, and building images, are irrelevant to the task. We want to gather relevant images such as cat images (positive samples) and horse images (difficult negative samples). However, it is difficult to distinguish relevant samples from other data until the model's performance is high. Gathering relevant samples and training good models are like a chicken and egg problem. Besides, as the pool contains much more irrelevant images than relevant images, it is all the more difficult to filter out irrelevant images. Let us clarify the difficulty. Suppose we have an auxiliary model that classifies whether a sample is relevant to the task, and this model's accuracy is $0.99$. As there are $10^{10}$ irrelevant images in the pool, it classifies $10^{10} \times (1 - 0.99) = 10^8$ irrelevant data as relevant. If there are $10^4$ relevant images in the pool, the morel classifies $10^4 \times 0.99$ relevant data as relevant. Thus, the majority of the selected data are irrelevant regardless of the high accuracy of the model. It indicates that finding out relevant data is crucial for active learning methods when the pool is huge and unbalanced.

Seafaring does not have an explicit mechanism to cope with this challenge. Nevertheless, we found that Seafaring implicitly solves this problem. Seafaring indeed struggles with finding relevant data in the initial iterations, and the ratio of negative data to positive data grows drastically. After several iterations, most data Seafaring has are negative. The target labels of the training data with which the model $f_{\theta}$ is trained are negative, and the model's output drastically tends to be negative. In other words, the model becomes very confident that most data are negative. Recall that the acquisition function is high when the model's output is not confident, and the entropy function and the acquisition function monotonically increase as the positive probability increases in the regime of the low positive probability (Figure \ref{fig: entillust}). Therefore, in this regime, low confidence and high positive probability are equivalent, and Seafaring searches for data that are likely to be positive. We observe that Seafaring once enters this phase in the initial phase and then finds some positive samples by maximizing the positive probability, the data become balanced, the model becomes accurate, and after that, Seafaring searches informative data with which the accurate model is confused.

\begin{figure*}[tb]
\centering
\includegraphics[width=\hsize]{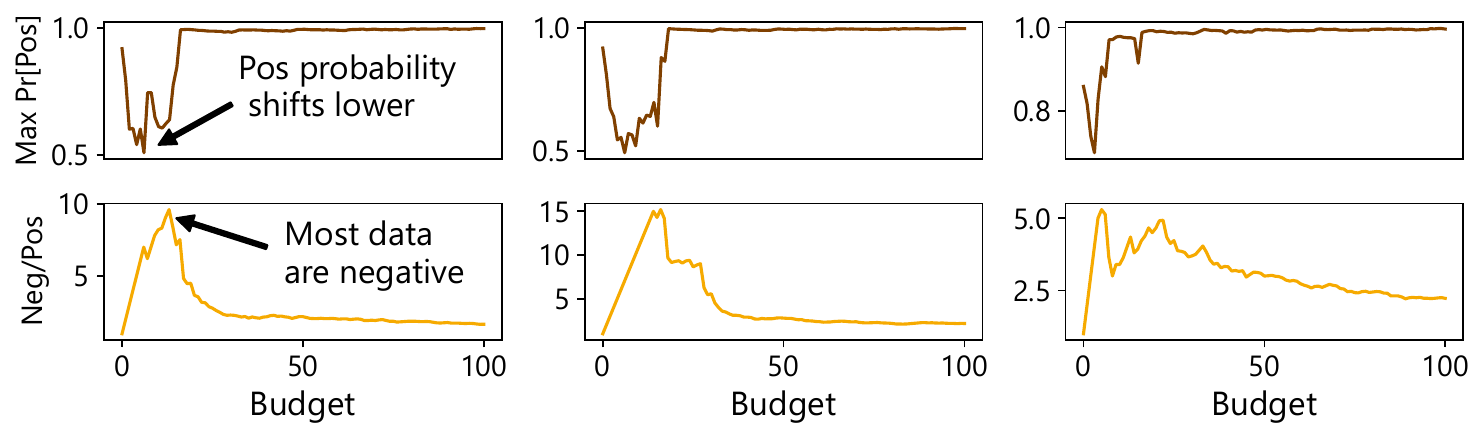}
\caption{\textbf{Transition of maximum positive probability and the ratio of positive and negative samples of Seafaring.} (Top) The maximum positive probability of the data Seafaring evaluates for each loop. In most iterations, the highest probability is around $1.0$, and the curve sticks to the top of the plot. But in the initial phase, even the highest positive probability is around $0.5$ because the model's outputs stick to negative. (Bottom) The ratio of the number of negative samples to the number of positive samples in the labeled data $\mathcal{L}$. Seafaring finds only negative samples at the beginning of iterations, and the ratio grows rapidly. Seafaring searches images with the highest positive probability in this regime. After some positive samples are found, the ratio drops and stabilizes. }
\label{fig: entropy}
\end{figure*}

Figure \ref{fig: entropy} shows the transition of statistics in the course of active learning. In the first to $30$-th iterations, most collected data are negative. The model becomes certain that most data are negative, and the top panels show that even the most likely image has a positive probability of around $0.5$. It indicates that most samples have positive probabilities of less than $0.5$. As Figure \ref{fig: entillust} shows, in this regime, maximizing the acquisition function is equivalent to finding the most likely image to be positive. Seafaring maximizes the positive probability of the data being labeled. After some trials, Seafaring finds out some positive data, and the ratio of negative and positive samples decreases. After the $30$-th iteration, the model's output probability gets back to normal, and Seafaring focuses on retrieving the most uncertain data.

We did not expect this phenomenon. We considered that we needed an explicit mechanism to retrieve positive data in the initial phase. However, it turned out that such a mechanism was unnecessary. This finding let us keep Seafaring simple and easy to implement.

\section{Discussions} \label{sec: discussion}

In the experiments, we used the Flickr database as the pool. As we mentioned earlier, our algorithm is general and can be combined with other search engines and Web services, such as Instagram, DeviantArt, and Google Image Search. The most attractive choice would be general-purpose image search engines, such as Google Image Search and Bing Image Search, because they retrieve images from the entire Web, and the pool becomes even larger. We did not use them because of their cost. For example, Google Search API charges 5 dollars per 1000 queries\footnote{\url{https://developers.google.com/custom-search/v1/overview}}, and Bing Image Search API charges 4 to 7 dollars per 1000 queries\footnote{\url{https://azure.microsoft.com/ja-jp/pricing/details/cognitive-services/search-api/}}. As we use $100$ to $1000$ queries per iteration and $100$ iterations for each run of the experiments, each run would take 40 to 500 dollars for the search API. These API costs are much higher than labeling costs in most applications, and it makes these APIs not attractive in most situations. By contrast, Flickr API offers $3600$ queries per hour for free\footnote{\url{https://www.flickr.com/services/developer/api/}}. Therefore, the entire process of active learning takes $3$ to $30$ hours and is for free, and the bottleneck of the computational cost is training deep models, and the bottleneck of the expense is labeling, which is inevitable in all tasks. Besides, Flickr API can control the licenses of the result images. For example, specifying ``\texttt{license=9, 10}'' restricts the results to public domain images, which resolves license issues. These are why we used Flickr in the experiments and why Flickr API is attractive in most applications. However, although Flickr hosts more than ten billion images, it lacks some types of images, such as illustrations and celebrity portraits. If the task at hand is relevant to such data, other databases, such as Instagram and DeviantArt, should be used. One of the interesting future works is to combine many search engines and services to make Seafaring more robust while keeping the API cost low. An attractive future avenue is a cost-sensitive approach that basically uses ``free'' databases such as Flickr and uses ``costly'' targets such as Google Image Search only when the costly APIs are promising, which alleviates the API cost issue.

We emphasize that Seafaring is applicable to any domain only if an API takes a tag as input and outputs data samples. For example, \texttt{/2/tweets/search/recent} API of Twitter can be used for NLP tasks, such as sentiment analysis and tweet recommendation tasks, where a hashtag is a tag, and a tweet is a data sample.

Another type of interesting target ``database'' is deep generative models (Figure \ref{fig: future}). Recently, high-quality text-to-image models have been developed and released \cite{rombach2022high,dhariwal2021diffusion}. They generate photo-realistic images and high-quality illustrations based on input text. Tiara and Seafaring can search images from any type of system only if they receive text and provide images, and Seafaring can be combined with such generative models. If we host these models locally, we save communication costs to collect unlabelled images and may be able to boost the performance of Seafaring further.

\section{Related Work}

\subsection{Active Learning}

The aim of active learning is to reduce the burden of labeling data \cite{settles2009active}. Active learning methods try to find out informative unlabelled data so that the total number of required labels is minimized while the performance of the trained model is maximized. There are many problem settings for active learning. In the stream-based tasks \cite{zliobaite2014active,attenberg2011online}, unlabelled data are shown one by one, and the active learning method decides whether we should label each data online. The membership query synthesis problem \cite{angluin1987queries,zhu2017generative} asks active learning methods to generate data from scratch to be labeled. Pool-based active learning is the most popular setting in the literature \cite{tong2001support,gal2017deep,wang2017cost,kirsch2019batchbald,beluch2018power}, and we have focused on it in this chapter. In pool-based active learning, a set of unlabelled data is shown, and active learning methods select some of them at once or iteratively. Although Seafaring is classified as a pool-based method, it is interesting to regard Seafaring as the middle of query synthesis and pool-based methods. The pool of Seafaring is extremely large and dense in the input space $\mathcal{X}$, and many new images are added to the pool every day. So the pool contains many images that we do not expect in advance, and searching from the pool shares preferable properties with query synthesis in terms of novelty and diversity. The important difference between Seafaring and query synthesis methods is that all of the images in the pool are indeed real-world data, while query synthesis may generate non-realistic data or corrupted data. If we use generative models as the target database, as we discuss in Section \ref{sec: discussion}, this difference becomes more nuanced.

The criterion on which active learning selects data is important, and there have been many studies on this topic. Typically, the criterion is represented by acquisition function $a\colon \mathcal{X} \to \mathbb{R}$, and active learning methods select the item with the highest acquisition value. The most popular acquisition functions are based on the uncertainty of the model \cite{gal2017deep,wang2017cost} or the amount of disagreement between ensemble models \cite{beluch2018power}. Some methods adaptively design the policy of active learning \cite{fang2017learning,haussmann2019deep}. In either case, the core idea is that uncertain data should be prioritized.

It is important to accelerate an iteration of active learning and reduce the number of iterations to make active learning efficient. A popular approach is batch active learning \cite{hoi2006batch,sener2018active,ash2020deep,shui2020deep,citovsky2021batch}, which allows learners to select multiple instances per iteration and reduces the number of iterations. Batch active learning has been successfully applied to large-scale problems and scales up to a large pool of ten million images and one million labels \cite{citovsky2021batch}. However, they require many labels, and labeling millions of data is not economical or infeasible in many tasks. Besides, this approach does not scale to billion-size pools so far.

The most significant difference between our proposed method and existing active learning lies in the size of the pool of unlabelled data. In typical settings, the pool contains thousands to millions of data, while we consider the pool with billions of data. If the pool does not contain informative data at all, active learning fails. Our huge pool avoids such cases.

\begin{figure}[tb]
\centering
\includegraphics[width=\hsize]{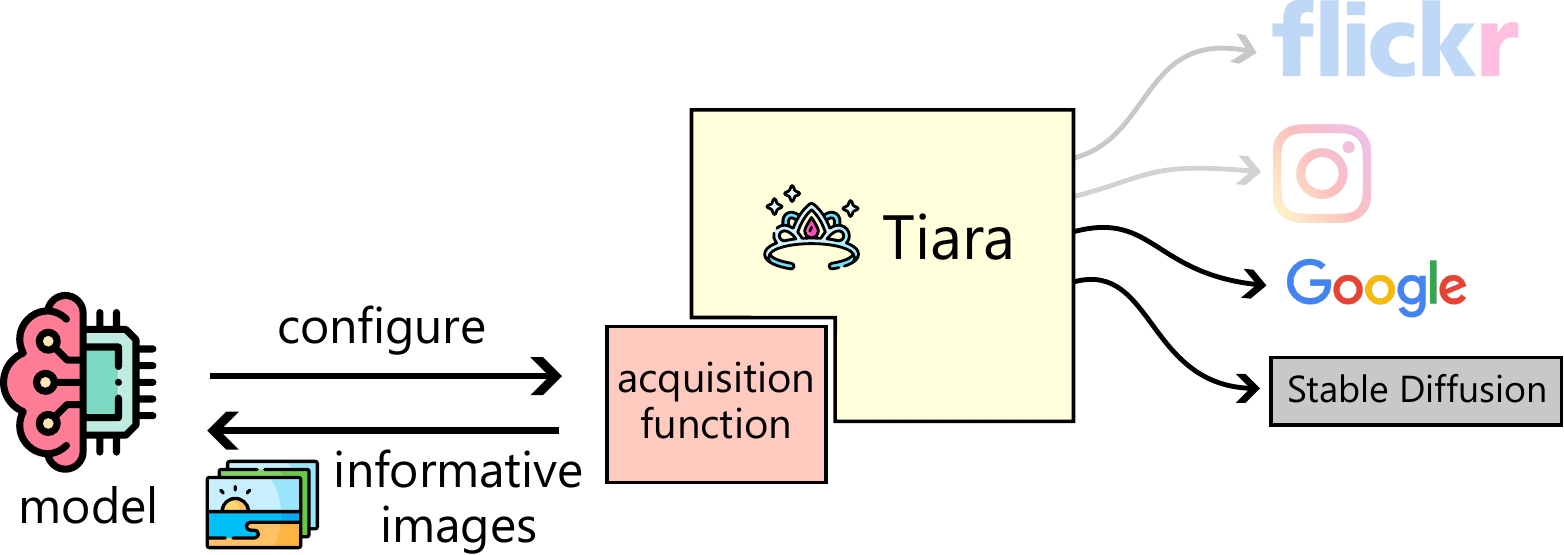}
\caption{\textbf{Future work.} Combining Seafaring with general-purpose search engines such as Google Image Search is interesting because it makes the pool even larger, say trillions of data. However, the API is costly, and they are not economical for many tasks currently. Seafaring can also be combined with generative models such as Stable Diffusion. Integrating Seafaring with such systems is an interesting future direction.}
\label{fig: future}
\end{figure}

\subsection{User-side algorithms}

User-side algorithms, such as user-side recommender systems \cite{sato2022private,sato2022principled} and search systems \cite{bra1994information,sato2022retrieving,sato2022clear}, are designed for users of Web services and software. This is in contrast to typical algorithms, which are designed for official developers of Web services and software. User-side information retrieval algorithms enable us to retrieve data from external databases with respect to user-designed score functions without privileged access to the databases.

Our proposed method uses an off-the-shelf user-side algorithm as a building block. The difference between our proposed method and existing works is its usage. The original motivation for the study of user-side algorithms is to create customized search engines for every user when the official system is not satisfactory. By contrast, we use them for gathering informative data for active learning.

\subsection{Web Mining}

Another relevant realm is Web mining. As Seafaring retrieves data from the Web, it can be seen as a Web mining method. The most relevant technique is focused crawling \cite{chakrabarti1999focused, mccallum2000automating, johnson2003evolving, baezayates2005crawling, guan2008guide,pham2019bootstrapping}, which aims at finding Web pages that meet some criteria by crawling the Web. There are various tasks in focused crawling depending on the criteria. For example, the target can be pages of specific topics \cite{chakrabarti1999focused, mccallum2000automating}, popular pages \cite{baezayates2005crawling}, structured data \cite{meusel2014focused}, and hidden pages \cite{barbosa2007adaptive}. To the best of our knowledge, there are no focused crawling methods to retrieve data for active learning.

We adopted an approach using a user-side algorithm instead of focused crawling. The advantage of user-side algorithms is that it works online and retrieves data in a few seconds to a few minutes, whereas focused crawling is typically realized by a resident program and takes a few hours to a few weeks, in which case the crawling procedure is a bottleneck of the time consumption. As active learning involves hundreds of iterations of data acquisition, it would be too expensive if each iteration took hours to weeks. With that being said, when the time constraint is not severe, the focused crawling approach is an appealing candidate because it does not rely on off-the-shelf APIs but can be applied to any Web pages. 

Another relevant realm is dataset mining from the Web \cite{zhang2021dsdd,castelo2021auctus}. These methods retrieve sets of compiled data, such as zip files and JSON files, while Seafaring retrieves unlabeled data from databases that are not designed for training models. The dataset mining methods and engines are easier to use and attractive if the target task is clear, while Seafaring can be used even when the task is not solid.

\section{Conclusion}

In this chapter, we have proposed Seafaring, an active learning method that uses a huge pool of unlabelled data. Compared with existing active learning research, Seafaring uses a several orders of magnitude larger pool. As the pool is extremely large, it is more likely that relevant images exist in the pool. Besides, the pool is general and contains diverse images, we use the same pool for many tasks, and we do not need to design the unlabelled pool for each task. So Seafaring reduces the cost of active learning further. The pool is, in addition, always up-to-date without any cost of maintenance because new images are added to the pool by the users of Flickr. In the experiments, we show that Seafaring is more effective than the counterpart of Seafaring with a smaller database and show the importance of the size of the pool. We also show that Seafaring automatically balances the ratio of positive and negative data without any explicit measures, while random sampling tends to have negative samples.

\chapter{Making Translators Privacy-aware on the User's Side}

\section{Introduction}

Machine translation systems are now essential in sectors including business and government for translating materials such as e-mails and documents~\citep{way2013emerging,brynjolfsson2019does,vieira2023machine}. Their rise in popularity can be attributed to recent advancements in language models~\citep{bahdanau2015neural,vaswani2017attention,openai2023gpt4} that have significantly improved translation accuracy, enhancing their overall utility. There is a growing demand to use these tools for private and sensitive information. For instance, office workers often need to translate e-mails from clients in other countries, but they want to keep these e-mails secret. Many are worried about using machine translation because there's a chance the information might get leaked. This means that, even with these helpful tools around, people often end up translating documents by themselves to keep the information safe.

Although many machine translation platforms claim they value privacy, the details and depth of this protection are not always clear. First, it's often uncertain how and to what level the data is kept safe. The details of the system are often an industrial secret of the service provider, and the source code is rarely disclosed. Even if providers are confident in their security, sophisticated attackers might still access private information. Also, there could be risks outside of these safeguards, including during data transfer, leading to potential leaks. Because of these concerns, users are cautious about using translation tools for sensitive data, missing out on their benefits.

In response to the prevalent concerns regarding data security in machine translation, we present PRISM (\underline{PRI}vacy \underline{S}elf \underline{M}anagement), which empowers users to actively manage and ensure the protection of their data. Instead of placing complete trust in the inherent security protocols of translation platforms, PRISM provides users with mechanisms for personal data safeguarding. This proactive strategy allows users to confidently use even translation engines that may not offer privacy measures. For platforms already equipped with privacy safeguards, PRISM acts as an additional protection layer, reinforcing their security mechanisms. PRISM adds these privacy features without much degradation of translation accuracy.

We propose two variants of PRISM. PRISM-R is a simple method with a theoretical guarantee of differential privacy. PRISM* (PRISM-Star) is a more sophisticated method that can achieve better translation accuracy than PRISM-R at the price of losing the theoretical guarantee. In practice, we recommend using PRISM* for most use cases and PRISM-R for cases where the theoretical guarantee is required.

In the experiments, we use real-world translators, namely T5~\cite{raffel2020exploring} and ChatGPT (GPT-3.5-turbo)~\cite{openai2023chatgpt,jiao2023chatgpt}, and the English $\to$ French and English $\to$ German translation. We confirm that PRISM can effectively balance privacy protection with translation accuracy.

The contributions of this chapter are as follows:

\begin{itemize}
  \item We formulate the problem of user-side realization of data privacy for machine translation systems.
  \item We propose PRISM, which enables users to preserve the privacy of data on their own initiative.
  \item We formally show that PRISM can preserve the privacy of data in terms of differential privacy.
  \item We propose an evaluation protocol for user-side privacy protection for machine translation systems.
  \item We confirm that PRISM can effectively balance privacy protection with translation accuracy using the real-world ChatGPT translator.
\end{itemize}

\section{Problem Formulation}

We assume that we have access to a black-box machine translation system $T$ that takes a source text $x$ and outputs a target text $y$. In practice, $T$ can be ChatGPT \cite{openai2023chatgpt}, DeepL \cite{deepl}, or Google Translate \cite{googletranslate}. We assume that the quality of the translation $T(x)$ is satisfactory, but $T$ may leak information or be unreliable in terms of privacy protection. Therefore, it is crucial to avoid feeding sensitive text $x$ directly into $T$. We have a sensitive source text $x_{\text{pri}}$, and our goal is to safely translate $x_{\text{pri}}$. We also assume that we have a dataset of non-sensitive source texts $\mathcal{D} = \{x_1, \dots, x_n\}$. $\mathcal{D}$ is unlabeled and need not be relevant to $x_{\text{pri}}$. Therefore, it is cheap to collect $\mathcal{D}$. In practice, $\mathcal{D}$ can be public news texts, and $x_{\text{pri}}$ can be an e-mail.

When considering user-side realization, the method should be simple enough to be executed on the user's side. For example, it is difficult for users to run a large language model or to train a machine learning model on their own because it requires a lot of computing resources and advanced programming skills. Therefore, we stick to simple and accessible methods.

In summary, our goal is to safely translate $x_{\text{pri}}$ using $T$ and $\mathcal{D}$, and the desiderata of the method are summarized as follows:
\begin{description}
  \item[Accurate] The final output should be a good translation of the input text $x_{\text{pri}}$.
  \item[Secure] The information passed to $T$ should not contain much information of the input text $x_{\text{pri}}$.
  \item[Simple] The method should be lightweight enough for end-users to use.
\end{description}

\section{Proposed Method (PRISM)}

\subsection{Overview}

\begin{figure}[tb]
  \centering
  \includegraphics[width=\hsize]{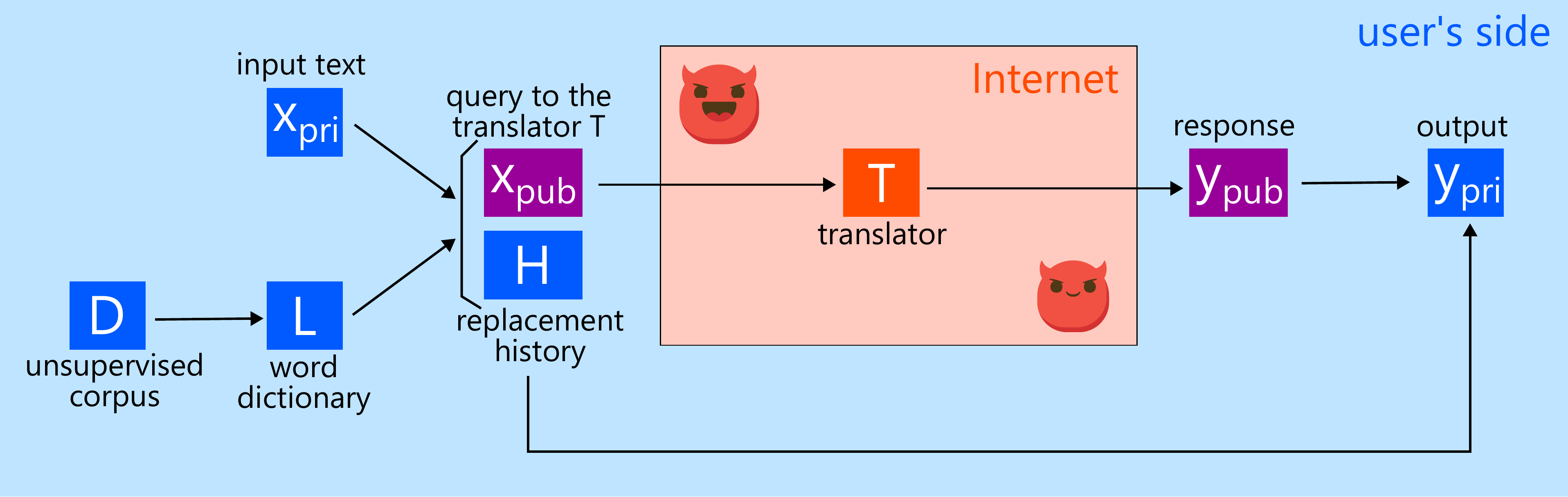}
  \caption{Overview of PRISM. The blue boxes indicate information kept on the user's side, the purple boxes indicate information exposed to the Internet, and the red region indicates the Internet. The purple boxes should not contain much information about the input text $x_{\text{pri}}$.}
  \label{fig: overview}
\end{figure}

PRISM has four steps as shown in Figure \ref{fig: overview}. (i) PRISM creates a word translation dictionary using $T$ and $\mathcal{D}$. This step should be done only once, and the dictionary can be used for other texts and users. (ii) PRISM converts the source text $x_{\text{pri}}$ to a non-sensitive text $x_{\text{pub}}$. (iii) PRISM translates $x_{\text{pub}}$ to $y_{\text{pub}}$ using $T$. (iv) PRISM converts $y_{\text{pub}}$ to $y_{\text{pri}}$ using the replacement history $\mathcal{H}$. We explain each step in detail in the following.

Let us first illustrate the behavior of PRISM with an example. let $x_{\text{pri}}$ be ``Alice is heading to the hideout.'' and $T$ be a machine translation system from English to French. PRISM converts $x_{\text{pri}}$ to $x_{\text{pub}}$ = ``Bob is heading to the store,'' which is not sensitive and can be translated with $T$. PRISM temporarily stores the substitutions (Alice $\to$ Bob) and (base $\to$ restaurant). Note that this substitution information is kept on the user's side and is not passed to $T$. Then, PRISM translates $x_{\text{pub}}$ to $y_{\text{pub}}$ = ``Bob se dirige vers la boutique.'' using the translator $T$. Finally, PRISM converts $y_{\text{pub}}$ to $y_{\text{pri}}$ = ``Alice se dirige vers la cachette.'' using the word translation dictionary, Alice (En) $\to$ Alice (Fr), Bob (En) $\to$ Bob (Fr), store (En) $\to$ boutique (Fr), and hideout (En) $\to$ cachette (Fr). The final output $y_{\text{pri}}$ is the translation of $x_{\text{pri}}$, and PRISM did not pass the information that Alice is heading to the hideout to $T$.

\subsection{Word Translation Dictionary} \label{sec: word_translation_dictionary}

We assume that a user does not have a word translation dictionary for the target language. We propose to create a word translation dictionary using the unsupervised text dataset $\mathcal{D}$. The desideratum is that the dictionary should be robust. Some words have multiple meanings, and we want to avoid incorrect substitutions in PRISM. Let $\mathcal{V}$ be the vocabulary of the source language. Let $S$ be a random variable that takes a random sentence from $\mathcal{D}$, and let $S_w$ be the result of replacing a random word in $S$ with $w \in \mathcal{V}$. We translate $S$ to the target language and obtain $R$ and translate $S_w$ to obtain $R_w$. Let \begin{align}
  p_{w, v} \stackrel{\text{def}}{=} \frac{\text{Pr}[v \in R_w]}{\text{Pr}[v \in R]}
\end{align} be the ratio of the probability of $v$ appearing in $R_w$ to the probability of $v$ appearing in $R$. The higher $p_{w, v}$ is, the more likely $v$ is the correct translation of $w$ since $v$ appears in the translation if and only if $w$ appears in the source sentence. Note that if we used only the numerator, article words such as ``la'' and ``le'' would have high scores, and therefore we use the ratio instead. Let $L(w)$ be the list of words in the decreasing order of $p_{w, v}$. $L(w, 1)$ is the most likely translation of $w$, and $L(w, 2)$ is the second most likely translation of $w$, and so on.

It should be noted that the translation engine used here is not necessarily the same as the one $T$ we use in the test phase. As we need to translate many texts here, we can use a cheaper translation engine. We also note that once we create the word translation dictionary, we can use it for other texts and users. We will distribute the word translation dictionaries for English $\to$ French and English $\to$ German, and users can skip this step if they use these dictionaries.

\subsection{PRISM-R} \label{sec: prismr}

PRISM-R is a simple method to protect data privacy on the user's side. PRISM-R randomly selects words $w_1, \ldots, w_k$ in the source text $x_{\text{pri}}$ and randomly selects substitution words $u_1, \ldots, u_k$ from the word translation dictionary. $x_{\text{pub}}$ is the result of replacing $w_1$ with $u_1$, \ldots, and $w_k$ with $u_k$. PRISM-R then translates $x_{\text{pub}}$ to $y_{\text{pub}}$ using $T$. Finally, PRISM-R converts $y_{\text{pub}}$ to $y_{\text{pri}}$ as follows. Possible translation words of $u_i$ are $L(u_i)$. PRISM-R first searches for $L(u_i, 1)$, the most likely translation of $u_i$, in $y_{\text{pub}}$. There should be $L(u_i, 1)$ in $y_{\text{pub}}$ if $T$ translated $u_i$ to $L(u_i, 1)$. If $L(u_i, 1)$ is found, PRISM-R replaces $L(u_i, 1)$ with $L(w_i, 1)$. However, if $u_i$ has many translation candidates, $T$ may not have translated $u_i$ to $L(u_i, 1)$. If $L(u_i, 1)$ is not found, it proceeds to $L(u_i, 2)$, the second most likely translation of $u_i$, and replaces $L(u_i, 2)$ with $L(w_i, 1)$, and so on.

The pseudo code is shown in Algorithm \ref{alg: prismr}.  

\begin{algorithm2e}[tb]
  \DontPrintSemicolon
  \caption{PRISM-R} \label{alg: prismr}
  \KwData{Source text $x_{\text{pri}}$; Word translation dictionary $L$; Ratio $r \in (0, 1)$.}
  \KwResult{Translated text $y_{\text{pri}}$.}
  \BlankLine
  $t_1, \ldots, t_n \gets \text{Tokenize}(x_{\text{pri}})$ \tcp*{Tokenize $x_{\text{pri}}$}
  $\mathcal{H} \gets \emptyset$ \tcp*{The history of substitutions}
  \For{$i \gets 1$ \textup{\textbf{to}} $n$}{
      $p \sim \text{Unif}(0, 1)$ \;
      \If{$p < r$}{
          $u_i \gets$ a random source word in $L$ \tcp*{Choose a substitution word}
          $\mathcal{H} \gets \mathcal{H} \cup \{(t_i, u_i)\}$ \tcp*{Update the history}
          $t_i \gets u_i$ 
      }
  }
  $x_{\text{pub}} \gets \text{Detokenize}(t_1, \ldots, t_n)$ \tcp*{Detokenize $t_1, \ldots, t_n$}
  $y_{\text{pub}} \gets T(x_{\text{pub}})$ \tcp*{Translate $x_{\text{pub}}$}
  $y_{\text{pri}} \gets y_{\text{pub}}$ \tcp*{Copy $y_{\text{pub}}$}
  \For{$(w, u) \in \mathcal{H}$}{
      \For{$v \in L(u)$}{
          \If{$v \in y_{\text{pri}}$}{
              $y_{\text{pri}} \gets$ replace $v$ with $L(w, 1)$ in $y_{\text{pri}}$
              \textbf{break}
          }
      }
  }
  \Return{$y_{\text{pri}}$}
\end{algorithm2e}

\subsection{Differential Privacy of PRISM-R}

Differential privacy~\citep{dwork2006differential} provides a formal guarantee of data privacy. We show that PRISM-R satisfies differential privacy. This result not only provides a privacy guarantee but also shows PRISM-R can be combined with other mechanisms due to the inherent composability and post-processing resilience of differential privacy~\citep{mcsherry2009privacy}.

We first define differential privacy. We say texts $x = w_1, \ldots, w_n$ and $x' = w'_1, \ldots, w'_n$ are neighbors if $w_i = w'_i$ for all $i \in \{1, \ldots, n\}$ except for one $i \in \{1, \ldots, n\}$. Let $x \sim x'$ denote that $x$ and $x'$ are neighbors. Let $A$ be a randomized mechanism that takes a text $x$ and outputs a text $y$. Differential privacy is defined as follows.
\begin{definition}[Differential Privacy]
  $A$ satisfies $\epsilon$-differential privacy if for all $x \sim x'$ and $S \subseteq \text{Im}(A)$, \begin{align}
    \frac{\text{Pr}[A(x) \in S]}{\text{Pr}[A(x') \in S]} \leq e^\epsilon.
  \end{align} 
\end{definition}
We show that the encoder of PRISM-R $A_{\text{PRISM-R}}\colon x_{\text{pri}} \mapsto x_{\text{pub}}$ is differential private, and therefore, $x_{\text{pri}}$ cannot be inferred from $x_{\text{pub}}$, which is the only information that $T$ can access. 

\begin{theorem} \label{thm: prismr}
  $A_{\textup{PRISM-R}}$ is $\ln \left(\frac{r + |\mathcal{V}| (1 - r)}{r}\right)$-differential private.
\end{theorem}

We emphasize that the additive constant $\delta$ is zero, i.e., PRISM-R is $(\epsilon, 0)$-differential private, which provides a strong guarantee of data privacy. 

\begin{proof}
  Let $x = w_1, \ldots, w_n$ and $x' = w'_1, \ldots, w'_n$ be any two neighboring texts. Without loss of generality, we assume that $w_1 = w'_1$, \ldots, $w_{n - 1} = w'_{n - 1}$ and $w_n \neq w'_n$. Let $s = s_1, \ldots, s_n$ be any text, and let \begin{align}
    c \stackrel{\text{def}}{=} \sum_{i = 1}^n \mathbbm{1}[s_i \neq w_i]
  \end{align} be the number of different words in $x$ and $s$. The probability that PRISM-R convert $x$ to $s$ is \begin{align}
    \text{Pr}[A_{\text{PRISM-R}}(x) = s] = \sum_{i = c}^n \binom{n - c}{i - c} r^i (1 - r)^{n - i} \left(\frac{1}{|\mathcal{V}|}\right)^i, 
  \end{align} where $i$ is the number of replaced words because all of the different words must be replaced, and the number of ways to choose the remaining words is $\binom{n - c}{i - c}$. This probability can be simplified as follows: \begin{align}
    \text{Pr}[A_{\text{PRISM-R}}(x) = s] &= \sum_{i = c}^n \binom{n - c}{i - c} r^i (1 - r)^{n - i} \left(\frac{1}{|\mathcal{V}|}\right)^i \\
    &= \sum_{i = 0}^{n - c} \binom{n - c}{i} r^{i + c} (1 - r)^{n - c - i} \left(\frac{1}{|\mathcal{V}|}\right)^{i + c} \\
    &= r^c (1 - r)^{n - c} \left(\frac{1}{|\mathcal{V}|}\right)^c \sum_{i = 0}^{n - c} \binom{n - c}{i} \left(\frac{r}{|\mathcal{V}| (1 - r)}\right)^i \\
    &= r^c (1 - r)^{n - c} \left(\frac{1}{|\mathcal{V}|}\right)^c \left(1 + \frac{r}{|\mathcal{V}| (1 - r)}\right)^{n - c}, \label{eq: prismr-x}
  \end{align} where we used the binomial theorem in the last equality. Similarly, let \begin{align}
    c' &\stackrel{\text{def}}{=} \sum_{i = 1}^n \mathbbm{1}[s_i \neq w'_i] \\
    &= \begin{cases}
      c & \text{if } \mathbbm{1}[s_n \neq w_n] = \mathbbm{1}[s_n \neq w'_n] = 1 \\
      c + 1 & \text{if } \mathbbm{1}[s_n \neq w_n] = 0 \text{ and } \mathbbm{1}[s_n \neq w'_n] = 1 \\
      c - 1 & \text{if } \mathbbm{1}[s_n \neq w_n] = 1 \text{ and } \mathbbm{1}[s_n \neq w'_n] = 0 \\
    \end{cases}.
  \end{align} be the number of different words in $x'$ and $s$. Then, \begin{align}
    \text{Pr}[A_{\text{PRISM-R}}(x') = s] = r^{c'} (1 - r)^{n - c'} \left(\frac{1}{|\mathcal{V}|}\right)^{c'} \left(1 + \frac{r}{|\mathcal{V}| (1 - r)}\right)^{n - c'}. \label{eq: prismr-xp}
  \end{align} Combining Eqs. (\ref{eq: prismr-x}) and (\ref{eq: prismr-xp}), \begin{align}
    \frac{\text{Pr}[A_{\text{PRISM-R}}(x) = s]}{\text{Pr}[A_{\text{PRISM-R}}(x') = s]} &= \begin{cases}
      1 & \text{if } \mathbbm{1}[s_n \neq w_n] = \mathbbm{1}[s_n \neq w'_n] = 1 \\
      \frac{r + |\mathcal{V}| (1 - r)}{r} & \text{if } \mathbbm{1}[s_n \neq w_n] = 0 \text{ and } \mathbbm{1}[s_n \neq w'_n] = 1 \\
      \frac{r}{r + |\mathcal{V}| (1 - r)} & \text{if } \mathbbm{1}[s_n \neq w_n] = 1 \text{ and } \mathbbm{1}[s_n \neq w'_n] = 0 \\
    \end{cases} \\
    &\le \frac{r + |\mathcal{V}| (1 - r)}{r}.
  \end{align}
\end{proof}

An interesting part of PRISM is that PRISM is resilient against the purturbation due to the final substitution step. Many of differential private algorithms add purturbation to the data \cite{feyisetan2020privacy,bo2021er, weggenmann2022dpvae,abadi2016deep, andrew2021differentially} and therefore, their final output becomes unreliable when the privacy constraint is severe. By contrast, PRISM enjoys both of the privacy guarantee and the reliability of the final output thanks to the purturbation step and the recovery step. The information $x_{\text{pub}}$ passed to $T$ has little information due ot the purturbation step. This, however, makes the intermediate result $y_{\text{pub}}$ an unreliable translation of $x_{\text{pri}}$. PRISM recovers a good translation $y_{\text{pri}}$ by the final substitution step.

\subsection{PRISM*} \label{sec: prismstar}

PRISM* is a more sophisticated method and achieves better accuracy than PRISM-R. PRISM* chooses words $w_1, \ldots, w_k$ in the source text $x_{\text{pri}}$ and substitution words $u_1, \ldots, u_k$ from the word translation dictionary more carefully while PRISM-R chooses them randomly to achieve differential privacy. PRISM* has two mechanisms to choose words. The first mechanism is to choose words so that the part-of-speech tags match. The second mechanism is to choose words that can be translated accurately by the word dictionary. We explain each mechanism in detail in the following.

\begin{algorithm2e}[tb]
  \DontPrintSemicolon
  \caption{PRISM*} \label{alg: prismstar}
  \KwData{Source text $x_{\text{pri}}$; Word translation dictionary $L$; Confidence Scores $c$, Ratio $r \in (0, 1)$.}
  \KwResult{Translated text $y_{\text{pri}}$.}
  \BlankLine
  $t_1, \ldots, t_n \gets \text{Tokenize}(x_{\text{pri}})$ \tcp*{Tokenize $x_{\text{pri}}$}
  $s_1, \ldots, s_n \gets \text{Part-of-Speech}(t_1, \ldots, t_n)$ \;
  $k \gets 0$ \tcp*{The number of substitutions}
  $\mathcal{H} \gets \emptyset$ \tcp*{The history of substitutions}
  \For{$i \in \{1, \ldots, n\}$ \textup{in the decreasing order of }$c(t_i, s_i)$}{
    $u_i \gets$ the unused source word $u_i$ with the highest confidence score $c(u_i, s_i)$ in $L$ \tcp*{Choose a substitution word}
    $\mathcal{H} \gets \mathcal{H} \cup \{(t_i, u_i, s_i)\}$ \tcp*{Update the history}
    $t_i \gets u_i$ \;
    $k \gets k + 1$ \;
    \If{$k \geq r n$}{
      \textbf{break}
    }
  }
  $x_{\text{pub}} \gets \text{Detokenize}(t_1, \ldots, t_n)$ \tcp*{Detokenize $t_1, \ldots, t_n$}
  $y_{\text{pub}} \gets T(x_{\text{pub}})$ \tcp*{Translate $x_{\text{pub}}$}
  $y_{\text{pri}} \gets y_{\text{pub}}$ \tcp*{Copy $y_{\text{pub}}$}
  \For{$(w, u, s) \in \mathcal{H}$}{
    \For{$v \in L(u, s)$}{
      \If{$v \in y_{\text{pri}}$}{
        $y_{\text{pri}} \gets$ replace $v$ with $L(w, s, 1)$ in $y_{\text{pri}}$
        \textbf{break}
      }
    }
  }
  \Return{$y_{\text{pri}}$}
\end{algorithm2e}

PRISM* creates a word translation dictionary with a part of speech tag. The procedure is the same as Section \ref{sec: word_translation_dictionary} except that we use the part-of-speech tag of the source word as the key of the dictionary. Let $(w, s)$ be a pair of a source word $w$ and its part-of-speech tag $s$. PRISM* replaces a random word with part-of-speech tag $s$ with $w$ to create $S_{w, s}$, obtains $R_{w, s}$ by translating $S_{w, s}$, and defines \begin{align}
  p_{w, s, v} \stackrel{\text{def}}{=} \frac{\text{Pr}[v \in R_{w, s}]}{\text{Pr}[v \in R]}.
\end{align} $L(w, s)$ is the list of words $v$ in the decreasing order of $p_{w, s, v}$. 

In the test time, PRISM* chooses substitute words so that the part-of-speech tags match, and use the word translation dictionary with part-of-speech tags to determine the translated word.

PRISM* also uses the confidence score \begin{align}
  c(w, s) \stackrel{\text{def}}{=} \max_v p_{w, s, v},
\end{align} which indicates the reliability of word translation $w \to L(w, 1)$, to choose words. Multiple-meaning words should not be substituted in PRISM because a word-to-word translation may fail. PRISM* chooses words to be substituted in the decreasing order of the confidence score, which results in selecting single-meaning and reliable words that can be translated accurately by the word dictionary. If a word $(w, s)$ has two possible translations $v_1$ and $v_2$ that are equally likely, $\text{Pr}[v \in R_{w, s}]$ is lower than $0.5$ for any $v$, even for $v = v_1$ and $v = v_2$, the confidence score $c(w, s)$ tends to be low, and PRISM* avoids selecting such $(w, s)$. The selected words $w$ can be reliably translated by $L(w, s, 1)$. PRISM* also chooses substitute words with high confidence scores so that PRISM* can robustly find the corresponding word $L(w, s, 1)$ in the translated text $y_{\text{pub}}$ in the final substitution step.

The pseudo code of PRISM* is shown in Algorithm \ref{alg: prismstar}.

Note that PRISM* does not enjoy the differential privacy guarantee of PRISM-R as (i) PRISM* replaces words with the same part-of-speech tag so that two texts with different part-of-speech templates have zero probability of transition, and (ii) PRISM* chooses words with high confidence scores so that the probability of transition is biased. Nevertheless, PRISM* empirically strikes a better trade-off between privacy and accuracy than PRISM-R as we will show in the experiments. Note that PRISM* can be combined with PRISM-R to guarantee differential privacy. For example, one can apply PRISM-R and PRISM* in a nested manner, which guarantees differential privacy due to the differential privacy of PRISM-R (Theorem \ref{thm: prismr}) and the post processing resilience of differential privacy~\citep{mcsherry2009privacy}. One can also apply PRISM-R with probability $(1 - \beta)$ and PRISM* with probability $\beta$, which also guarantees differential privacy because the minimum probability of transition is bounded from below due to the PRISM-R component.

\section{Experiments}

We confirm the effectiveness of our proposed methods through experiments.

\subsection{Evaluation Protocol}

\begin{figure}[tb]
  \centering
  \includegraphics[width=\hsize]{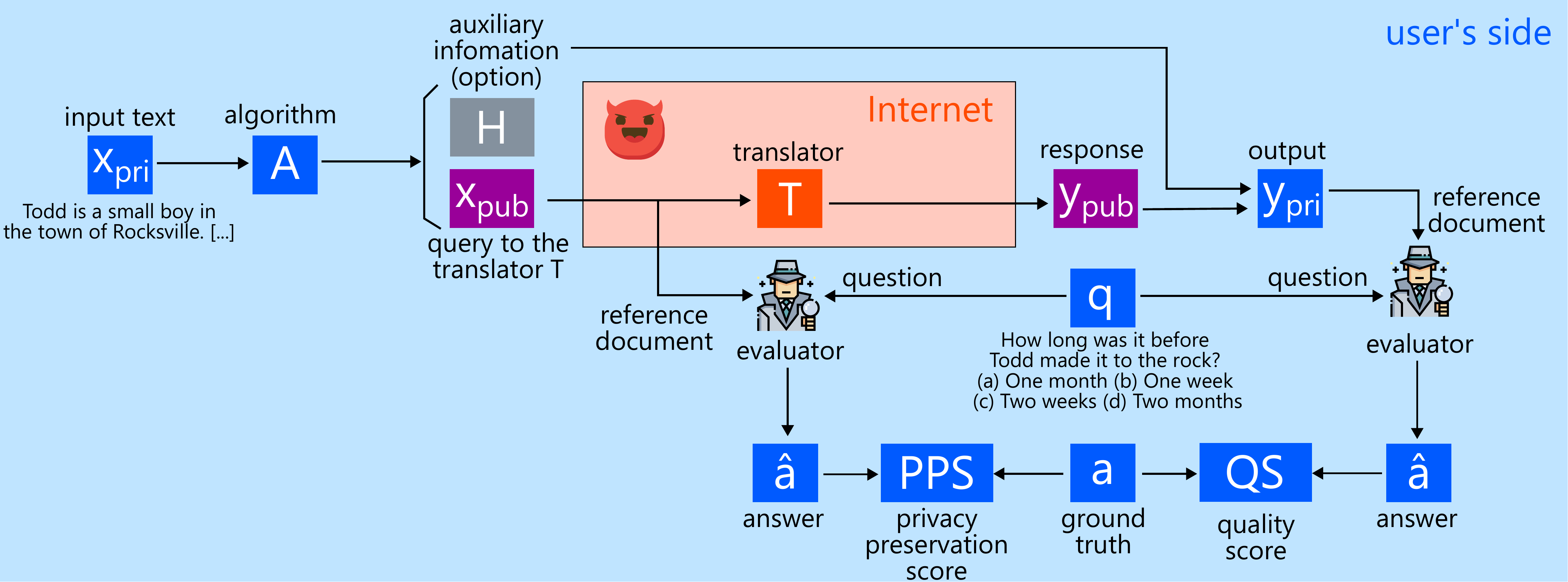}
  \caption{Overview of the evaluation protocol.}
  \label{fig: eval}
\end{figure}

As our problem setting is novel, we first propose an evaluation protocol for the user-side realization of privacy-aware machine translation systems. We evaluate the translation accuracy and privacy protection as follows.

Let $\mathcal{X} = \{x_1, x_2, \ldots, x_N\}$ be a set of test documents to be translated. Our aim is to read $\mathcal{X}$ in the target language without leaking information of $\mathcal{X}$.

For evaluation purposes, we introduce a question-answering (QA) dataset $\mathcal{Q} = \{(q_{ij}, a_{ij})\}$, where $q_{ij}$ and $a_{ij}$ are a multiple-choice question and answer regarding the document $x_i$, respectively. $\mathcal{Q}$ is shown only to the evaluator, and not to the translation algorithm.

\textbf{Privacy-preserving Score.} The idea of our privacy score is based on an adversarial evaluation where adversaries try to extract information from the query sent by the user. Let $x^{\text{pub}}_i$ be the query sent to the translator $T$. An evaluator is given $x^{\text{pub}}_i$ and $q_{ij}$, and asked to answer the question. The privacy-preserving score of the translation algorithm is defined as $\text{PPS} = (1 - \text{acc})$, where acc is the accuracy of the evaluator. The higher the privacy-preserving score is, the better the privacy protection is. Intuitively, if the accuracy is low, the evaluator cannot draw any information from $x^{\text{pub}}_i$ to answer the question. Conversely, if the accuracy is high, the evaluator can infer the answer solely from $x^{\text{pub}}_i$, which means that $x^{\text{pub}}_i$ leaks information. We note that the translation algorithm does not know the question $q_{ij}$, and therefore, the translation algorithm needs to protect all information to achieve a high privacy-preserving score so that any answer on the document cannot be drawn from $x^{\text{pub}}_i$. The rationale behind this score is that we cannot predict what form information leaks will take in advance. Even if $x_\text{pub}$ does not look like $x_\text{pri}$ at a glance, sophisticated adversaries might extract information that can be used to infer $x_\text{pri}$. Therefore, we employ an outside evaluator and adopt an adversarial evaluation.

\textbf{Quality Score.} Let $y^{\text{pri}}_i$ be the final output of the translation algorithm. We use the same QA dataset and ask an evaluator to answer the question $q_{ij}$ using $y^{\text{pri}}_i$. The quality score of the translation algorithm is defined as $\text{QS} = \text{acc}$, where acc is the accuracy of the evaluator. The higher the quality score is, the better the translation quality is. Intuitively, if the accuracy is high, the evaluator can answer the question correctly using $y^{\text{pri}}_i$, which means that $y^{\text{pri}}_i$ contains sufficient information on $x_i$. We note again that the translation algorithm does not know the question $q_{ij}$, and therefore, the translation algorithm needs to preserve all information to achieve a high quality score so that any answer on the document can be drawn from $y^{\text{pri}}_i$.

The protocol is illustrated in Figure \ref{fig: eval}.

\begin{algorithm2e}[tb]
  \DontPrintSemicolon
  \caption{AUPQC} \label{alg: aupqc}
  \BlankLine
  $\mathcal{P} \gets $ the set of the trade-off parameters in the increasing order of the privacy-preserving score. \;
  $s \gets 0$ \tcp*{The area under the curve}
  $(\text{PPS}_{\text{prev}}, \text{QS}_{\text{prev}}) \gets (\text{None}, \text{None})$ \;
  \For{$\alpha \gets \mathcal{P}$}{
    $(\text{PPS}, \text{QS}) \gets \text{Evaluate}(\alpha)$ \tcp*{Evaluate the privacy-preserving score and the quality score}
    \If{$\text{PPS}_{\text{prev}} \textup{ is None}$}{
      $s \gets s + \text{PPS} \times \text{QS}$ \tcp*{Add the area of the first rectangle}
    }
    \Else{
      $s \gets s + (\text{PPS} - \text{PPS}_{\text{prev}}) \times (\text{QS}_{\text{prev}} + \text{QS}) / 2$ \tcp*{Add the area of the trapezoid}
    }
    $(\text{PPS}_{\text{prev}}, \text{QS}_{\text{prev}}) \gets (\text{PPS}, \text{QS})$
  }
  \Return{$s$}
\end{algorithm2e}

We introduce the area-under-privacy-quality curve (AUPQC) to measure the effectiveness of methods. The privacy-preserving score and the quality score are in a trade-off relationship. Most methods, including PRISM-R and PRISM*, have a parameter to control the trade-off. An effective method should have a high privacy-preserving score and a high quality score at the same time. We use the AUPQC to measure the trade-off. Specifically, we scan the trade-off parameter and plot the privacy-preserving score and the quality score in the two-dimensional space. The AUPQC is the area under the curve. The larger the AUPQC is, the better the method is. The pseudo code is shown in Algorithm \ref{alg: aupqc}.

We also introduce QS@p, a metric indicating the quality score at a specific privacy-preserving score. The higher QS@p is, the better the method is.  In realistic scenarios, we may have a severe security budget $p$ which represents the threshold of information leakage we can tolerate. QS@p is particularly useful under such constraints as it provides a direct measure of the quality we can enjoy under the security budget. It is noteworthy that the privacy-preserving score can be evaluated before we send information to the translator $T$. Therefore, we can tune the trade-off parameter and ensure that we enjoy the privacy-preserving score = $p$ and the quality score = QS@p.

\subsection{Experimental Setups}

We use the MCTest dataset~\cite{richardson2013mctest} for the documents $x_i$, question $q_{ij}$, and answer $a_{ij}$. Each document in the MCTest dataset is a short story with four questions and answers. The reason behind this choice is that the documents of the MCTest dataset were original ones created by crowdworkers. This is in contrast to other reading comprehension datasets such as NarrativeQA~\cite{kocisky2018narrativeqa} and CBT~\cite{hill2016goldilocks} datasets, which are based on existing books and stories, where the evaluator can infer the answers without relying on the input document $x^{\text{pub}}_i$.

We use T5~\cite{raffel2020exploring} and GPT-3.5-turbo~\cite{openai2023chatgpt} as the translation algorithm $T$. We use the prompt ``Directly translate English to [Language]: [Source Text]'' to use GPT-3.5-turbo for translation.

We also use GPT-3.5-turbo as the evaluator. Specifically, the prompt is composed of four parts. The first part of the prompt is the instruction ``Read the following message and solve the following four questions.'' The second part is the document to be evaluated, which is the query document $x^{\text{pub}}_i$ for the privacy-preserving score and the final output $y^{\text{pri}}_i$ for the quality score. The third part is the four questions. The last part is the instruction ``Output only four characters representing the answers, e.g.,\textbackslash n1. A\textbackslash n2. B\textbackslash n3. A\textbackslash n4. D.'' We parse the output of GPT-3.5-turbo to extract the answers and evaluate the accuracy.

We use the following four methods.

\textbf{Privacy- and Utility-Preserving Textual Analysis (PUP)}~\cite{feyisetan2020privacy} is a differential private algorithm to convert a document to a non-sensitive document without changing the meaning of $x$. PUP has a trade-off parameter $\lambda$ for privacy and utility. We convert the source text $x_{\text{pri}}$ to $x_{\text{pub}}$ using PUP and translate $x_{\text{pub}}$ to obtain the final output $y_{\text{pri}}$.

\textbf{NoDecode} translates the encoded text $x_{\text{pub}}$ of PRISM* to obtain the final output $y_{\text{pri}}$. NoDecode does not decode the output of the translator $T$. This method has the same privacy-preserving property as PRISM* but the accuracy should be lower. The improvements from NoDecode are the contribution of our framework.

\textbf{PRISM-R} is our method proposed in Section \ref{sec: prismr}.

\textbf{PRISM*} is our method proposed in Section \ref{sec: prismstar}.

We change the ratio $r$ of NoDecode, PRISM-R, and PRISM* and the parameter $\lambda$ of PUP to control the trade-off between privacy-preserving score and the quality score.

\subsection{Results}

Figure \ref{fig: privacy_quality} shows the trade-off, where the x-axis is the privacy-preserving score and the y-axis is the quality score. PRISM* clearly strikes the best trade-off, and the results of PRISM-R are also better than those of NoDecode and PUP, especially when the privacy-preserving score is high. 

The maximum privacy-preserving score is around $0.5$ for all methods, even though there are four choices in each question. Intuitively, the accuracy of the evaluator should be $0.25$ when the reference document is random, so the maximum PPR should be $0.75$. We found that this is because some questions can be inferred solely from the question text. For example, there is a question ``How did the girl hurt her knee? (a) she was in the street (b) she had no friends (c) she fell down, and (d) the old lady's bike hit her.'' We can infer the answer is (c) or (d) as (a) and (b) do not make sense (the answer is (c)). To verify this hypothesis, we had GPT-3.5-turbo answer the questions using only the question text. The accuracy was $0.492$. Therefore, PPS $\approx 0.5$ indicates that the query has no more information than the empty text. This experiment also shows that the GPT-3.5-turbo evaluator is so powerful that it can infer the answer from the question text only, and it is an effective adversarial evaluator.

Table \ref{table: translation} shows the quantitative results. We report QS@0.5, i.e., the quality score when PPS is $0.5$, which roughly means the quality we can enjoy when no information is leaked based on the above analysis. PRISM* consistently achieves the best scores across all the metrics and settings, and PRISM-R achieves the second-best results in most of the metrics and settings. Notably, PRISM* achieves QS $\approx 0.8$ when no information is leaked. This result shows that PRISM* can accurately translate the texts while protecting the privacy of the texts.

Table \ref{tab: example1} shows sample translations of PRISM*. The leaked information $x_{\text{pub}}$ does not make sense and reveals little about the secret text $x_{\text{pri}}$. Although it contains some grammatical errors, the output $y_{\text{pri}}$ is generally a correct translation of the input text $x_{\text{pri}}$, which is useful for native speakers to grasp the content.

\begin{figure}
  \centering
  \includegraphics[width=\hsize]{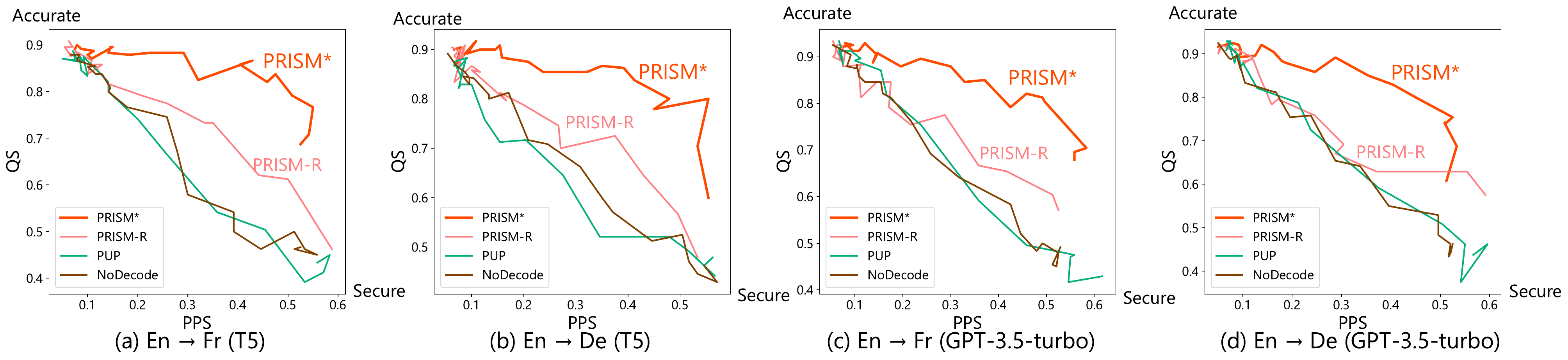}
  \caption{Trade-off between the privacy-preserving score and the quality score. The x-axis is the privacy-preserving score and the y-axis is the quality score.}
  \label{fig: privacy_quality}
\end{figure}

\begin{table}[t]
  \centering
  \caption{Quantitative Results. The best results are shown in \textbf{bold}, and the second best results are shown in \underline{underline}.}
  \label{table: translation}
  \begin{tabular}{lcccc}
  \toprule
   & \multicolumn{2}{c}{En $\rightarrow$ Fr (T5)} & \multicolumn{2}{c}{En $\rightarrow$ De (T5)}  \\
   & AUPQC $\uparrow$ & QS@0.5 $\uparrow$ & AUPQC $\uparrow$ & QS@0.5 $\uparrow$ \\
  \midrule
  NoDecode & 0.355 & 0.493 & 0.373 & 0.524 \\
  PUP & 0.363 & 0.439 & 0.363 & 0.505 \\
  PRISM-R & \underline{0.431} & \underline{0.613} & \underline{0.396} & \underline{0.557} \\
  PRISM* & \textbf{0.454} & \textbf{0.803} & \textbf{0.473} & \textbf{0.789} \\
  \bottomrule
  \end{tabular}
\newline
\vspace*{0.1 in}
\newline
  \begin{tabular}{lcccc}
  \toprule
   & \multicolumn{2}{c}{En $\rightarrow$ Fr (ChatGPT)} & \multicolumn{2}{c}{En $\rightarrow$ De (ChatGPT)}  \\
   & UPQC $\uparrow$ & QS@0.5 $\uparrow$ & AUPQC $\uparrow$ & QS@0.5 $\uparrow$ \\
  \midrule
  NoDecode & 0.376 & 0.495 & 0.370 & 0.480 \\
  PUP & \underline{0.415} & 0.487 & 0.391 & 0.511 \\
  PRISM-R & 0.399 & \underline{0.611} & \underline{0.432} & \underline{0.629} \\
  PRISM* & \textbf{0.482} & \textbf{0.799} & \textbf{0.445} & \textbf{0.769} \\
  \bottomrule
  \end{tabular}
\end{table}

\begin{table}
  \caption{Examples of PRISM* translations. The query $x_{\text{pub}}$ to the translator does not make sense and provides little information on the secret text $x_{\text{pri}}$. Although it contains some grammatical errors, the output $y_{\text{pri}}$ of PRISM* is generally a correct translation of the input text $x_{\text{pri}}$.}\label{tab: example1}
  \centering
  \begin{tabular}{p{4.3cm}|p{4.3cm}|p{4.3cm}}
    \toprule
    $x_{\text{pri}}$ &  $x_{\text{pub}}$ &  $y_{\text{pri}}$ \\
    \midrule
    Jimmy didn't eat breakfast. Because he didn't eat breakfast he was very hungry and could not stop thinking about food. He was thinking about all of the things that he liked to eat the most. He was thinking about breakfast foods like bacon and eggs but he was also thinking about lunch and dinner foods like pizza and macaroni and cheese.  &  Logan didn't eat hamster. Because he didn't eat circus he was very fishing and could not stop thinking about smile. He was thinking about all of the things that he screamed to eat the most. He was thinking about zoo foods like bacon and grandparents but he was also thinking about guitar and hamburger foods like pizza and lettuce and jungle.  &  Jimmy n'a pas mangé de déjeuner. Parce qu'il n'a pas mangé de déjeuner, il était très faim et ne pouvait pas arrêter de penser au nourriture. Il pensait à toutes les choses qu'il aimé de vouloir manger le plus. Il pensait aux aliments du déjeuner comme le bacon et les ufs, mais il pensait aussi à des aliments de déjeuner et de dîner comme la pizza et la macaroni et la fromage.  \\
    \midrule
    A boy was trying to pick out what instrument that he wanted to play. His parents wanted him to pick a good one because playing an instrument was very important to them. So, the boy went to a music store with his parents.  & A dragon was trying to pick out what zoo that he wanted to play. His grandchildren wanted herself to pick a good one because playing an Hey was very important to them. Shelly, the bacon went to a mud store with his ants.  &  Un garçon essayait de choisir quelle instrument il travaillé jouer. Parents parents voulaient eux-mêmes en choisir un bon car jouer avec un Instrument était très important pour eux. So, le garçon est allé dans un magasin de musique avec parents fourmis. \\
    \bottomrule
  \end{tabular}
\end{table}

\section{Related Work}

\textbf{Privacy Protection of Texts.} There is a growing demand for privacy protection measures for text data and many methods have been proposed. The U.S. Health Insurance Portability and Accountability Act (HIPAA), which requires that the personal information of patients should be protected, is one of the triggers of heightening concerns on privacy protection of data~\cite{baumer2000privacy,ness2007influence,lane2010balancing}. One of the challenges to following HIPAA is to protect information hidden in medical records written in free texts~\cite{meystre2010automatic}. The rule-based method proposed by Neamatullah et al.~\cite{neamatullah2008automated} is one of the early attempts to delete sensitive information from free texts. Li et al.~\cite{li2017anonymizing} claimed that hiding only the sensitive information is not enough to protect privacy because side information may also leak information and proposed a robust method. Many other methods~\cite{lison2021anonymisation,feyistein2019leveraging,mosallanezhad2019deep} aim at anonymizing texts so that the authors or the attributions of the authors~\cite{shetty2018a4nt} cannot be inferred. Some methods ensure the rigorous privacy guarantee of differential privacy~\cite{bo2021er, weggenmann2022dpvae}. The most relevant work to ours is the work by Feyisetan et al.~\cite{feyisetan2020privacy}, which aims at protecting the privacy of texts while preserving the utility of the texts. Their proposed method is simple enough to implement on the user's side. However, their definition of privacy is different from ours. They aim at protecting the privacy of the author of the text, while we aim at protecting the content. Their method leaks much information on the content of the text. We confirmed this in the experiments. Many of the other methods also aim at protecting the author of the text and keeping the content of the text intact even after the anonymization \cite{feyistein2019leveraging,bo2021er}. 

\textbf{Homomorphic encryption.} Homomorphic encryption~\cite{gentry2009fully, bachrach2016cryptonets, laur2006cryptographically} enables to compute on encrypted data without decrypting them. The service provider can carry out the computation without knowing the content of the data with this technology~\cite{blatt2020secure,acar2018survey}. However, users cannot enjoy the benefit of secure computing unless the service provider implements the technology. Homomorphic encryption is notoriously slow~\cite{naehrig2011can} and can degrade the performance, and therefore, the service provider may be reluctant to implement it. To the best of our knowledge, no commercial translators use homomorphic encryption. PRISM does not require the service provider to implement it. Rather, PRISM applies homomorphic-like (but much lighter) encryption on the user's side. PRISM can be seen as a combination of client-side encryption, which has been adopted in cloud storage services~\cite{cryptomator,wilson2014share}, and homomorphic encryption.

\textbf{User-side Realization.} Users are dissatisfied with services. Since the service is not tailor-made for a user, it is natural for dissatisfaction to arise. However, even if users are dissatisfied, they often do not have the means to resolve their dissatisfaction. The user cannot alter the source code of the service, nor can they force the service to change. In this case, the user has no choice but to remain dissatisfied or quit the service. User-side realization provides a solution to this problem. User-side realization~\cite{sato2022private,sato2022clear} provides a general algorithm to deal with common problems on the user's side. Many user-side algorithms for various problems have been proposed. Consul~\cite{sato2022towards} turns unfair recommender systems into fair ones on the user's side. Tiara~\cite{sato2022retrieving} realizes a customed search engine the results of which are tailored to the user's preference on the user's side. WebShop~\cite{yao2022webshop} enables automated shopping in ordinary e-commerce sites on the user's side by using an agent driven by a large language model. WebArena~\cite{zhou2023webarena} is a general environment to test agents realizing rich functionalities on the user's side. EasyMark~\cite{sato2023embarassingly} realizes large language models with text watermarks on the user's side. Overall, there are many works on user-side realization, but most of them are on recommender systems and search engines. Our work is the first to protect the privacy of texts on the user's side.

\section{Conclusion}

We proposed a novel problem setting of user-side privacy protection for machine translation systems. We proposed two methods, PRISM-R and PRISM*, to turn external machine translation systems into privacy-preserving ones on the user's side. We showed that PRISM-R is differential private and PRISM* striked a better trade-off between privacy and accuracy. We also proposed an evaluation protocol for user-side privacy protection for machine translation systems, which is valuable for facilitating future research in this area.

\chapter{Conclusion and Future Directions}

\section{Conclusion}

In this thesis, we tackled user-side realization problems, which have not been well explored in the literature. To the best of our knowledge, this is the first systematic document for user-side realization.

We first introduced the requirements, challenges, and approaches of user-side realization in Chapter 1. These materials are a good introduction to this field. Especially, the taxonomy we introduced, i.e., the wrapper method and the reverse method, provides a useful guideline for building new user-side algorithms.

We proposed five user-side algorithms and one application in Chapters 2 to 7. These methods are broad from recommender systems to translation systems in terms of the service domain, and image-to-image search to privacy protection in terms of the realized functionalities. These results demonstrate the broad scope of the concept of user-side realization. One of the prominent characteristics of these projects is that many in-the-wild experiments were conducted, including the movie rating service (IMDb), social networking service (Twitter or X), photo sharing service (Flickr), and chatbots (ChatGPT). These services are real ones in operation on the Internet and have billions of users. We also built a system demonstration using CLEAR. These results show the concept does work in the real world.

\section{Future Directions}

There are several future directions for user-side realization.

First, many more domains and features are waiting for user-side realization. Although we propose many user-side algorithms for broad domains and features, I hardly claim the scope of user-side realization has been covered. We mainly focused on web services in this thesis, but user-side realization can shine also in offline services and embedded systems. Low-layer systems such as operating systems and hardware drivers may be interesting subjects of user-side algorithms. User-side realization is not limited to software and could be applied to real-world services such as procedures at city halls, banks, and stations. Connecting knowledge of user-side realization for software, web services, and real-world services and establishing an integrated theory is a challenging but interesting direction.

Second, more open resources, including datasets and libraries, should be developed. Although we open-sourced algorithms, demonstrations, and evaluation pipelines, they are scattered and there have not been systematic libraries or datasets so far. Establishing standard datasets and libraries is important for facilitating further research in this field. We feel that user-side realization algorithms for recommender systems and search engines are mature and are ready for building standard libraries and datasets.

Third, user-side realization should be deployed in the real world more broadly. As we confirmed the effectiveness of the proposed methods in the wild, they are ready to deploy in the real world. Packaging them as browser add-ons will make these algorithms reach other end-users and broaden the scope of users. Building useful libraries is also important for this goal. We hope more and more researchers, engineers, and end-users are involved in the projects and realize attractive applications.

We hope everyone enjoys user-side realization and receives their own satisfiable services in the distant future. We believe this thesis becomes a milestone in user-side realization and facilitates further research in this field.


\bibliography{citation.bib}

\end{document}